\documentclass[journal]{IEEEtran}

\usepackage[]{authblk}
\usepackage{times}
\usepackage{epsf}
\usepackage{psfrag}
\usepackage{epsfig}
\usepackage{amsmath}
\usepackage{amsfonts}
\usepackage{amssymb}
\usepackage{amstext}
\usepackage{latexsym}
\usepackage{ifthen}
\usepackage{multirow}
\usepackage{subfigure}

\newcommand{\var}[1]{\mbox{Var} \left[ #1 \right]}
\newcommand{\cov}[1]{\mbox{Cov} \left[ #1 \right]}
\newcommand{\be}{\begin{equation}}
\newcommand{\ee}{\end{equation}}
\newcommand{\bear}{\begin{eqnarray}}
\newcommand{\eear}{\end{eqnarray}}
\newcommand{\bears}{\begin{eqnarray*}}
\newcommand{\eears}{\end{eqnarray*}}
\newcommand{\bi}{\begin{itemize}}
\newcommand{\ei}{\end{itemize}}
\newcommand{\ben}{\begin{enumerate}}
\newcommand{\een}{\end{enumerate}}

\newcommand{\beq}{\begin{equation}}
\newcommand{\eeq}{\end{equation}}

\newcommand{\lp}{ \left(}
\newcommand{\rp}{ \right)}

\newtheorem{theorem}{Theorem}[section]

\newtheorem{defn}[theorem]{Definition} 
\newtheorem{lemma}[theorem]{Lemma}
\newtheorem{corollary}[theorem]{Corollary}

\newcommand{\abf}{\mbox{${\bf a }$} }
\newcommand{\bbf}{\mbox{${\bf b }$} }

\newcommand{\gbf}{\mbox{${\bf g }$} }

\newcommand{\ubf}{\mbox{${\bf u }$} }

\newcommand{\xbf}{\mbox{${\bf x }$} }
\newcommand{\ybf}{\mbox{${\bf y }$} }
\newcommand{\zbf}{\mbox{${\bf z }$} }

\newcommand{\Abf}{\mbox{${\bf A }$} }
\newcommand{\Bbf}{\mbox{${\bf B }$} }

\newcommand{\Fbf}{\mbox{${\bf F }$} }
\newcommand{\Gbf}{\mbox{${\bf G }$} }
\newcommand{\Hbf}{\mbox{${\bf H }$} }
\newcommand{\Ibf}{\mbox{${\bf I }$} }

\newcommand{\Qbf}{\mbox{${\bf Q }$} }

\newcommand{\Sbf}{\mbox{${\bf S }$} }
\newcommand{\Ubf}{\mbox{${\bf U }$} }
\newcommand{\Vbf}{\mbox{${\bf V }$} }
\newcommand{\Xbf}{\mbox{${\bf X }$} }
\newcommand{\Ybf}{\mbox{${\bf Y }$} }
\newcommand{\Zbf}{\mbox{${\bf Z }$} }

\newcommand{\Zerobf}{\mbox{${\large{\bf 0 }}$} }

\newcommand{\FF}{\mbox{$\mathbb{F}$} }

\newcommand{\CC}{\mbox{$\mathbb{C}$} }
\newcommand{\ZZ}{\mbox{$\mathbb{Z}$} }
\newcommand{\EE}{\mbox{$\mathbb{E}$} }

\newcommand{\Prob}{\mbox{${\mathbb P}$} }

\newcommand{\prob}[1]{\Prob \left\{ #1 \right\}}

\newenvironment{proofOutline}{
    \noindent
    {\bf Proof outline:}
}{
    \hfill$\blacksquare$
}

\newcommand{\SNR}{{\sf SNR}}

\begin{document}

\title{Wireless Network Information Flow:\\ A Deterministic Approach}


\author{A.~Salman~Avestimehr,~\IEEEmembership{Member,~IEEE,}
        Suhas~N.~Diggavi,~\IEEEmembership{Member,~IEEE,}
        and~David~N.~C. Tse,~\IEEEmembership{Fellow,~IEEE}
\thanks{A. S. Avestimehr is with the School of Electrical and Computer Engineering, Cornell University, Ithaca, USA. Email: {\sffamily avestimehr@ece.cornell.edu}.}
\thanks{S. N. Diggavi is with the Department of Electrical Engineering, UCLA, Los Angeles, USA. Email: {\sffamily suhas@ee.ucla.edu}.}
\thanks{D. N. C. Tse is with the Department of Electrical Engineering and Computer Sciences, UC Berkeley, Berkeley, USA. Email: {\sffamily dtse@eecs.berkeley.edu}.}
\thanks{The research of D. Tse and A. Avestimehr were supported in part by
the National Science Foundation under grants 0326503, 0722032 and 0830796, and by a gift from Qualcomm Inc. The
research of S. Diggavi was supported in part by the Swiss National
Science Foundation NCCR-MICS center.}
\thanks{Manuscript received July 28, 2009; revised June 1, 2010 and August 5, 2010. Date of current version September 16, 2010.}
\thanks{Communicated by M. Franceschetti, Associate Editor for Communication Networks. }}



\maketitle

\begin{abstract}
In a wireless network with a single source and a single destination
and an arbitrary number of relay nodes, what is the maximum rate of
information flow achievable? We make progress on this long standing
problem through a two-step approach. First we propose a deterministic
channel model which captures the key wireless properties of signal
strength, broadcast and superposition. We obtain an exact
characterization of the capacity of a network with nodes connected by
such deterministic channels. This result is a natural generalization
of the celebrated max-flow min-cut theorem for wired networks. Second, we use the
insights obtained from the deterministic analysis to design a new {\em
  quantize-map-and-forward} scheme for Gaussian networks. In this scheme, each relay quantizes the received signal at the
noise level and maps it to a random Gaussian codeword for forwarding, and the final destination decodes the source's message based on  the received signal. We show
that, in contrast to existing schemes, this scheme can achieve the
cut-set upper bound to within a gap which is independent of the
channel parameters. In the case of the relay channel with a single
relay as well as the two-relay Gaussian diamond network, the gap is $1$
bit/s/Hz. Moreover, the scheme is universal in the sense that the
relays need no knowledge of the values of the channel parameters to (approximately) achieve the rate supportable by the network.  We also present
extensions of the results to multicast networks, half-duplex networks
and ergodic networks.
\end{abstract}

\section{Introduction}
\label{sec:intro}

Two main distinguishing features of wireless communication are:
\begin{itemize}

\item {\em broadcast}:
  wireless users communicate over the air and signals from any one
  transmitter are heard by multiple nodes with possibly different
  signal strengths.

\item {\em superposition}: a wireless node receives
  signals from multiple simultaneously transmitting nodes, with the
  received signals all superimposed on top of each other.
\end{itemize}

Because of these effects, links in a wireless network are never
isolated but instead interact in seemingly complex ways. On the one
hand,  this facilitates the spread of information among users in a
network; on the other hand it can be harmful by creating signal
interference among users. This is in direct contrast to wired
networks, where transmitter-receiver pairs can be thought of as
isolated point-to-point links. Starting from the max-flow-min-cut theorem of Ford-Fulkerson \cite{FF56}, there has been significant progress in understanding
network flow over wired networks. Much less, however, is known for wireless networks.

The linear additive Gaussian channel model is a commonly used model to
capture signal interactions in wireless channels. Over the past couple
of decades, capacity study of Gaussian networks has been an active
area of research. However, due to the complexity of the Gaussian
model, except for the simplest networks such as the one-to-many
Gaussian broadcast channel and the many-to-one Gaussian multiple
access channel, the capacity of most Gaussian networks is still
unknown. For example, even the capacity of a Gaussian single-relay
network, in which a point to point communication is assisted by one
relay, has been open for more than 30 years.  In order to make
progress on this problem, we take a two-step approach. We first focus
on the signal interaction in wireless networks rather than on the
noise. We present a new deterministic channel model which is
analytically simpler than the Gaussian model but yet still captures
three key features of wireless communication: channel strength, broadcast. and
superposition. A motivation to study such a model is that in contrast
to point-to-point channels where noise is the only source of
uncertainty, networks often operate in the {\em interference-limited}
regime where the noise power is small compared to signal
powers. Therefore, for a first level of understanding, our focus is on
such signal interactions rather than the background noise.  Like the
Gaussian model, our deterministic model is linear, but unlike the
Gaussian model, operations are on a finite-field. The
  simplicity of {\em scalar} finite-field channel models has also been noted
  in \cite{SevenPeoplePaper}. We provide
a complete characterization of the capacity of a network of nodes
connected by such deterministic channels. The first result is a natural
generalization of the max-flow min-cut theorem for wired networks.

The second step is to utilize the insights from the deterministic
analysis to find ``approximately optimal'' communication schemes for  Gaussian
relay networks. The analysis for  deterministic networks not only
gives us insights for potentially successful coding schemes for the
Gaussian case, but also gives tools for the proof techniques used.
We show that in Gaussian networks, an {\em approximate}
max-flow min-cut result can be shown, where the approximation is
within an additive constant which is universal over the values of the channel parameters
 (but could depend on the number of nodes in
the network). For example, the additive gap for both the single-relay
network and for the two-relay diamond network is $1$ bit/s/Hz. This is
the first result we are aware of that provides such performance
guarantees on relaying schemes. To highlight the strength of this
result, we demonstrate that none of the existing strategies in the
literature, like amplify-and-forward, decode-and-forward and Gaussian
compress-and-forward, yield such a universal approximation for
arbitrary networks. Instead, a scheme, which we term {\em
  quantize-map-and-forward}, provides such a universal approximation.

In this paper we focus on unicast and multicast communication
scenarios. In the unicast scenario, one source wants to communicate
to a single destination. In the multicast scenario source wants
to transmit the {\em same} message to multiple destinations. Since in
these scenarios, all destination nodes are interested in the same
message, there is effectively only one information stream in the network. Due to
the broadcast nature of the wireless medium, multiple copies of a
transmitted signal {\em are}  received at different relays and
superimposed with other received signals.  However, since they are
all a function of the same message, they are not considered as
interference. In fact, the quantize-map-and-forward strategy exploits
this broadcast nature by forwarding all the available information
received at the various relays to the final destination. This is in
contrast to more classical approaches of dealing with simultaneous
transmissions by either avoiding them through transmit scheduling or
treating signals from all nodes other than the intended transmitter
as interference adding to the noise floor. These approaches attempt
to convert the wireless network into a wired network but  are
strictly sub-optimal.

\subsection{Related Work}

In the literature, there has been extensive research over the last
three decades   to characterize the capacity of relay networks. The
single-relay channel was first introduced in 1971 by van der Meulen
\cite{Meulen} and the most general strategies for this network were
developed by Cover and El Gamal \cite{coverElgamal}. There has also
been a significant effort to generalize these ideas to arbitrary
multi-relay networks with simple channel models. An early
attempt was done in the Ph.D. thesis of Aref \cite{ArefThesis} where
a max-flow min-cut result was established to characterize the
unicast capacity of a deterministic broadcast relay network  \emph{without superposition}. This was an early precursor to network coding which
established the multicast capacity of wired networks, a
deterministic capacitated graph without  broadcast or superposition
\cite{ACLY00,LinNetCod,KoetterMedard}. These two ideas were combined
in \cite{RK06}, which established a max-flow min-cut
characterization for multicast flows for ``Aref networks''. However,
such complete characterizations are not known for arbitrary (even
deterministic) networks with both broadcast and superposition. One
notable exception is the work \cite{GuptaBhadraShakkottaiIntNet}
which takes a scalar deterministic linear finite-field model and
uses probabilistic erasures to model channel failures. For this
model using results of erasure broadcast networks
\cite{HassibiErasureNet}, they established an asymptotic result on
the unicast capacity as the field size grows. However, in all these
works there is no connection between the proposed channel
model and the physical wireless channel.

There has also been a rich body of literature in directly tackling
the noisy relay network capacity problem. In \cite{ScheinThesis} the
``diamond'' network of parallel relay channels with no direct link
between the source and the destination was examined.  Xie and Kumar
generalized the decode-forward encoding scheme for a network of
multiple relays \cite{XieKumarNetInfTheory}. Kramer et al.
\cite{KramerGastparGuptaRelay} also generalized the compress-forward
strategy to networks with a single layer of relay nodes. Though
there have been many interesting and important ideas developed in
these papers, the capacity characterization of Gaussian relay
networks is still unresolved. In fact even a performance guarantee,
such as establishing how far these schemes are from an upper bound
is unknown. In fact, as we will see in Section \ref{sec:motivation}, these
strategies do not yield an approximation guarantee for general
networks.

There are subtle but critical differences between the
quantize-map-forward strategy, proposed in this paper, with the natural extension of
compress-forward to networks for the following reasons. The
compress-forward scheme proposed in \cite{coverElgamal}, quantized the
received signal and then mapped the digital bin index onto the transmit
sequence. This means that we need to make choices on the binning rates
at each relay node. However, the quantize-map-forward scheme proposed
in this paper directly maps the the quantized sequence to the transmit
sequence, and therefore does not make such choices on the binning
rates. In fact this gives the scheme a ``universality'' property, which
allows the same relay operation to work for multiple destinations
(multicast) and network situations (compound networks); a property
that could fail to hold if specific choices of binning rates were
made.  Moreover, our scheme unlike the classical compress-forward scheme, does not require the
quantized values at the relays to be reconstructed at the destination, while it is
attempting to decode the transmitted message. These are the essential differences
between our scheme and the  traditional compress-forward, or the natural network generalization
of it.

Our results are connected to the concept of network coding in several ways.  The most direct connection is that our results on the multicast capacity of deterministic networks are direct generalizations of network coding results  \cite{ACLY00,LinNetCod,KoetterMedard,YeungMono,ChristinaMono} as well as Aref networks \cite{ArefThesis,RK06}. The coding techniques for the deterministic case are inspired by and generalize the random network coding technique of  \cite{ACLY00} and the linear coding technique of \cite{LinNetCod,KoetterMedard,Ho07}. The quantize-map-and-forward technique proposed in this paper for the Gaussian wireless networks uses the insights from the deterministic framework and is philosophically the network coding technique generalized to noisy wireless networks.

\subsection{Outline of the paper}

We first develop an analytically simple linear finite-field model and motivate it by
connecting it to the Gaussian model in the context of several simple
multiuser networks. We also discuss its limitations. This is done in Section \ref{sec:model}. This
model also suggests achievable strategies to explore in
Gaussian relay networks, as done in Section \ref{sec:motivation},
where we illustrate the deterministic approach on several
progressively more complex example networks. The deterministic model
also makes clear that several well-known strategies can be in fact
arbitrarily far away from optimality in these example networks.

Section \ref{sec:mainResults} summarizes the main results of the
paper.  Section \ref{sec:detCapacity} focuses on the capacity
analysis of networks with nodes connected by deterministic channels.
We examine arbitrary deterministic channel model (not
necessarily linear nor finite-field) and establish an achievable
rate for an arbitrary network. For the special
case of linear finite-field deterministic models, this
achievable rate matches the cut-set bound, therefore
exact characterization is possible. The achievable strategy
involves each node randomly mapping the received signal to a
transmitted signal, and the final destination solving for the
information bits from all the received equations.

The examination of the deterministic relay network motivates the
introduction of a simple {\em quantize-map-and-forward} strategy for
general Gaussian relay networks. In this scheme each relay first
quantizes the received signal at the noise level, then randomly maps
it to a Gaussian codeword and transmits it\footnote{This is distinct
  from the compress and forward scheme studied in \cite{coverElgamal}
  where the quantized value is to be reconstructed at the
  destination. Our scheme does not require the quantized values to be
  reconstructed, but just the source codeword to be decoded.}.  In
Section \ref{sec:GaussCapacity} we use the insights of the
deterministic result to demonstrate that we can achieve a rate that is
guaranteed to be within a constant gap from the cut-set upper bound on
capacity. As a byproduct, we show in Section \ref{sec:connections}
that a deterministic model formed by quantizing the received signals
at noise level at all nodes and then removing the noise is within a
constant gap to the capacity of the Gaussian relay network.

In Section \ref{sec:extensions}, we show that the quantize-map-and-forward scheme has the
desirable property that the relay nodes do not need the knowledge of
the channel gains. As long as the network can support a
given rate, we can achieve it without the relays' knowledge of the
channel gains. In Section \ref{sec:extensions}, we also establish
several other extensions to our results, such as relay networks with
half-duplex constraints, and relay networks with fading or frequency
selective channels.

\section{Deterministic modeling of wireless channel}
\label{sec:model}

The goal of this section is to introduce the linear deterministic
model and illustrate how we can deterministically model three key
features of a wireless channel.

\subsection{Modeling signal strength}
Consider the \textit{real} scalar Gaussian model for a point-to-point
link, \beq \label{eq:p2p} y=hx+z \eeq where $z \sim \mathcal{N}
(0,1)$. There is also an average power constraint $E[|x|^2]\leq 1$ at
the transmitter. The transmit power and noise power are both
normalized to be equal to 1 and the channel gain $h$ is related to the signal-to-noise ratio ($\SNR$) by \beq
|h|=\sqrt{\SNR}. \eeq It is well known that the capacity of this
point-to-point channel is
\begin{eqnarray} 
\label{eq:p2p_Gaussian_cap} 
C_{\text{AWGN}}& =&  \frac{1}{2} \log \lp 1+ \SNR \rp. 
\end{eqnarray} 
To get an intuitive understanding of this capacity formula let us
write the received signal in Equation (\ref{eq:p2p}), $y$, in terms of
the binary expansions of $x$ and $z$. For simplicity assuming $h$, $x$
and $z$ are positive real numbers and $x$ has a peak power constraint
of 1, we have \beq y= 2^{\frac{1}{2}\log \SNR}
\sum_{i=1}^{\infty}x(i)2^{-i} +
\sum_{i=-\infty}^{\infty}z(i)2^{-i}.\eeq To simplify the effect of
background noise assume it has a peak power equal to 1. Then we can
write
\begin{eqnarray} y&=& 2^{\frac{1}{2}\log \SNR}
  \sum_{i=1}^{\infty}x(i)2^{-i} +
  \sum_{i=1}^{\infty}z(i)2^{-i} \end{eqnarray}
or,
\beq y \approx 2^{n} \sum_{i=1}^{n}x(i)2^{-i} + \sum_{i=1}^{\infty}\lp x(i+n)+z(i) \rp 2^{-i}
\eeq
where $n= \lceil \frac{1}{2} \log \SNR \rceil^+$. Therefore if we just
ignore the 1 bit of the carry-over from the second summation
($\sum_{i=1}^{\infty}\lp x(i+n)+z(i) \rp 2^{-i}$) to the first
summation ($2^{n} \sum_{i=1}^{n}x(i)2^{-i}$) we can approximate
a point-to-point Gaussian channel as a pipe that truncates the
transmitted signal and only passes the bits that are above the noise
level. Therefore think of transmitted signal $x$ as a sequence of bits
at different signal levels, with the highest signal level in $x$ being
the most significant bit  and the lowest level being the least
significant bit. In this simplified model the receiver can see
the $n$ most significant bits of $x$ without any noise and the rest
are not seen at all. There is a correspondence between $n$ and
$\SNR$ in dB scale,
\beq
n \leftrightarrow \lceil \frac{1}{2} \log  \SNR \rceil^+.
\eeq
This simplified model, shown in Figure \ref{fig:p2p_det}, is deterministic. Each circle in the figure represents a signal level which holds a binary digit for transmission. The most significant $n$ bits are received at the destination while less significant bits are not.

These signal levels can potentially be created using a multi-level
lattice code in the AWGN channel \cite{ForneyMultiLevel}. Then the
first $n$ levels in the deterministic model represent those levels (in
the lattice chain) that are above noise level, and the remaining are
the ones that are below noise level. We can algebraically write this input-output relationship
by shifting $\mathbf{x}$ down by $q-n$ elements
\beq \mathbf{y}={\bf S^{q-n}}\mathbf{x} \eeq
where $\mathbf{x}$ and $\mathbf{y}$ are binary vectors of length $q$
denoting transmit and received signals respectively and $\Sbf$ is the
$q \times q$ shift matrix, \beq \Sbf=\left(
         \begin{array}{ccccc}
           0 & 0 & 0 & \cdots & 0 \\
           1 & 0 & 0 & \cdots & 0 \\
           0 & 1 & 0 & \cdots & 0 \\
           \vdots & \ddots & \ddots & \ddots & \vdots \\
           0 & \cdots & 0 & 1 & 0 \\
         \end{array}
       \right).
 \eeq

\begin{figure}
     \centering

       \includegraphics{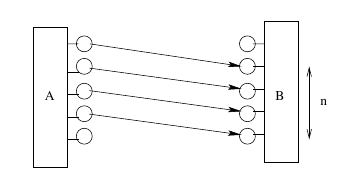}
 \caption{Pictorial representation of the deterministic model for point-to-point channel. \label{fig:p2p_det}}
\end{figure}

The capacity of this deterministic point-to-point channel is
$n$, where $n= \lceil \frac{1}{2} \log \SNR \rceil^+$. This capacity is within $\frac{1}{2}$-bit
approximation of the capacity of the AWGN channel. In the
case of complex Gaussian channel we set $n= \lceil \log \SNR \rceil^+$
and we get an approximation within 1-bit of the capacity.

\subsection{Modeling broadcast}

Based on the intuition obtained so far, it is straightforward to think
of a deterministic model for a broadcast scenario. Consider the real
scalar Gaussian broadcast channel (BC). Assume there are only two
receivers. The received $\SNR$ at receiver $i$ is denoted by $\SNR_i$
for $i=1,2$ ($\SNR_2 \le
\SNR_1$). Consider the binary expansion of the transmitted signal,
$x$. Then we can deterministically model the Gaussian broadcast
channel as the following:
\begin{itemize}
\item Receiver 2 (weak user) receives only the most significant $n_2$ bits in the
  binary expansion of $x$. Those bits are the ones that arrive above
  the noise level.
\item Receiver 1 (strong user) receives the most significant $n_1$ ($n_1>n_2$)
  bits in the binary expansion of $x$. Clearly these bits contain what
  receiver 2 gets.
\end{itemize}

The deterministic model makes explicit the functioning of superposition coding and successive interference cancellation decoding
in the Gaussian broadcast channel. The most significant $n_2$ levels in
the deterministic model represent the cloud center that is decoded by
both users, and the remaining $n_1-n_2$ levels represent the cloud
detail that is decoded only by the strong user (after decoding the
cloud center and canceling it from the received signal).

Pictorially the deterministic model is shown in Figure \ref{fig:bc} (a). In this particular example
$n_1=5$ and $n_2=2$, therefore both users receive the two most
significant bits of the transmitted signal. However user 1 (strong
user) receives three additional bits from the next three signal levels
of the transmitted signal. There is also the same correspondence
between $n$ and channel gains in dB: 
\beq 
\label{eq:chGainRelBC} 
n_i\leftrightarrow \lceil \frac{1}{2}\log \SNR_i \rceil^+, \quad i=1,2. 
\eeq

To analytically demonstrate how closely we are modeling the Gaussian
BC channel, the capacity region of the Gaussian BC channel and
the deterministic BC channel are shown in Figure \ref{fig:bc} (b). As it
is seen their capacity regions are very close to each other. In fact
it is easy to verify that for all SNR's these regions are always
within one bit per user of each other, that is, if $(R_1,R_2)$ is in the capacity region of the deterministic BC then
there is a rate pair within one bit component-wise of $(R_1,R_2)$ that is in
the capacity region of the Gaussian BC. However, this is only the
worst-case gap and in the typical case where $\SNR_1$ and $\SNR_2$ are
very different, the gap is much smaller than one bit.

\begin{figure*}
     \centering
     \subfigure[Pictorial representation of the deterministic model for Gaussian BC]{
       \includegraphics{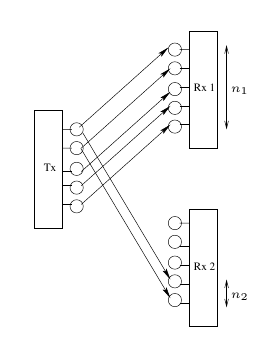} }
    \hspace{.5in}
    \subfigure[Capacity region of Gaussian BC (solid line) and deterministic BC (dashed line).]{
       \includegraphics{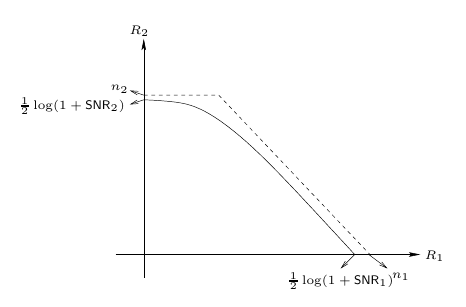}}
     \caption{Pictorial representation of the deterministic model for
       Gaussian BC is shown in (a). Capacity region of Gaussian and
       deterministic BC are shown in (b). \label{fig:bc}}
\end{figure*}

\subsection{Modeling superposition}

Consider a superposition scenario in which two users are
simultaneously transmitting to a node. In the Gaussian model the
received signal can be written as
\beq y=h_1 x_1+h_2 x_2+z. \eeq
To intuitively see what happens in superposition in the Gaussian
model, we again write the received signal, $y$, in terms of the binary
expansions of $x_1$, $x_2$ and $z$. Assume $x_1$, $x_2$ and $z$ are
all positive real numbers smaller than one, and also the channel gains are \beq
h_i=\sqrt{\SNR_i}, \quad i=1,2. \eeq Without loss of generality assume
$\SNR_2 < \SNR_1$. Then  
{\small \beq \nonumber y= 2^{\frac{1}{2}\log \SNR_1}
\sum_{i=1}^{\infty}x_1(i)2^{-i} + 2^{\frac{1}{2}\log \SNR_2}
\sum_{i=1}^{\infty}x_2(i)2^{-i} +
\sum_{i=-\infty}^{\infty}z(i)2^{-i}. \eeq } To simplify the effect of
background noise assume it has a peak power equal to 1. Then we can
write 
{\small \beq \nonumber y= 2^{\frac{1}{2}\log \SNR_1}
\sum_{i=1}^{\infty}x_1(i)2^{-i} + 2^{\frac{1}{2}\log \SNR_2}
\sum_{i=1}^{\infty}x_2(i)2^{-i} + \sum_{i=1}^{\infty}z(i)2^{-i}\eeq}
or,
{\small \begin{eqnarray*}
 \nonumber y& \approx & 2^{n_1} \sum_{i=1}^{n_1-n_2}x_1(i)2^{-i} + 2^{n_2} \sum_{i=1}^{n_2}\lp x_1(i+n_1-n_2)+x_2(i)\rp 2^{-i} \\ &&+ \sum_{i=1}^{\infty}\lp x_1(i+n_1)+x_2(i+n_2)+ z(i) \rp 2^{-i}
\end{eqnarray*}}
where $n_i= \lceil \frac{1}{2} \log \SNR_i \rceil ^+$ for
$i=1,2$. Therefore based on the intuition obtained from the
point-to-point and broadcast AWGN channels, we can approximately model
this as the following:
\begin{itemize}
\item That part of $x_1$ that is above $\SNR_2$ ($x_1(i)$, $1 \leq i
  \leq n_1-n_2$) is received clearly without any contribution from
  $x_2$.
\item The remaining part of $x_1$ that is above noise level ($x_1(i)$,
  $n_1-n_2 < i \leq n_1$) and that part of $x_2$ that is above noise
  level ($x_1(i)$, $1 \leq i \leq n_2$) are superposed on each other and
  are received without any noise.
\item Those parts of $x_1$ and $x_2$ that are below noise level are
  truncated and not received at all.
\end{itemize}

The key point is how to model the superposition of the bits that
are received at the same signal level. In our deterministic model we
ignore the carry-overs of the real addition and we model the
superposition by the modulo 2 sum of the bits that are arrived at the
same signal level. Pictorially the deterministic model is shown in Figure \ref{fig:mac} (a). Analogous to the
deterministic model for the point-to-point channel, as seen in Figure
\ref{fig:AlgRepShiftMat}, we can write
\beq
\mathbf{y}={\bf S^{q-n_1}}\mathbf{x_1} \oplus {\bf
  S^{q-n_2}}\mathbf{x_2}
\eeq
where the summation is in $\FF_2$
(modulo 2). Here $\mathbf{x_i}$ ($i=1,2$) and $\mathbf{y}$ are binary
vectors of length $q$ denoting transmitted and received signals
respectively and $\Sbf$ is a $q \times q$ shift matrix.  The relationship between $n_i$'s and the channel gains is the same as in Equation (\ref{eq:chGainRelBC}).

\begin{figure}[h]
\begin{center}
\includegraphics{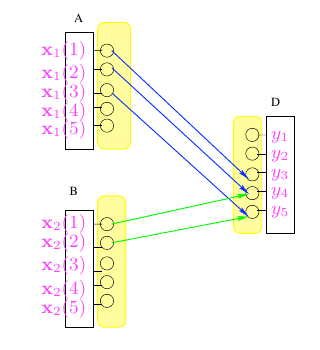}
\end{center}
\caption{Algebraic representation of shift matrix deterministic model.\label{fig:AlgRepShiftMat}}
\end{figure}

Compared to the point-to-point case we now have interaction
between the bits that are received at the same signal level at the
receiver. We limit the receiver to observe only the modulo
2 summation of those bits that arrive at the same signal level. This way of modeling signal interaction has two advantages over the simplistic collision model. First, if two bits arrive simultaneously at the same signal level, they are not both dropped and the receiver gets their modulo 2 summation. Second, unlike in the collision model
where the entire packet is lost when there is collision, the most
significant bits of the stronger user remain intact. This is
reminiscent of the familiar {\em capture} phenomenon in CDMA
systems: the strongest user can be heard even when multiple users
simultaneously transmit.

\begin{figure*}
     \centering
     \subfigure[Pictorial representation of the deterministic MAC.]{
       \includegraphics{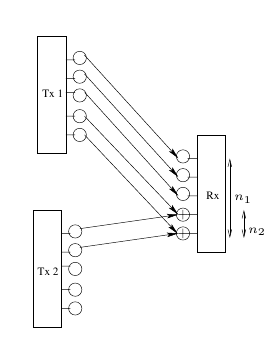}}
    \hspace{.5in}
    \subfigure[Capacity region of Gaussian MAC (solid line) and deterministic MAC (dashed line).]{
       \includegraphics{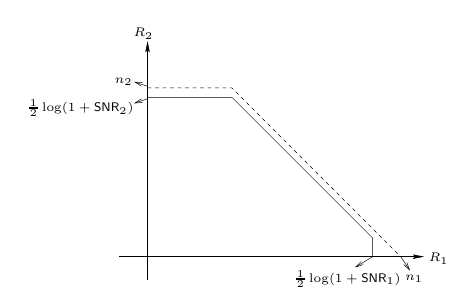}}
 \caption{Pictorial representation of the deterministic MAC is shown in (a). Capacity region of Gaussian and deterministic MACs are shown in (b). \label{fig:mac}}
\end{figure*}

Now we can apply this model to the Gaussian MAC,
in which \beq y=h_1 x_1+h_2 x_2+z \eeq where $z \sim \mathcal{CN}
(0,1)$. There is also an average power constraint equal to 1 at both
transmitters. A natural question is how close is the capacity region
of the deterministic model to that of the actual Gaussian
model. Assume $\SNR_2 < \SNR_1$.
The capacity region of this channel is known to be the set of non-negative pairs $(R_1,R_2)$ satisfying \begin{eqnarray}R_i & \leq & \log (1+\SNR_i), \quad i=1,2 \\
  R_1+R_2 & \leq & \log (1+\SNR_1+\SNR_2). \end{eqnarray} This region
is plotted with solid line in Figure \ref{fig:mac} (b).

It is easy to verify that the capacity region of the deterministic MAC
is the set of non-negative pairs $(R_1,R_2)$ satisfying
\begin{eqnarray}R_2 & \leq & n_2  \\
R_1+R_2 & \leq & n_1
\end{eqnarray}
where $n_i= \lceil \log \SNR_i \rceil^+ $ for $i=1,2$. This region is plotted with
dashed line in Figure \ref{fig:mac} (b). In this deterministic model
the ``carry-over'' from one level to the next that would happen with
real addition is ignored. However as we notice still the capacity
region is very close to the capacity region of the Gaussian model. In
fact it is easy to verify that they are within one bit per user of
each other. The intuitive explanation for this is that in real
addition once two bounded signals are added together the magnitude
can become as large as twice the larger of the two signals. Therefore the number of bits in the sum
is increased by at most one bit. On the other hand in finite-field
addition there is no magnitude associated with signals and the
summation is still in the same field as the individual
signals. So the gap between Gaussian and deterministic model for two
user MAC is intuitively this one bit of cardinality increase. Similar
to the broadcast example, this is only the worst case gap and when the
channel gains are different it is much smaller than one bit.

Now we define the linear finite-field deterministic model for the relay network.

\subsection{Linear finite-field deterministic model}
\label{subsec:LFFDetModel}

The relay network is defined using a set of vertices $\mathcal{V}$.
The communication
link from node $i$ to node $j$ has a non-negative integer
gain  $n_{ij}$ associated
with it. This number models the channel gain in the corresponding
Gaussian setting.  At each time $t$, node $i$ transmits a vector
${\bf x}_i[t] \in \FF_{p}^q$ and receives a vector ${\bf
  y}_i[t] \in \FF_{p}^q$ where $q=\max_{i,j}(n_{(ij)})$ and $p$ is a
positive integer indicating the field size. The received signal at
each node is a deterministic function of the transmitted signals at
the other nodes, with the following input-output relation: if the
nodes in the network transmit ${\bf x}_1[t], {\bf x}_2[t] , \ldots
{\bf x}_N[t]$ then the received signal at node j, $1 \leq j \leq N$
is:
\begin{equation}
\label{eq:LinDetModel}
{\bf y}_{j}[t]=\sum {\bf S}^{q-n_{ij}}{\bf x}_{i}[t]
\end{equation}
where the summations and the multiplications are in $\FF_{p}$. Throughout this
paper the field size, $p$, is assumed to be $2$, unless it is stated
otherwise.

\subsection{Limitation: Modeling MIMO}

The examples in the previous subsections may give the impression that
the capacity of any Gaussian channel is within a constant gap to that
of the corresponding linear deterministic model. The following example
shows that is not the case.

\begin{figure}
     \centering \subfigure[]{
       \scalebox{0.6}{\includegraphics{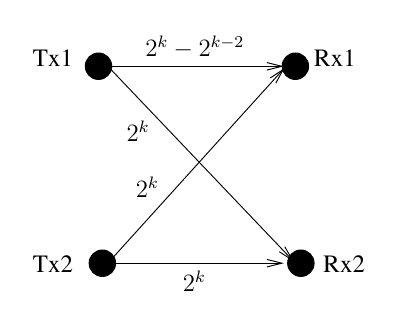}}
}
     \hspace{.1in}
     \subfigure[]{
       \scalebox{0.6}{\includegraphics{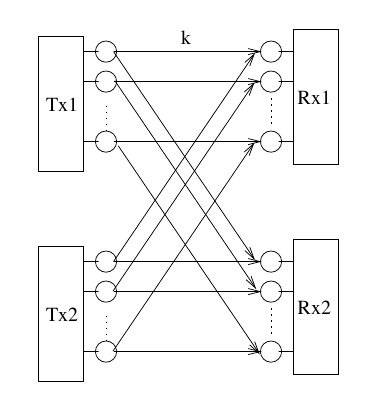}}
}
\caption{An example of a $2 \times 2$ Gaussian MIMO channel is shown in $(a)$.  The corresponding linear finite-field deterministic MIMO channel is shown in $(b)$.\label{fig:exMIMO}}
\end{figure}

Consider a $2 \times 2$ MIMO real Gaussian channel with channel gain values as shown in Figure \ref{fig:exMIMO} (a), where $k$ is an integer larger than 2. The channel matrix is
\beq {\bf H}=2^k \left(\begin{array}{cc}\frac{3}{4} & 1 \\1 & 1\end{array}\right).   \eeq
 The channel gain parameters of the corresponding linear finite-field deterministic model are:
\begin{eqnarray}
\nonumber n_{11} & = & \lceil \frac{1}{2}\log_2 |h_{11}|^2 \rceil^+ = \lceil \log_2 (2^k-2^{k-2}) \rceil^+=k\\
n_{12}&=& n_{21}=n_{22}= \lceil \log_2 2^k \rceil^+ =k
\end{eqnarray}
Now let us compare the capacity of the MIMO channel under these two models for large values of $k$. For the Gaussian model, both singular values of ${\bf H}$ are of the order of $2^k$. Hence, the capacity of the real Gaussian MIMO channel is of the order of
$$ 2 \times \frac{1}{2} \log (1+ |2^k|^2) \approx 2k.$$
However the capacity of the corresponding linear finite-field deterministic MIMO is simply
\begin{eqnarray}
\label{eq:LFFMIMOCap}C_{\text{LFF}}&=& \text{rank}  \left(\begin{array}{cc}I_k & I_k \\I_k & I_k\end{array}\right)  =k. 
\end{eqnarray}
Hence the gap between the two capacities goes to infinity as $k$ increases.

Even though the linear deterministic channel model does not
approximate the Gaussian channel in all scenarios, it is still 
useful in providing insights in many cases, as will be seen in the
next section. Moreover, its analytic simplicity allows an exact
analysis of the relay network capacity. This in turns provides the
foundation for the analysis of the Gaussian network.

\section{Motivation of our approach}
\label{sec:motivation}

In this section we motivate and illustrate our approach. We look at
three simple relay networks and illustrate how the analysis of these
networks under the simpler linear finite-field deterministic model
enables us to conjecture an approximately optimal relaying scheme for the
Gaussian case. We progress from the relay channel where
several strategies yield uniform approximation to more complicated
networks where progressively we see that several ``simple'' strategies
in the literature fail to achieve a constant gap. Using the
deterministic model we can whittle down the potentially successful
strategies. This illustrates the power of the deterministic
model to provide insights into transmission techniques for noisy
networks.

The network is assumed to be synchronized, {\em i.e.,} all
transmissions occur on a common clock. The relays are allowed to do
any {\em causal} processing. Therefore their current output depends
only its past received signals.  For any such network, there is a
natural information-theoretic cut-set bound \cite{CoverThomas}, which
upper bounds the reliable transmission rate $R$. Applied to the relay
network, we have the cut-set upper bound $\overline{C}$ on its
capacity: \beq
\label{eq:CutSetRef} \overline{C}=
\max_{p(\{\xbf_j\}_{j\in\mathcal{V}})} \min_{\Omega\in\Lambda_D}
I(\ybf_{\Omega^c};\xbf_{\Omega}|\xbf_{\Omega^c}) \eeq where
$\Lambda_D=\{\Omega:S\in\Omega,D\in\Omega^c\}$ is all
source-destination cuts. In words, the value of a
given cut $\Omega$ is the information rate achieved when the nodes
in $\Omega$ fully cooperate to transmit and the nodes in $\Omega^c$
fully cooperate to receive. In the case of Gaussian networks, this
is simply the mutual information achieved in a MIMO channel, the
computation of which is standard. We will use this cut-set bound to
assess how good our achievable strategies are.

\subsection{Single-relay network}
\label{subsec:oneRelayMotivation}

\begin{figure}
     \centering
     \subfigure[The Gaussian relay channel ]{
     \scalebox{0.6}{\includegraphics{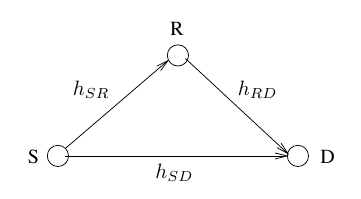}}
}
\hspace{-0.1in}
\subfigure[The  linear finite-field deterministic relay channel]{
       \scalebox{0.6}{\includegraphics{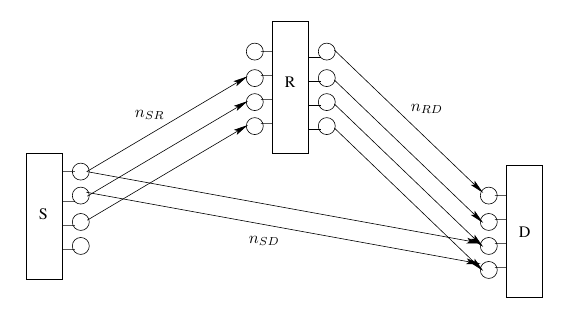}}
}
\caption{The relay channel: (a) Gaussian model, (b) Linear
finite-field deterministic model. \label{fig:oneRelay}}
\end{figure}

\begin{figure}
     \centering

       \scalebox{0.8}{\includegraphics{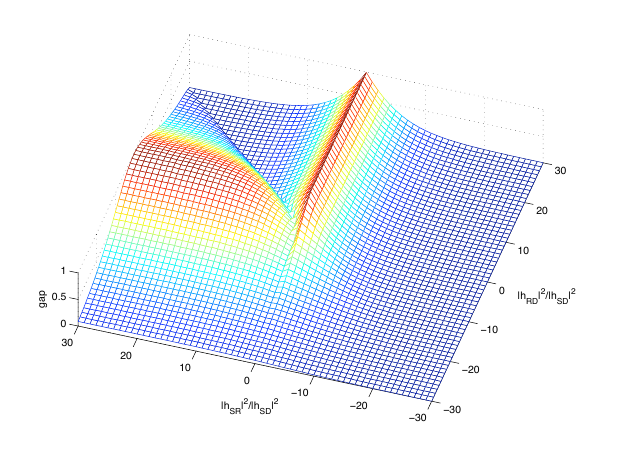}}
 \caption{The x and y axis respectively represent the channel gains from relay to
destination and source to relay normalized by the gain of the direct
link (source to destination) in dB scale. The z axis shows the value
of the gap between the cut-set upper bound and the achievable rate of decode-forward scheme in bits/sec/Hz.\label{fig:relay}}
\end{figure}

We start by looking at the simplest Gaussian relay network with only
one relay as shown in Figure \ref{fig:oneRelay} (a). To approximate its capacity uniformly (uniform
over all channel gains), we need to
find a relaying protocol that achieves a rate close to an upper bound
on the capacity for all channel parameters. To find such a scheme we
use the linear finite-field deterministic model to gain insight. The
corresponding linear finite-field deterministic model of this relay
channel with channel gains denoted by $n_{SR}$, $n_{SD}$ and $n_{RD}$
is shown in Figure \ref{fig:oneRelay} (b).  It is easy to see that the
capacity of this deterministic relay channel, $C_{relay}^d$, is
smaller than both the maximum number of bits that the source can broadcast, and the maximum number of bits that the destination can receive. Therefore
\begin{align} \nonumber & C_{relay}^d \leq \min \lp
\max(n_{SR},n_{SD}), \max(n_{RD},n_{SD}) \rp \\
\label{eq:det_relay_cap} & \quad =  \left\{%
\begin{array}{ll}
    n_{SD}, & \hbox{if $n_{SD}>\min\lp n_{SR},n_{RD} \rp $;} \\
    \min\lp n_{SR},n_{RD} \rp, & \hbox{otherwise.} \\
\end{array}%
\right.     \end{align} It is not difficult to see that this is
in fact the cut-set upper bound for the linear deterministic
network.

Note that Equation (\ref{eq:det_relay_cap}) naturally implies a
capacity-achieving scheme for this deterministic relay network: if
the direct link is better than any of the links to/from the relay
then the relay is silent, otherwise it helps the source by decoding
its message and sending innovations. In the example of Figure
\ref{fig:oneRelay}, the destination receives two bits directly from
the source, and the relay increases the capacity by $1$ bit by
forwarding the least significant bit it receives on a level that
does  not overlap with the direct transmission at the destination.
This suggests a decode-and-forward scheme for the original Gaussian
relay channel. The question is: how does it perform?  Although
unlike in the deterministic network, the decode-forward protocol
cannot achieve exactly the cut-set bound in the Gaussian nettwork,
the following theorem shows it is close.

\begin{theorem} \label{thm:oneRelay1Bit}
The decode-and-forward relaying protocol achieves within 1 bit/s/Hz
of the cut-set bound  of the single-relay Gaussian network, for all
channel gains.
\end{theorem}
\begin{proof}
See Appendix \ref{app:oneRelay1Bit}.
\end{proof}

We should point out that even this 1-bit gap is too conservative for
many parameter values. In fact the gap would be at the maximum value
only if two of the channel gains are exactly the same. This
is rare in wireless scenarios. In Figure \ref{fig:relay} the gap between the
achievable rate of decode-forward scheme and the cut-set upper bound
is plotted for different channel gains.

The deterministic network in Figure \ref{fig:oneRelay}
(b) suggests that several other relaying strategies are also
optimal. For example, compress-and-forward \cite{coverElgamal} will also
achieve the cut-set bound. Moreover a ``network coding'' strategy of
sending the sum (or linear combination) of the received bits is
also optimal as long as the destination receives linearly
independent equations. All these schemes can also be translated
to the Gaussian case and can be shown to be uniformly approximate
strategies. Therefore for the simple relay channel there are many
successful candidate strategies. 

\subsection{Diamond network}
\label{subsec:diamondMotivation}

\begin{figure*}
     \centering
     \subfigure[The Gaussian diamond network]{
     \includegraphics{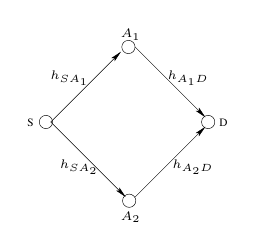}
}
\hspace{0in}
\subfigure[The  linear finite-field deterministic diamond network]{
       \scalebox{0.7}{\includegraphics{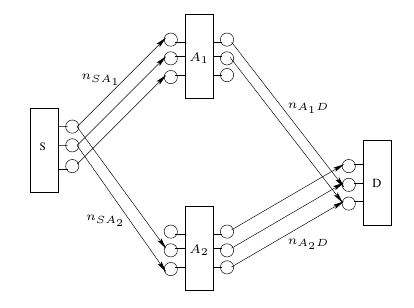}}
}
\subfigure[Wired diamond network]{
       \scalebox{0.6}{\includegraphics{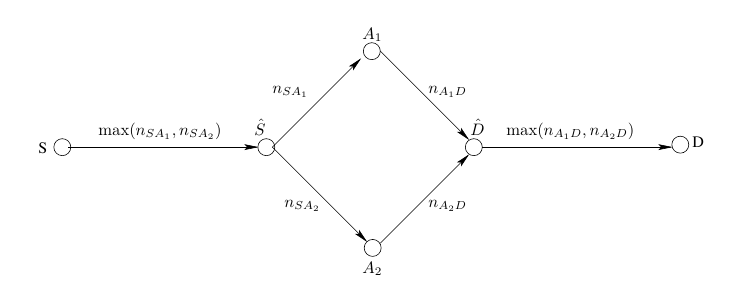}}
}
\caption{Diamond network with two relays: (a) Gaussian model, (b) Linear finite-field deterministic model, and (c) wired model. \label{fig:diamond}}
\end{figure*}


Now consider the diamond Gaussian relay network, with two relays, as
shown in Figure \ref{fig:diamond} (a). Schein introduced this network
in his Ph.D.  thesis \cite{ScheinThesis} and investigated its
capacity. However the capacity of this network is still open. We would like to uniformly approximate
its capacity.

First we build the corresponding linear finite-field deterministic
model for this relay network as shown in Figure \ref{fig:diamond}
(b). To investigate its capacity first we relax the interactions
between incoming links at each node and create the wired network
shown in Figure \ref{fig:diamond} (c). In this network there are two
other links added, which are from $S$ to $\hat{S}$ and from $\hat{D}$
to $D$. Since the capacities of these links are respectively equal to
the maximum number of bits that can be sent by the source and maximum
number of bits that can be received by the destination in the original
linear finite-field deterministic network, the capacity of the
wired diamond network cannot be smaller than the capacity of the
linear finite-field deterministic diamond network. Now by the max-flow
min-cut theorem we know that the capacity $C_{diamond}^w$ of the
wired diamond network is equal to the value of its minimum
cut. Hence
\begin{align}
\nonumber & C_{\text{diamond}}^d  \leq C_{\text{diamond}}^w  = \min  \{ \max(n_{SA_1},n_{SA_2}),\\
&\quad \max(n_{A_1D},n_{A_2D}) ,n_{SA_1}+n_{A_2D},n_{SA_2}+n_{A_1D}  \}.
\end{align}
As we will show in Section \ref{sec:detCapacity}, this upper bound is
in fact the cut-set upper bound on the capacity of the deterministic
diamond network.

Now, we know that the capacity of a wired network is
achieved by a routing solution. We can indeed mimic the wired network routing solution in the linear finite-field
deterministic diamond network and  send the same amount of
information through non-interfering links from source to relays and
then from relays to destination. Therefore the capacity of the
deterministic diamond network is equal to its cut-set upper bound.

A natural analogy of this routing scheme for the Gaussian network is the following partial-decode-and-forward strategy:
\begin{enumerate}
  \item The source broadcasts two messages, $m_1$ and $m_2$, at rate
  $R_1$ and $R_2$ to relays $A_1$ and $A_2$, respectively.
  \item Each relay $A_i$ decodes message $m_i$, $i=1,2$.
  \item Then $A_1$ and $A_2$ re-encode the messages and transmit them
  via the MAC channel to the destination.
\end{enumerate}
Clearly  the destination can decode both $m_1$ and $m_2$
if $(R_1,R_2)$ is inside the capacity region of the BC from source
to relays as well as the capacity region of the MAC from relays to
the destination.  The following theorem  shows how good this scheme
is.

\begin{theorem}\label{thm:twoRelay1Bit}
Partial-decode-and-forward relaying protocol achieves within 1 bit/s/Hz
of the cut-set upper bound of the two-relay diamond Gaussian
network, for all channel gains.
\end{theorem}

\begin{proof}
See Appendix \ref{app:twoRelay1Bit}.
\end{proof}

We can also use the linear finite-field deterministic model to
understand why other simple protocols such as decode-forward and
amplify-forward are not universally-approximate strategies for the
diamond  network.

Consider an example linear finite-field  diamond
network shown in Figure \ref{fig:diamondDetEx} (a). The
cut-set upper bound on the capacity of this network is 3 bits/unit
time. In a decode-forward scheme, all participating relays should be
able to decode the message. Therefore the maximum rate of the
message broadcasted from the source can at most be 2 bits/unit time.
Also, if we ignore relay $A_2$ and only use the stronger relay,
still it is not possible to send information more at a rate more
than 1 bit/unit time. As a result we cannot achieve the capacity of
this network by using a decode-forward strategy.

We next show that this 1-bit gap can be translated into an unbounded
gap in the corresponding Gaussian network, as shown in Figure
\ref{fig:diamondDetEx} (b).  By looking at the cut between the
destination and the rest of the network, it can be seen that for
large $a$, the cut-set upper bound is approximately \beq
\label{eq:CutSetDiamond} \overline{C} \approx 3 \log a.\eeq The achievable rate of the decode-forward strategy is upper
bounded by \beq \label{eq:DFUpperBound} R_{DF} \leq 2 \log a. \eeq
Therefore, as $a$ gets larger, the gap between the achievable rate
of decode-forward strategy and the cut-set upper bound
(\ref{eq:CutSetDiamond}) increases.

Let us look at the amplify-forward scheme.  Although this scheme
does not require all relays to decode the entire message, it can be
quite sub-optimal if relays inject significant noise into the
system. We use the deterministic model to intuitively see this
effect. In a deterministic network, the amplify-forward operation
can be simply modeled by shifting bits up and down at each node.
However, once the bits are shifted up, the newly created LSB's
represent the amplified bits of the noise and we model them by
random bits. Now, consider the example shown in Figure
\ref{fig:diamondDetEx} (a). We notice that to achieve a rate of 3
from the source to the destination, the least significant bit of the source's signal should go through $A_1$ while the
remaining two bits go through $A_2$.  Now if $A_2$ is doing
amplify-forward, it will have two choices: to either forward the
received signal without amplifying it, or to amplify the received
signal to have three signal levels in magnitude and forward it.

The effective networks under these two strategies are respectively
shown in Figure \ref{fig:diamondDetEx} (c) and
\ref{fig:diamondDetEx} (d). In the first case, since the total rate
going through the MAC from $A_1$ and $A_2$ to $D$ is less than two,
the overall achievable rate cannot exceed two. In the second case,
however, the inefficiency of amplify-forward strategy comes from the
fact that $A_2$ is transmitting pure noise on its lowest signal
level. As a result, it is corrupting the bit transmitted by $A_1$
and reducing the total achievable rate again to two bits/channel use.
Therefore, for this channel realization, the amplify-forward scheme does
not achieve the capacity. This intuition can again be translated to
the corresponding Gaussian network to show that amplify-and-forward
is not a universally-approximate strategy for the diamond network.

\begin{figure}
     \centering
     \subfigure[]{
       \scalebox{0.6}{\includegraphics{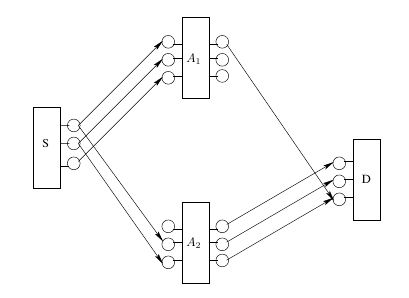}}}
    \hspace{0in}
    \subfigure[]{
       \scalebox{0.6}{\includegraphics{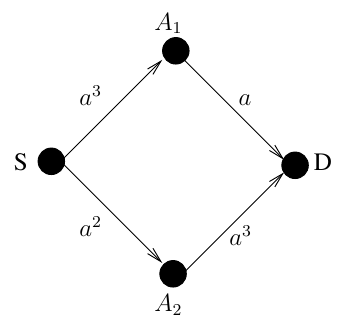}}}
    \hspace{0in}
    \subfigure[]{
       \scalebox{0.6}{\includegraphics{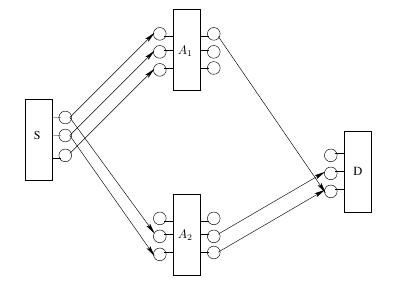}}}
       \hspace{0in}
    \subfigure[]{
       \scalebox{0.6}{\includegraphics{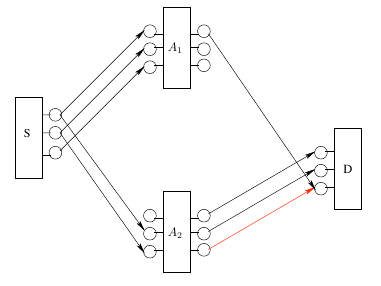}}}
\caption{An example of the linear finite-field deterministic diamond
network is shown in (a). The corresponding Gaussian network is shown
in (b), with the gains chosen such that the ratio of the gains in dB
scale match the ratios of the gains in the deterministic network.
The effective network when $R_2$ just forwards the received signal
is shown in (c). The effective network when $R_2$ amplifies the
received signal to shift it up one signal level and then forward the
message is shown in (d).} \label{fig:diamondDetEx}
\end{figure}

\subsection{A four-relay network}
\label{subsec:fourRelayMotivation}

\begin{figure}
     \centering
     \subfigure[ ]{
       \scalebox{0.7}{ \includegraphics{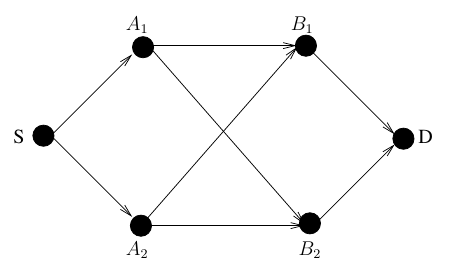}}
}
\caption{\label{fig:fourRelay} A two layer relay network with four relays.}
\end{figure}

\begin{figure*}
     \centering
     \subfigure[]{
       \scalebox{0.6}{\includegraphics{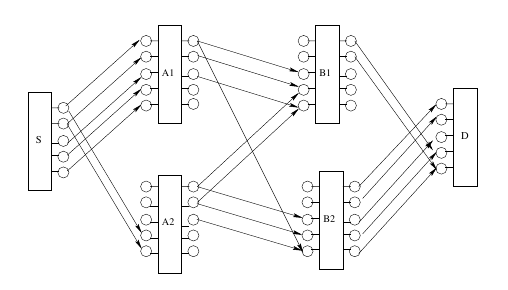}}}
    \hspace{0in}
    \subfigure[]{
       \scalebox{0.7}{\includegraphics{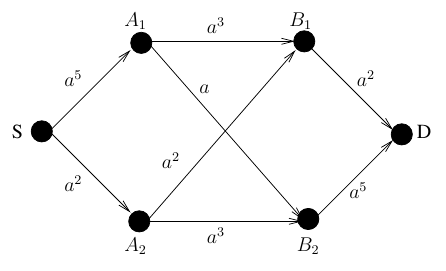}}}
    \hspace{0in}
        \subfigure[ ]{
       \scalebox{0.7}{ \includegraphics{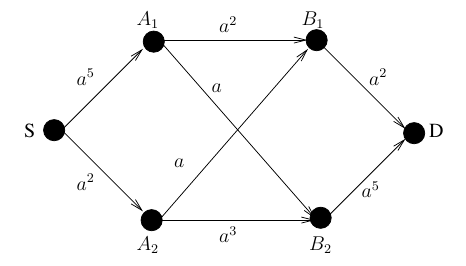}}
}
 \caption{An example of a four relay linear finite filed deterministic relay network is shown in (a). The corresponding Gaussian relay network is shown in (b). The effective Gaussian network for compress-forward strategy is shown in (c).} \label{fig:fourRelayDetEx2}
\end{figure*}

\begin{figure}
     \centering
       \scalebox{0.5}{ \includegraphics{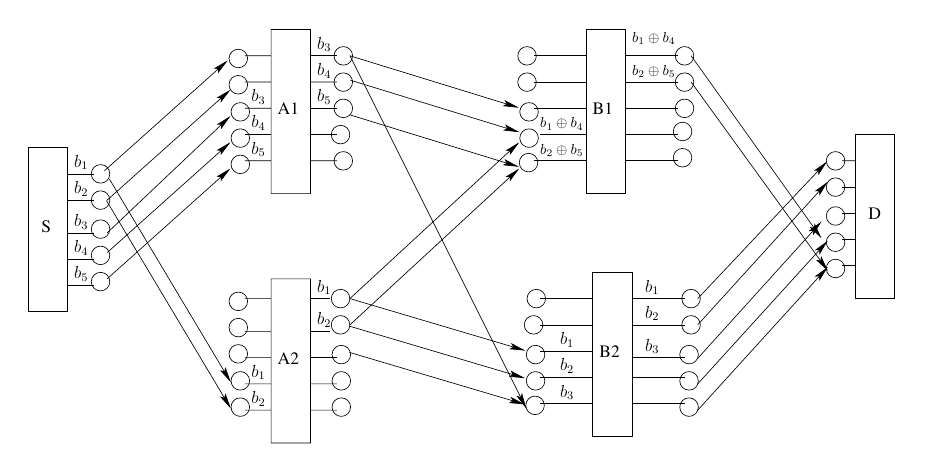}}
\caption{\label{fig:fourRelayStrategy} Demonstration of a capacity achieving strategy.}
\end{figure}

We now look at a more complicated relay network with four relays, as
shown in Figure \ref{fig:fourRelay}. As the first step let us find
the optimal relaying strategy for the corresponding linear finite
field deterministic model. Consider an example of a linear finite
field deterministic relay network shown in Figure
\ref{fig:fourRelayDetEx2} (a).  Now focus on the relaying strategy
that is pictorially shown in Figure \ref{fig:fourRelayStrategy}. In
this scheme,
\begin{itemize}
    \item Source broadcasts $\bbf=[b_1,\ldots,b_5]^t$
    \item Relay $A_1$ decodes $b_3,b_4,b_5$ and relay $A_2$ decodes $b_1,b_2$
    \item Relay $A_1$ and $A_2$ respectively send $\xbf_{A_1}=[b_3,b_4,b_5,0,0]^t$ and $\xbf_{A_2}=[b_1,b_2,0,0,0]^t$
    \item Relay $B_2$ decodes $b_1,b_2,b_3$ and sends $\xbf_{B_2}=[b_1,b_2,b_3,0,0]^t$
    \item Relay $B_1$ receives $\ybf_{B_1}=[0,0,b_3,b_4 \oplus  b_1,b_5 \oplus b_2]^t$ and forwards the last two equations, $\xbf_{B_1}=[b_4 \oplus b_1,b_5 \oplus b_2,0,0,0]^t$
    \item The destination gets $\ybf_{D}=[b_1,b_2,b_3,b_4 \oplus  b_1,b_5 \oplus b_2]^t$ and is able to  decode all five bits.
\end{itemize}

This scheme can achieve $5$ bits per unit time, clearly the best
that one can do since the destination only receives $5$ bits per
unit time. In this optimal scheme the relay $B_1$
is not decoding or partially decoding the original flows of bits that were broadcasted by the source; it is decoding and forwarding a linear combination of them. One may wonder if this is necessary. To answer this question note that since all transmitted signal levels of $A_1$ and $A_2$ are
interfering with each other, it is not possible to get a rate of
more than 3 bits/unit time by any scheme which does not allow mixing of the flows
of information bits originating from the source.

The last stage in the above scheme can actually be interpreted as a
compress-and-forward strategy: relays $B_1$ and $B_2$ want to send
their $3$-bit received vectors to the destination $D$, but because
the link from $B_1$ to $D$ only supports $2$ bits, the dependency
between these received vectors must be exploited. However, in the
Gaussian network, we {\em cannot} implement this strategy using a
standard compress-and-forward scheme pretending that the two
received signals at $B_1$ and $B_2$ are jointly Gaussian. They are
not. Relay $A_2$ sends nothing on its LSB, allowing the MSB of relay
$A_1$ to come through and appear as the LSB of the received signal
at $B_2$. In fact, the statistical correlation between the
real-valued received signals at $B_1$ and $B_2$ is quite weak since
their MSBs are totally independent. Only when one views the received
signals as vectors of bits, as guided by the deterministic model, 
the dependency between them becomes apparent. In fact, it can be
shown that a compress-and-forward strategy assuming jointly Gaussian
distributed received signals cannot achieve a constant gap to the
cut-set bound.

\subsection{Summary}

We learned two key points from the above examples:

\begin{itemize}

\item All the schemes that achieve capacity of the deterministic
networks in the examples forward the received bits at the various
signal levels.

\item Using the deterministic model as a guide, it is revealed that
commonly used schemes such as decode-and-forward, amplify-and-forward and Gaussian
compress-and-forward can all be very far-away from the cut-set
bound.
\end{itemize}

We devote the rest of the paper to generalizing the steps we took
for the examples. As we will show, in the deterministic relay
network the optimal strategy for each relay is to
simply shuffle and linearly combine the received signals at various levels and forward them. This insight leads to a natural {\em
quantize-map-and-forward} strategy for noisy (Gaussian) relay networks.
The strategy for each relay is to quantize the received signal at the 
distortion of the noise power. This in effect extracts the bits of
the received signals above the noise level. These bits are then
mapped randomly to a transmit Gaussian codeword. The main result of
our paper is to show that such a scheme is indeed universally
approximate for arbitrary noisy Gaussian relay networks.

\section{Main Results}
\label{sec:mainResults}

In this section we precisely state the main results of the paper and
briefly discuss their implications. The capacity of a relay
network, $C$, is defined as the supremum of all achievable rates of
reliable communication from the source to the destination. Similarly,
the multicast capacity of relay network is defined as the maximum rate
at which the source can send the same information simultaneously to all
destinations.

\subsection{Deterministic networks}

\subsubsection{General deterministic relay network}

\label{subsec:MainResGenDet}
\label{subsec:genDetModel}

In the general deterministic model the received vector signal
${\bf y}_j$ at node $j\in\mathcal{V}$ at time $t$ is given by
\begin{equation}
\label{eq:GenDetModel}
{\bf y}_j[t] = {\bf g}_j(\{{\bf x}_i[t]\}_{i\in \mathcal{V}}),
\end{equation}
where $\{{\bf x}_i[t]\}_{i\in\mathcal{V}}$ denotes the transmitted signals at all of the nodes in the network.
Note that this implies a deterministic multiple access channel for node $j$
and a deterministic broadcast channel for the transmitting nodes, so
both broadcast and multiple access is allowed in this model. This is
a generalization of Aref networks \cite{ArefThesis} which only allow
broadcast.

The cut-set bound of a general deterministic relay network is:
\begin{eqnarray}
\label{eq:CutSet}
\overline{C} &=&\max_{p(\{\textbf{x}_j\}_{j\in\mathcal{V}})} \min_{\Omega\in\Lambda_D}
I(\ybf_{\Omega^c};\textbf{x}_{\Omega}|\textbf{x}_{\Omega^c}) \\ \label{eq:CutSetGenDet}
& \stackrel{(a)}{=} &
\max_{p(\{\textbf{x}_j\}_{j\in\mathcal{V}})} \min_{\Omega\in\Lambda_D}
H(\ybf_{\Omega^c}|\textbf{x}_{\Omega^c})
\label{eq:CutSetLFF}
\end{eqnarray}
where $\Lambda_D=\{\Omega:S\in\Omega,D\in\Omega^c\}$ is all
source-destination cuts. Step $(a)$ follows since we are dealing with
deterministic networks.

The following are our main results for arbitrary deterministic networks.

\begin{theorem}
\label{thm:GenDetNet} A rate of
\begin{equation}
\label{eq:GenDetNet}
\max_{\prod_{i\in\mathcal{V}} p(\textbf{x}_i)} \min_{\Omega\in\Lambda_D}
H(\ybf_{\Omega^c}|\textbf{x}_{\Omega^c})
\end{equation}
can be achieved on a deterministic network.
\end{theorem}
This theorem easily extends to the multicast case, where we want to
simultaneously transmit one message from $S$ to all destinations in
the set $D\in\mathcal{D}$:
\begin{theorem}
\label{thm:GenDetNetMulticast} A multicast rate of
\begin{equation}
\label{eq:GenDetNetMulticast}
\max_{\prod_{i\in\mathcal{V}} p(\textbf{x}_i)} \min_{D\in\mathcal{D}} \min_{\Omega\in\Lambda_D}
H(\ybf_{\Omega^c}|\textbf{x}_{\Omega^c})
\end{equation}
to all the destinations $D\in\mathcal{D}$ can be achieved on a
deterministic network.

\end{theorem}

 Note that when we compare
(\ref{eq:GenDetNet}) to the cut-set upper bound in
(\ref{eq:CutSetGenDet}), we see that the difference is in the
maximizing set, {\em i.e.,} we are only able to achieve independent
(product) distributions whereas the cut-set optimization is over any
arbitrary distribution. In particular, if the network and the
deterministic functions are such that the cut-set is optimized by the
product distribution, then we would have matching upper and lower
bounds. This happens for deterministic networks with broadcast only, specializing to the result in \cite{RK06}. It also happens when we consider the linear finite-field
model, whose results are stated next.

\subsubsection{Linear finite-field deterministic relay network}
\label{subsec:MainResLinDet}

Applying the cut-set bound to the linear finite-field deterministic
relay network defined in Section \ref{subsec:LFFDetModel},
(\ref{eq:LinDetModel}), and using (\ref{eq:CutSetLFF}) since we have
a deterministic network, we get:
\beq
\label{eq:CutSetEv}
\overline{C}
=
\max_{p(\{\textbf{x}_j\}_{j\in\mathcal{V}})} \min_{\Omega\in\Lambda_D}
H(\ybf_{\Omega^c}|\textbf{x}_{\Omega^c})
\stackrel{(b)}{=}
 \min_{\Omega\in\Lambda_D}
\mathrm{rank}(\Gbf_{\Omega,\Omega^c})
\eeq
where $\Gbf_{\Omega,\Omega^c}$ is
the transfer matrix associated with the cut $\Omega$, {\em i.e.,} the matrix
relating the vector of all the inputs at the nodes in $\Omega$ to the
vector of all the outputs in $\Omega^c$ induced by
(\ref{eq:LinDetModel}). This is illustrated in Figure \ref{fig:LinCutSet}. Step $(b)$ follows since in the linear
finite-field model all cut values (\emph{i.e.},
$H(\ybf_{\Omega^c}|\textbf{x}_{\Omega^c})$) are simultaneously optimized by
independent and uniform distribution of $\{\textbf{x}_i\}_{i\in\mathcal{V}}$
and the optimum value of each cut $\Omega$ is logarithm of the size of
the range space of the transfer matrix $\Gbf_{\Omega,\Omega^c}$
associated with that cut. Theorems \ref{thm:GenDetNet}  and \ref{thm:GenDetNetMulticast} immediately imply that this cutset bound is achievable.

\begin{figure}[h]
\begin{center}
\scalebox{0.8}{\includegraphics{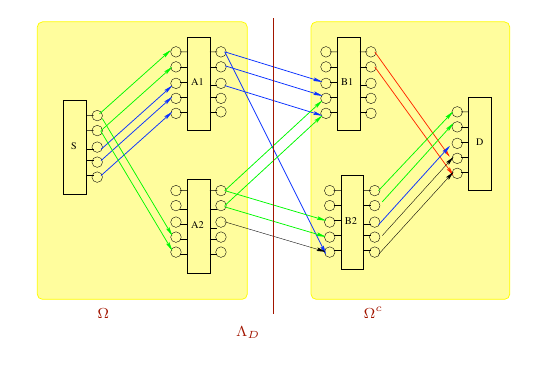}}
\end{center}
\vspace{-0.2in}
\caption{Illustration of cut-set bound and cut-set transfer matrix
  $\Gbf_{\Omega,\Omega^c}$.}
\label{fig:LinCutSet}
\end{figure}

\begin{theorem}
\label{thm:LinDetNet} The capacity $C$ of a linear finite-field
deterministic relay network is given by
\begin{eqnarray}
\label{eq:ThmLinDetNet}
C = \min_{\Omega\in\Lambda_D}
\mathrm{rank}(\Gbf_{\Omega,\Omega^c}).
\end{eqnarray}
\end{theorem}

\begin{theorem}
\label{thm:LinDetNetMulticast} The multicast capacity $C$ of a
linear finite-field deterministic relay network is given by
\begin{eqnarray}
\label{eq:ThmLinDetNetMulticast}
C = \min_{D\in\mathcal{D}}\min_{\Omega\in\Lambda_D}
\mathrm{rank}(\Gbf_{\Omega,\Omega^c})
\end{eqnarray}
where $\mathcal{D}$ is the set of destinations.
\end{theorem}

\textit{Remark:} Note that the results in Theorems \ref{thm:GenDetNet},
\ref{thm:GenDetNetMulticast}, \ref{thm:LinDetNet} and
\ref{thm:LinDetNetMulticast} apply to networks with arbitrary
topology, possible including cycles. For a single
source-destination pair the result in Theorem \ref{thm:LinDetNet}
generalizes the classical max-flow min-cut theorem for wired
networks and for multicast, the result in Theorem
\ref{thm:LinDetNetMulticast} generalizes the network coding result in
\cite{ACLY00}. As we will see in the proof, the encoding
functions at the relay nodes for the linear finite-field model can
be restricted to linear functions to obtain the result in Theorem
\ref{thm:LinDetNet}. 

\subsection{Gaussian relay networks}
\label{subsec:MainResGaussRelay}
\label{subsec:GaussianModel}
In the Gaussian model each node $j\in\mathcal{V}$ has $M_j$ transmit and $N_j$ receive antennas. The received signal ${\bf y}_j$ at node $j$ and
time $t$ is
\begin{equation}
  \label{eq:GaussianModel}
  {\bf y}_{j}[t]=\sum_{i\in \mathcal{V}} {\bf H}_{ij} {\bf x}_i[t] + {\bf z}_j[t]
\end{equation}
where ${\bf H}_{ij}$ is an $M_i \times N_j$ complex matrix whose
$(k,l)$ element represents the channel gain from the $k$-th transmit
antenna in node $i$ to the $l$-th receive antenna in node
$j$. Furthermore, we assume there is an average power constraint equal
to 1 at each transmit antenna. Also ${\bf z}_j$, representing the
channel noise, is modeled as complex Gaussian random vector. The
Gaussian noises at different receivers are assumed to be independent
of each other.

The following are our main results for Gaussian relay networks; it is proved in Section \ref{sec:GaussCapacity}.
\begin{theorem}
\label{thm:GaussMain}
The capacity $C$ of the Gaussian relay network satisfies
\begin{equation}
\label{eq:GaussRelayMain}
\overline{C}-\kappa \leq C \leq \overline{C},
\end{equation}
where $\overline{C}$ is the cut-set upper bound on the capacity of
$\mathcal{G}$ as described in Equation (\ref{eq:CutSetRef}), and
$\kappa$ is a constant and is upper bounded by
$12 \sum_{i=1}^{|\mathcal{V}|} N_i+ 3 \sum_{i=1}^{|\mathcal{V}|}M_i $.
\end{theorem}

\textit{Remark:} The gap $\kappa$ holds for all values of the channel gains and the result is
relevant particularly in the high rate regime. It is a stronger result than a degree-of-freedom result,
because it is non-asymptotic and provides a uniform guarantee to optimality for all channel SNRs. This is the first constant-gap approximation
of the capacity of Gaussian relay networks.  As shown in Section
\ref{sec:motivation}, the gap between the achievable rate of well
known relaying schemes and the cut-set upper bound in general depends
on the channel parameters and can become arbitrarily large. Analogous
to the results for deterministic networks, the result in Theorem
\ref{thm:GaussMain} applies to a network with arbitrary topology,
possibly with cycles.

The result in Theorem \ref{thm:GaussMain} easily extends to the
multicast case where we want to simultaneously transmit one message
from $S$ to all destinations in the set $D\in\mathcal{D}$.

\begin{theorem}
\label{thm:GaussMainMulticast} The multicast capacity $C_{\mathrm{mult}}$ of the Gaussian 
relay network satisfies
\begin{equation}
\label{eq:GaussRelayMainMult}
\overline{C}_{\mathrm{mult}}-\kappa \leq C_{\mathrm{mult}} \leq \overline{C}_{\mathrm{mult}},
\end{equation}
where $\overline{C}_{\mathrm{mult}}$ is the multicast cut-set upper bound on the capacity of
$\mathcal{G}$ given by 
\begin{equation}
\label{eq:MultCutSet}
\displaystyle
\overline{C}_{\mathrm{mult}} = \max_{p(\{\xbf_j\}_{j\in\mathcal{V}})} \min_{D\in\mathcal{D}}\min_{\Omega\in\Lambda_D}
I(\ybf_{\Omega^c};\xbf_{\Omega}|\xbf_{\Omega^c}),
\end{equation} 
and $\kappa$ is a constant and is upper bounded by $12
\sum_{i=1}^{|\mathcal{V}|} N_i+
3 \sum_{i=1}^{|\mathcal{V}|}M_i$.
\end{theorem}

\textit{Remark:} The gap $\kappa$ stated in Theorems
\ref{thm:GaussMain}-\ref{thm:GaussMainMulticast} hold for 
scalar quantization scheme explored in detail in Section \ref{sec:GaussCapacity}. It is shown in \cite{OD10} that a vector
quantization scheme even with structured lattice codebooks can improve
this constant to $2\sum_{i=1}^{|\mathcal{V}|}
N_i+\min\{\sum_{i=1}^{|\mathcal{V}|}M_i,\sum_{i=1}^{|\mathcal{V}|}N_i\}\leq
3\sum_{i=1}^{|\mathcal{V}|} N_i$, which means when all nodes have
single antennas, the gap is at most $3|\mathcal{V}|$, for complex
Gaussian networks (or $1.5|\mathcal{V}|$ for real Gaussian networks). Also, the
results have been extended to the case when there are multiple
sources and {\em all} destinations need to decode {\em all} the
sources, \emph{i.e.,} multi-source multicast, in \cite{PDT09}.

\subsection{Proof program}
In the following sections we formally prove these main results. The
main proof program consists of first proving Theorem
\ref{thm:LinDetNet} and the corresponding multicast result for linear
finite-field deterministic networks in Section
\ref{sec:detCapacity}. Since the proof logic of the achievable rate
for general deterministic networks (\ref{eq:GenDetNet},
\ref{eq:GenDetNetMulticast}) is similar to that for the  linear
case, Theorems \ref{thm:GenDetNet} and \ref{thm:GenDetNetMulticast}
are proved in Appendix \ref{app:GenDet}. We use the proof ideas for the deterministic analysis to obtain the universally-approximate
capacity characterization for Gaussian relay networks in Section
\ref{sec:GaussCapacity}. In both cases we illustrate the proof by first
going through an example.

\section{Deterministic relay networks}
\label{sec:detCapacity}

In this section we characterize the capacity of linear finite-field
deterministic relay networks and prove Theorems \ref{thm:LinDetNet}
and \ref{thm:LinDetNetMulticast}. 

To characterize the capacity of linear finite-field deterministic
relay networks, we first focus on networks that have a layered
structure, i.e., all paths from the source to the destination have
equal lengths. With this special structure we get a major
simplification: a sequence of messages can each be encoded into a
block of symbols and the blocks do not interact with each other as
they pass through the relay nodes in the network. The proof of the
result for layered network is similar in style to the random coding
argument in Ahlswede et al. \cite{ACLY00}. We do this in Section
\ref{sec:LayFF}. Next, in Section \ref{sec:TimExpDet}, we extend the
result to an arbitrary network by expanding the network over
time\footnote{The concept of time-expanded network is also used in
  \cite{ACLY00}, but the use there is to handle cycles. Our main use
  is to handle interaction between messages transmitted at different
  times, an issue that only arises when there is superposition of
  signals at nodes.}. Since the time-expanded network is layered we
can apply our result in the first step to it and complete the proof.

\subsection{Layered networks}
\label{sec:LayFF}
The network given in Figure \ref{fig:Confusability} is an example of a
layered network where the number of hops for each path from
$S$ to $D$ is three. We start by describing the encoding scheme.

\subsubsection{Encoding for layered linear deterministic relay network}
\label{subsec:EncLinDetLay}
We have a single source $S$ with a sequence of messages $w_k
\in\{1,2,\ldots,2^{TR}\}$, $k=1,2,\ldots$.  Each message is encoded by
the source $S$ into a signal over $T$ transmission times (symbols),
giving an overall transmission rate of $R$. Relay $j$ operates over
blocks of time $T$ symbols, and uses a mapping
$f_j:\mathcal{Y}_j^T\rightarrow \mathcal{X}_j^T$ on its received symbols
from the previous block of $T$ symbols to transmitted signals in the next
block. For the linear deterministic model (\ref{eq:LinDetModel}), we use linear
mappings $f_j(\cdot)$, {\em i.e.,}
\begin{equation}
\label{eq:LinEncFnLay}
\xbf_j = \Fbf_j \ybf_j ,
\end{equation}
where the vectors $\xbf_j=[\xbf_j[1], \ldots, \xbf_j[T]]^t$ and
$\ybf_j=[\ybf_j[1], \ldots, \ybf_j[T]]^t$ respectively represent the
transmit and received signals over $T$ time units, and the matrix
$\Fbf_j$ is chosen uniformly randomly over all matrices in
$\FF_2^{qT\times qT}$.  Each relay does the encoding prescribed by
(\ref{eq:LinEncFnLay}).  Given the knowledge of all the encoding
functions $\Fbf_j$ at the relays, the destination $D$
attempts to decode each message $w_k$ sent by the source. This
encoding strategy is illustrated in Figure \ref{fig:LinEncStrategy}.

\begin{figure}
\centering
\scalebox{0.8}{\includegraphics{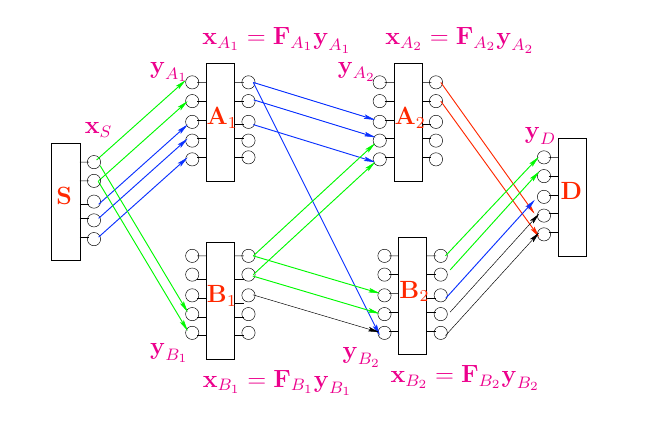}}
\caption{Illustration of linear encoding strategy.\label{fig:LinEncStrategy}}
\end{figure}


Suppose message $w_k$ is sent by the source in block $k$. Since each relay $j$ operates only on block of lengths $T$ and the network is layered, the signals received at block $k$ at any relay
pertain to only message $w_{k-l_j}$ where $l_j$ is the path length
from source to relay $j$.

\subsubsection{Proof illustration}
\label{subsec:PfIdeaDet}
In order to illustrate the proof ideas of Theorem \ref{thm:GenDetNet}
we examine the network shown in Figure \ref{fig:Confusability}.

\begin{figure}[h]
\centering
\scalebox{0.7}{\includegraphics{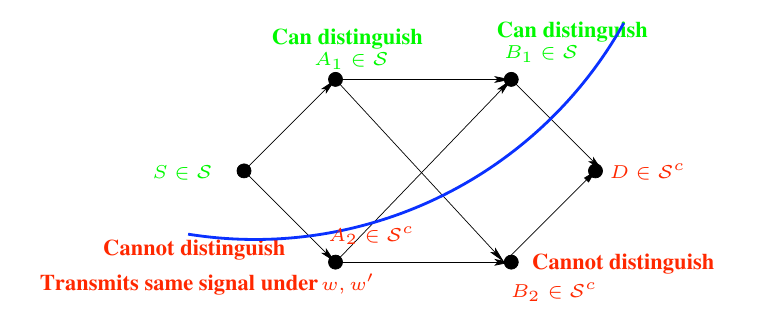}}
\caption{An example of layered relay network. Nodes on the left hand
  side of the cut can distinguish between messages $w$ and $w'$, while
  nodes on the right hand side cannot. \label{fig:Confusability} }
\end{figure}

Without loss of generality
consider the message $w=w_1$ transmitted by the source at block
$k=1$. At node $j$ the signals pertaining to this message are received
by the relays at block $l_j$. For notational simplicity we
will drop the block numbers associated with the transmitted and
received signals for this analysis.

Now, since we have a deterministic network, the message $w$ will be
mistaken for another message $w'$ only if the received signal
$\ybf_D(w)$ under $w$ is the same as that would have been
received under $w'$. This leads to a notion of {\em
distinguishability}: messages $w,w'$ are distinguishable
at any node $j$ if $\ybf_j(w)\neq \ybf_j(w')$.

The probability of error at destination $D$ can be upper bounded using the
union bound as
\begin{equation}
\label{eq:PeUB}
P_e \leq 2^{RT}\prob{w\rightarrow w'}
= 2^{RT}\prob{\ybf_D(w)=\ybf_D(w')} .
\end{equation}
Since channels are deterministic, the randomness is only due to that
of the encoder maps. Therefore, the probability of this event depends
on the probability that we choose such encoder maps.  Now, we can
write
{\small
\begin{align}
\nonumber & \prob{w\rightarrow w'} = \sum_{\Omega\in\Lambda_D}\\
& \label{eq:PwErr}
 \underbrace{\prob{\mbox{Nodes in } \Omega
\mbox{ can distinguish } w,w' \mbox{ and nodes in } \Omega^c \mbox{
cannot} }}_{\mathcal{P}}, ~
\end{align}
}
since the events that correspond to occurrence of the
distinguishability sets $\Omega\in\Lambda_D$ are disjoint. Let us
examine one term in the summation in (\ref{eq:PwErr}). For example,
consider the cut $\Omega=\{S,A_1,B_1\}$ shown in Figure
\ref{fig:Confusability}. A necessary condition for this cut to be the
distinguishability set is that $\ybf_{A_2}(w)=\ybf_{A_2}(w')$, along
with $\ybf_{B_2}(w)=\ybf_{B_2}(w')$ and $\ybf_{D}(w)=\ybf_{D}(w')$.
We first define the following events:
\begin{eqnarray}
  \nonumber \mathcal{A}_i&=& \text{the event that $w$ and $w'$ are undistinguished at node} \\
\nonumber  && \text{$A_i$ (\emph{i.e.}, $\ybf_{A_i}(w) = \ybf_{A_i}(w')$)},\quad i=1,2\\
  \nonumber \mathcal{B}_i&=& \text{the event that $w$ and $w'$ are undistinguished at node} \\
\nonumber   && \text{$B_i$ (\emph{i.e.}, $\ybf_{B_i}(w) = \ybf_{B_i}(w')$)},\quad i=1,2\\
\nonumber \nonumber \mathcal{D}&=& \text{the event that $w$ and $w'$ are undistinguished at node} \\
\label{eq:EventsDetDef} && \text{$D$ (\emph{i.e.}, $\ybf_{D}(w) = \ybf_{D}(w')$)}.
\end{eqnarray}
We now have
\begin{eqnarray}
\nonumber \mathcal{P} &=& \mathbb{P} \{\mathcal{A}_2,\mathcal{B}_2,\mathcal{D},\mathcal{A}_1^c,\mathcal{B}_1^c\} \\
\nonumber &=& \mathbb{P} \{\mathcal{A}_2 \} \times \mathbb{P} \{\mathcal{B}_2,\mathcal{A}_1^c|\mathcal{A}_2 \}  \times \mathbb{P} \{ \mathcal{D},\mathcal{B}_1^c| \mathcal{A}_2,\mathcal{B}_2,\mathcal{A}_1^c \} \\
\nonumber & \leq &  \mathbb{P} \{\mathcal{A}_2 \} \times \mathbb{P} \{\mathcal{B}_2|\mathcal{A}_1^c,\mathcal{A}_2 \} \times \mathbb{P} \{ \mathcal{D}|\mathcal{B}_1^c, \mathcal{A}_2,\mathcal{B}_2,\mathcal{A}_1^c \} \\
 & = & \mathbb{P} \{\mathcal{A}_2 \} \times \mathbb{P} \{\mathcal{B}_2|\mathcal{A}_1^c,\mathcal{A}_2 \}  \times \mathbb{P} \{ \mathcal{D}|\mathcal{B}_1^c,\mathcal{B}_2\} 
\end{eqnarray}

where the last step is true since there is an independent random
mapping at each node and we have the following Markov structure in the
network \beq X_S \rightarrow (Y_{A_1},Y_{A_2}) \rightarrow
(Y_{B_1},Y_{B_2}) \rightarrow Y_{D}. \eeq

As the source does a random linear mapping of the message onto
$\xbf_S(w)$, the probability of $\mathcal{A}_2$ is
\begin{eqnarray}
\nonumber
\mathbb{P} \{\mathcal{A}_2 \}&=&\prob{(\Ibf_T\otimes\Gbf_{S,A_2})(\xbf_S(w)-\xbf_S(w'))=\mathbf{0}}\\
\label{eq:ExLinEP1} & =&
2^{-T\mbox{rank}(\Gbf_{S,A_2})},
\end{eqnarray}
because the random mapping given in (\ref{eq:LinEncFnLay}) induces
independent uniformly distributed $\xbf_S(w),\xbf_S(w')$. Here,
$\otimes$ is the Kronecker matrix product\footnote{If $A$ is an m-by-n
  matrix and $B$ is a p-by-q matrix, then the Kronecker product $A
  \otimes B$ is the mp-by-nq block matrix $A \otimes B=
  \left[\begin{array}{ccc}a_{11}B & \cdots & a_{1n}B \\ \vdots &
      \ddots & \vdots \\ a_{m1}B & \cdots & a_{mn}
      B \end{array}\right]$}.  Now, in order to analyze the second
probability, we see that $\mathcal{A}_2$ implies
$\xbf_{A_2}(w)=\xbf_{A_2}(w')$, {\em i.e.,} the {\em same} signal is
sent under both $w,w'$. Also if $\ybf_{A_1}(w) \neq \ybf_{A_1}(w')$,
then the random mapping given in (\ref{eq:LinEncFnLay}) induces
independent uniformly distributed $\xbf_{A_1}(w),\xbf_{A_1}(w')$
Therefore, we get 
\begin{eqnarray}
\nonumber  \mathbb{P} \{\mathcal{B}_2|
\mathcal{A}_1^c,\mathcal{A}_2 \} & =&
\prob{(\Ibf_T\otimes\Gbf_{A_1,B_2})(\xbf_{A_1}(w)-\xbf_{A_1}(w'))=
  \mathbf{0}}\\
  \label{eq:ExLinEP2} &=& 2^{-T\mbox{rank}(\Gbf_{A_1,B_2})}.  \end{eqnarray}

Similarly, we get
\begin{eqnarray}
\nonumber 
\mathbb{P} \{ \mathcal{D}|\mathcal{B}_1^c,\mathcal{B}_2\} &=&
\prob{(\Ibf_T\otimes\Gbf_{B_1,D})(\xbf_{B_1}(w)-\xbf_{B_1}(w'))=\mathbf{0}}\\
\label{eq:ExLinEP3}  &=& 2^{-T\mbox{rank}(\Gbf_{B_1,D})}.  \end{eqnarray}

Putting these together we see
that in (\ref{eq:PwErr}), for the network in Figure
\ref{fig:Confusability}, we have,
\begin{eqnarray}
\nonumber \mathcal{P} & \leq &
2^{-T\mbox{rank}(\Gbf_{S,A_2})}2^{-T\mbox{rank}(\Gbf_{A_1,B_2})}
2^{-T\mbox{rank}(\Gbf_{B_1,D})}\\ \label{eq:ConfProb} &=& 2^{-T\{\mbox{rank}(\Gbf_{S,A_2})
+\mbox{rank}(\Gbf_{A_1,B_2})+\mbox{rank}(\Gbf_{B_1,D})\}}.
\end{eqnarray}
Note that since
\[
\Gbf_{\Omega,\Omega^c} = \left [ \begin{array}{ccc}\Gbf_{S,A_2} &
\Zerobf & \Zerobf \\ \Zerobf  & \Gbf_{A_1,B_2} & \Zerobf \\
\Zerobf & \Zerobf & \Gbf_{B_1,D} \end{array}\right ],
\]
the upper bound for $\mathcal{P}$ in (\ref{eq:ConfProb}) is exactly
$2^{-T\mbox{rank}(\Gbf_{\Omega,\Omega^c})}$. Therefore, by
substituting this back into (\ref{eq:PwErr}) and (\ref{eq:PeUB}), we
get
\begin{equation}
\label{eq:RateBndLinEx}
\displaystyle
P_e \leq 2^{RT} |\Lambda_D| 2^{-T\min_{\Omega\in\Lambda_D}
\mbox{rank}(\Gbf_{\Omega,\Omega^c})},
\end{equation}
which can be made as small as desired if $R<\min_{\Omega\in\Lambda_D}
\mbox{rank}(\Gbf_{\Omega,\Omega^c})$, which is the result claimed in
Theorem \ref{thm:LinDetNet}.

\subsubsection{Proof of Theorems \ref{thm:LinDetNet} and
  \ref{thm:LinDetNetMulticast} for general layered networks}
\label{subsec:LinDet}

Consider the message $w=w_1$ transmitted by the source at block
$k=1$. The message $w$ will be mistaken for another message $w'$ only
if the received signal $\ybf_D(w)$ under $w$ is the same as that would
have been received under $w'$. Hence the probability of error at
destination $D$ can be upper bounded by,
\begin{equation}
\label{eq:PeUB2}
P_e \leq 2^{RT}\prob{w\rightarrow w'}
= 2^{RT}\prob{\ybf_D(w)=\ybf_D(w')} .
\end{equation}
Similar to  Section \ref{subsec:PfIdeaDet}, we can write
\begin{align}
\nonumber & \prob{w\rightarrow w'} = \sum_{\Omega\in\Lambda_D}
\\
 \label{eq:PwErr2} & \underbrace{\prob{\mbox{Nodes in } \Omega
\mbox{ can distinguish } w,w' \mbox{ and nodes in } \Omega^c \mbox{
cannot} }}_{\mathcal{P}} ~
\end{align}
For any such cut $\Omega$, define the following sets:
\begin{itemize}
\item $L_l (\Omega)$: the nodes that are in $\Omega$ and are at layer
  $l$ (for example $S \in L_1 (\Omega)$),
\item $R_l (\Omega)$: the nodes that are in $\Omega^c$ and are at
  layer $l$ (for example $D \in R_{l_D} (\Omega)$).
\end{itemize}

We now define the following events:
\begin{itemize}
\item $\mathcal{L}_l$: Event that the nodes in $L_l$ can distinguish
  between $w$ and $w'$, \emph{i.e.}, $\ybf_{L_l}(w) \neq
  \ybf_{L_l}(w')$,
\item $\mathcal{R}_l$: Event that the nodes in $R_l$ cannot
  distinguish between $w$ and $w'$, \emph{i.e.}, $\ybf_{R_l}(w) =
  \ybf_{R_l}(w')$.
\end{itemize}

Similar to Section \ref{subsec:PfIdeaDet}, we can write
\begin{eqnarray}
  \mathcal{P} &=& \mathbb{P} \{ \mathcal{R}_l,\mathcal{L}_{l-1} ,l=2,\ldots,l_D  \} \\
  &=& \prod_{l=2}^{l_D}  \mathbb{P} \{\mathcal{R}_l,\mathcal{L}_{l-1}|\mathcal{R}_j,\mathcal{L}_{j-1}, j=2,\ldots,l-1\} \\
  &\leq& \prod_{l=2}^{l_D}  \mathbb{P} \{\mathcal{R}_l|\mathcal{R}_j,\mathcal{L}_{j}, j=2,\ldots,l-1\} \\
  &\stackrel{(a)}{=}& \prod_{l=2}^{l_D}  \mathbb{P} \{\mathcal{R}_l|\mathcal{R}_{l-1},\mathcal{L}_{l-1}\} 
 \end{eqnarray}
 where (a) is true due to the Markovian nature of the layered
 network. Note that as in the example, all nodes in $R_{l-1}$ transmit the same
 signal under both $w$ and $w'$ (\emph{i.e.}, $\xbf_j(w)=\xbf_j(w')$,
 $\forall j\in R_{l-1}$). Therefore, just as in Section
 \ref{subsec:PfIdeaDet}, we see that \emph{i.e.},
 \begin{align*}
&   \mathbb{P} \{\mathcal{R}_l|\mathcal{R}_{l-1},\mathcal{L}_{l-1}\}\\
   &=
   \mathbb{P} \{ \ybf_{R_l}(w)=\ybf_{R_l}(w')|\ybf_{L_{l-1}}(w)\neq\ybf_{L_{l-1}}(w'),\\  & \quad \quad \quad \ybf_{R_{l-1}}(w)=\ybf_{L_{R-1}}(w') \} \\
   &=  \mathbb{P} \{\ybf_{R_l}(w)=\ybf_{R_l}(w')|\ybf_{L_{l-1}}(w)\neq\ybf_{L_{l-1}}(w'), \\  & \quad \quad \quad \xbf_{R_{l-1}}(w)=\xbf_{L_{R-1}}(w')\}\\
   &= \mathbb{P} \{(\Ibf_T\otimes\Gbf_{L_{l-1},R_{l}})(\xbf_{L_{l-1}}(w)-\xbf_{L_{l-1}}(w'))= \mathbf{0} | \\  & \quad \quad \quad \ybf_{L_{l-1}}(w)\neq\ybf_{L_{l-1}}(w') \} \\
   &\stackrel{(a)}{=} 2^{-T\mbox{rank}(\Gbf_{L_{l-1},R_{l}})}.
\end{align*}
where $\Gbf_{L_{l-1},R_{l}}$ is the transfer matrix from transmitted
signals in $L_{l-1}$ to the received signals in $R_l$.  Step (a) is
true since $\ybf_{L_{l-1}}(w)\neq\ybf_{L_{l-1}}(w')$ and hence the
random mapping given in (\ref{eq:LinEncFnLay}) induces independent
uniformly distributed $\xbf_{L_{l-1}}(w),\xbf_{L_{l-1}}(w')$.

Therefore we get 
\begin{eqnarray}
\label{eq:GenPEPLin} \mathcal{P} \leq    \prod_{l=2}^d
2^{-T\mbox{rank}(\Gbf_{L_{l-1},R_{l}})}  = 2^{-T\mbox{rank}(\Gbf_{\Omega,\Omega^c})}. 
\end{eqnarray}

By substituting this back into (\ref{eq:PwErr2}) and (\ref{eq:PeUB2}),
we see that
\begin{equation}
\label{eq:RateBndLinDet}
\displaystyle
P_e \leq 2^{RT} |\Lambda_D| 2^{-T\min_{\Omega\in\Lambda_D}
\mbox{rank}(\Gbf_{\Omega,\Omega^c})},
\end{equation}
which can be made as small as desired if $R<\min_{\Omega\in\Lambda_D}
\mbox{rank}(\Gbf_{\Omega,\Omega^c})$, which is the result claimed in
Theorem \ref{thm:LinDetNet} for layered networks.

To prove Theorem \ref{thm:LinDetNetMulticast} for layered networks, we
note that for {\em any} destination $D\in\mathcal{D}$, the probability
of error expression in \eqref{eq:RateBndLinDet} holds. Therefore, if
{\em all} receivers in $\mathcal{D}$ have to be able to decode the
message, then an error occurs if any of them fails to
decode. Therefore, using the union bound and \eqref{eq:RateBndLinDet}
we can bound this error probability as,
\begin{eqnarray}
\nonumber P_e & \leq & 2^{RT} \sum_{D\in\mathcal{D}}|\Lambda_D| 2^{-T\min_{\Omega\in\Lambda_D}} \\
\label{eq:RateBndLinDetMult} &\leq  & 2^{RT} 2^{|\mathcal{V}|} 2^{-T\min_{D\in\mathcal{D}}\min_{\Omega\in\Lambda_D}
\mbox{rank}(\Gbf_{\Omega,\Omega^c})},
\end{eqnarray}
which clearly goes to zero as long as  $R<\min_{D\in\mathcal{D}}\min_{\Omega\in\Lambda_D}
\mbox{rank}(\Gbf_{\Omega,\Omega^c})$, which is the result claimed in
Theorem \ref{thm:LinDetNetMulticast} for layered networks.

Therefore, we have proved a special case of Theorem \ref{thm:LinDetNetMulticast} for
layered networks. 

\subsection{Arbitrary networks (not necessarily layered)}
\label{sec:TimExpDet}

\begin{figure*}
     \centering \subfigure[]{
       \scalebox{0.7}{\includegraphics{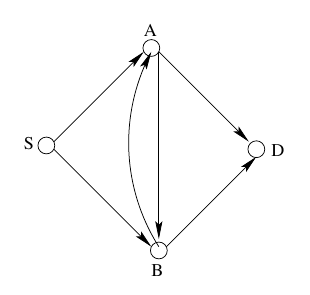}}
}
     \hspace{0in}
     \subfigure[]{
      \scalebox{0.7}{ \includegraphics{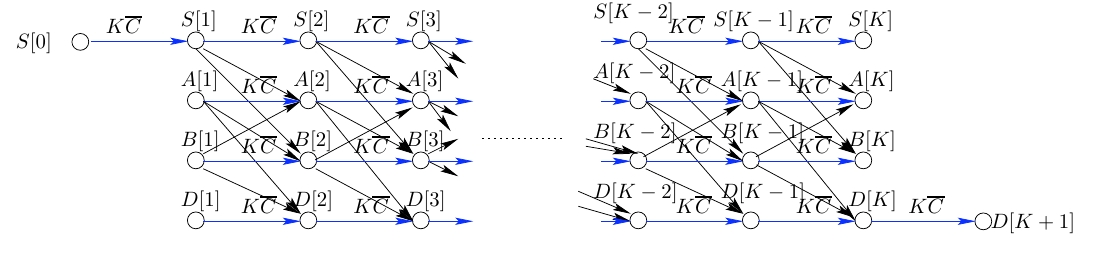}}
}
\caption{An example of a general deterministic network with unequal
paths from S to D is shown in $(a)$.  The corresponding unfolded
network is shown in $(b)$.\label{fig:unfolding}}
\end{figure*}

Given the proof for layered networks with equal path lengths, we are
ready to tackle the proof of Theorem \ref{thm:LinDetNet} and Theorem
\ref{thm:LinDetNetMulticast} for general relay networks. The
ingredients are developed below. 

We first unfold the network $\mathcal{G}$ over time to create a
layered network. The idea is to unfold the network to $K$ stages such
that i-th stage represents what happens in the network during $(i-1)T$
to $iT-1$ symbol times. More concretely, the $K$ time-steps unfolded
network, $\mathcal{G}_{\text{unf}}^{(K)}=( \mathcal{V}_{\text{unf}}^{(K)},  \mathcal{E}_{\text{unf}}^{(K)})$, is constructed from
$\mathcal{G}=(\mathcal{V},\mathcal{E})$ as follows:
\begin{itemize}
\item The network has $K+2$ stages (numbered from 0 to $K+1$)
\item Stage 0 has only node $S[0]$ and stage $K+1$ has only node
  $D[K+1]$. $S[0]$ and $D[K+1]$ respectively represent the source and
  the destination in $\mathcal{G}_{\text{unf}}^{(K)}$.
\item  Each node $v \in \mathcal{V}$ appears at stage $i$ as a relay denoted by $v[i]$, $i=1,\ldots,K$.
\end{itemize}

Also, the links in $\mathcal{G}_{\text{unf}}^{(K)}$ are as follows
\begin{itemize}
\item There are wired links (i.e., links that are orthogonal to all
  other transmissions in the network) of capacity $K \overline{C}$,
  where $\overline{C}=\min_{\Omega\in\Lambda_D}
\mathrm{rank}(\Gbf_{\Omega,\Omega^c})$ is the min-cut value of $\mathcal{G}$, between
\begin{enumerate}
\item $(S[0],S[1])$ and $(D[K],D[K+1])$
\item $(v[i],v[i+1])$, for all $v \in \mathcal {V}$ and $1 \leq i < K$,
\end{enumerate}
\item Node $v[i]$ is connected to node $w[i+1]$ with the linear
  finite-field deterministic channel of the original network
  $\mathcal{G}$, for all $(v,w) \in \mathcal{E}$, $v \neq w$.
\end{itemize}
The transmit vector of node $v[i] \in \mathcal{V}_{\text{unf}}^{(K)}$
is denoted by the pair $(\xbf_{v[i]}^{(1)},\xbf_{v[i]}^{(2)})$, where
$\xbf_{v[i]}^{(1)}\in \FF_2^{K\overline{C}}$ and $\xbf_{v[i]}^{(2)}\in
\FF_2^q$ ($q$ is the size of the vectors in the linear finite field
network) are respectively the inputs of the wired and the linear
finite-field channels. Intuitively, the wired channels represent the
memory at each node. Furthermore, the cut-set bound on the capacity of  $\mathcal{G}_{\text{unf}}^{(K)}$ is denoted by $\overline{C}_{\text{unf}}^{(K)}$, i.e.
\beq \label{eq:cutSetUnfold} \overline{C}_{\text{unf}}^{(K)}= \min_{\Omega_{\text{unf}}\in
  \Lambda_{\text{unf}}}
\mathrm{rank}(\Gbf_{\Omega_{\text{unf}},\Omega_{\text{unf}}^c}), \eeq
where the minimum is taken over all
cuts $\Omega_{\text{unf}}$ in $\mathcal{G}_{\text{unf}}^{(K)}$.

For example in Figure \ref{fig:unfolding} (a) a network with unequal
paths from $S$ to $D$ is shown. Figure \ref{fig:unfolding}(b) shows
the unfolded form of this network.

We now prove the following lemma.

\begin{lemma}
\label{lem:UnfLinDet} Any communication rate $R<\frac{1}{K}  \overline{C}_{\text{unf}}^{(K)}$ is achievable in $\mathcal{G}$, where $\overline{C}_{\text{unf}}^{(K)}$ is defined in (\ref{eq:cutSetUnfold}).
\end{lemma}

\begin{proof}
  Note that $\mathcal{G}_{\text{unf}}^{(K)}$ is a layered linear
  finite-field network. Therefore, by our result of Section
  \ref{subsec:LinDet}, we can achieve any rate
  $R_{\text{unf}}<\overline{C}_{\text{unf}}^{(K)}$ in
  $\mathcal{G}_{\text{unf}}^{(K)}$. In particular, it is achieved by
  the encoding strategy described in Section
  \ref{subsec:EncLinDetLay}, in which each node $v[i] \in
  \mathcal{V}_{\text{unf}}^{(K)}$, $i=1,\ldots,K$, operates over
  blocks of size $T$ symbols and transmits $\xbf_{v[i]}^{(1)} =
  \Fbf_{v[i]}^{(1)} \ybf_{v[i]} $ and $\xbf_{v[i]}^{(2)} =
  \Fbf_{v[i]}^{(2)} \ybf_{v[i]} $ respectively over the wired and the
  linear finite-field channels.

Now, we can implement the scheme in $\mathcal{G}$ by  using $K$ blocks of size $T$ symbols. The construction is as follows:
\begin{itemize}
\item The source $S$ transmits $\xbf_{S[i]}^{(2)}$ at block $i$, $i=1,\ldots,K$,
\item Each node $v \in \mathcal{V}$, $v \notin \{S,D\}$, transmits $\xbf_{v[i]}^{(2)}$ and puts  $\xbf_{v[i]}^{(1)}$ in its memory at block $i$, $i=1,\ldots,K$ (note that this is possible, because  $\xbf_{v[i]}^{(1)}$ and $\xbf_{v[i]}^{(2)}$ are only a function of the received signal at node $v$ in the previous block and the the signal stored in the memory of node $v$ at the beginning of block $i$).
\end{itemize}
Finally, the destination decodes based on  $\xbf_{D[K]}^{(1)}$, which is a function of  the received signal at the destination during the $K$ blocks. Therefore, the rate $\frac{1}{K}R_{\text{unf}}$ is achievable in $\mathcal{G}$ and  the proof is complete.
\end{proof}

Now, if we show that $ \lim_{K \rightarrow \infty }\frac{1}{K} \overline{C}_{\text{unf}}^{(K)} = \overline{C}$, then by using Lemma \ref{lem:UnfLinDet}, the proof of Theorem \ref{thm:LinDetNet} will be complete. We will show this next.

\begin{lemma}
\label{lem:minCutConnection}
\beq \lim_{K \rightarrow \infty }\frac{1}{K} \overline{C}_{\text{unf}}^{(K)} =\overline{C}, \eeq
where $\overline{C}=\min_{\Omega\in\Lambda_D}
\mathrm{rank}(\Gbf_{\Omega,\Omega^c})$ and $\overline{C}_{\text{unf}}^{(K)}$ is defined in (\ref{eq:cutSetUnfold}).
\end{lemma}
\begin{proof}
Any cut $\Omega_{\text{unf}} \in\Lambda_{\text{unf}}$ is a subset of nodes in $\mathcal{G}_{\text{unf}}^{(K)}$ such that $S[0] \in \Omega_{\text{unf}}$ and $D[K+1] \in \Omega_{\text{unf}}^c$. Now for any cut $\Omega_{\text{unf}}$ we define 
\beq \label{eq:cutIncRel} \mathcal{V}[i] = \{v \in \mathcal{V} | v[i] \in \Omega_{\text{unf}} \}, \quad i=0,\ldots K+1. \eeq
In other words, $\mathcal{V}[i]$ is the set of nodes of $\mathcal{G}$ such that at stage $i$ they appear in $\Omega_{\text{unf}}$.

Every cut $\Omega\in\Lambda_D$ in the original network $\mathcal{G}$
corresponds to a cut in the unfolded network
$\mathcal{G}_{\text{unf}}^{(K)}$, by choosing
$\mathcal{V}[1]=\cdots=\mathcal{V}[K]=\Omega$.   Also, the value of
such a ``steady'' cut is
$K\mathrm{rank}(\Gbf_{\Omega,\Omega^c})$, thereby
\beq \label{eq:need3}  \overline{C}_{\text{unf}}^{(K)}\leq K
\overline{C}. \eeq
Therefore, we need to only focus on cuts whose values
are smaller than $ K \overline{C}$.  We will next identify other cuts which
have value larger than $ K \overline{C}$, in order to reduce the set
of cuts to consider for $\overline{C}_{\text{unf}}^{(K)}$.

We claim that the value of any cut $\Omega_{\text{unf}}
\in\Lambda_{\text{unf}}$ is at least $K \overline{C}$, 
if the following is {\em not} satisfied:
\beq 
\label{eq:IncrSubsets}
\mathcal{V}[1] \subseteq \mathcal{V}[2] \subseteq \cdots \subseteq \mathcal{V}[K]
 \eeq

The reason is that if $\mathcal{V}[1] \subseteq \mathcal{V}[2]
\subseteq \cdots \subseteq \mathcal{V}[K]$ is {\em not} true, then there
exists a node $v \in \mathcal{V}$ and a stage $j$ ($1 \leq j < K$)
such that
\beq v[j] \in \mathcal{V}[j]  \quad \text{and} \quad v[j+1] \notin \mathcal{V}[j+1].\eeq
If this happens, then the edge $(v[j],v[j+1])$, which has capacity $K
\overline{C}$, traverses from $\Omega_{\text{unf}}$ to
$\Omega_{\text{unf}}^c$, hence the cut-value (i.e.,
$\mathrm{rank}(\Gbf_{\Omega_{\text{unf}},\Omega_{\text{unf}}^c})$)
becomes at least $K \overline{C}$.

Hence, we only need to focus on cuts, $\Omega_{\text{unf}}$ that satisfy
\eqref{eq:IncrSubsets}, {\em i.e.,} contain
an increasing set of nodes at the stages. Since there are total of
$|V|$ nodes in $\mathcal{G}$, we can have at most $|V|$ transitions in
the size of $\mathcal{V}[i]$s. Now, using the notation in Figure
\ref{fig:ordering} and the fact that the network is layered, for any
cut $\Omega_{\text{unf}} \in\Lambda_{\text{unf}}$ satisfying
(\ref{eq:IncrSubsets}) we can write
\begin{align*}
& \mathrm{rank}(\Gbf_{\Omega_{\text{unf}},\Omega_{\text{unf}}^c}) = \sum_{i=0}^K \mathrm{rank}(\Gbf_{\mathcal{V}[i],\mathcal{V}[i+1]^c}) \\
& = \sum_{i=1}^{|V|} (\ell_i-1) \mathrm{rank}(\Gbf_{\mathcal{U}_i,\mathcal{U}_i^c})
+ \sum_{i=1}^{|V|-1}\mathrm{rank}(\Gbf_{\mathcal{U}_i,\mathcal{U}_{i+1}^c}) + \\ & \quad \quad \mathrm{rank}(\Gbf_{\mathcal{V}[0],\mathcal{U}_{1}^c})+\mathrm{rank}(\Gbf_{\mathcal{U}_K,\mathcal{V}[K+1]^c})\\
& \geq  \sum_{i=1}^{|V|} (\ell_i-1) \mathrm{rank}(\Gbf_{\mathcal{U}_i,\mathcal{U}_i^c})\\
& \stackrel{(a)}{\geq}  (\sum_{i=1}^{|V|} (\ell_i-1)) \overline{C} \\
&\stackrel{(b)}{=} (K-|V|) \overline{C},  
\end{align*}
where $(a)$ follows because $\mathrm{rank}(\Gbf_{\Omega,\Omega^c})\geq
\overline{C}$, for
any cut $\Omega$ in $\mathcal{G}$; and $(b)$ is because there are at most
$|V|$ transitions implying that $\sum_{i=1}^{|V|} (\ell_i-1))=(K-|V|)$.
As a result,
\beq \label{eq:need4} 
\overline{C}_{\text{unf}}^{(K)}\geq (K-|V|)\overline{C}. 
\eeq

Combining (\ref{eq:need3}) and (\ref{eq:need4}), we get 
\beq \lim_{K \rightarrow \infty } \frac{1}{K} \overline{C}_{\text{unf}}^{(K)} =
  \overline{C}.\eeq
\end{proof}

\begin{figure}
     \centering 
      \scalebox{0.45}{ \includegraphics{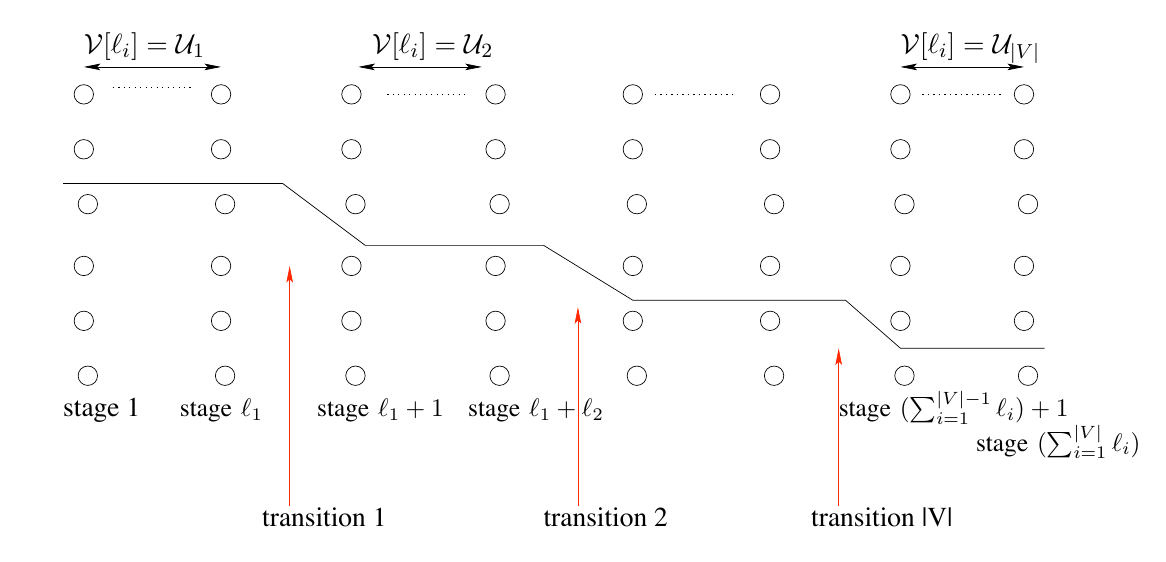}}
\caption{Illustration of cuts in $\mathcal{G}_{\text{unf}}^{(K)}$ which can have a value smaller than $K \overline{C}$.\label{fig:ordering}}
\end{figure}

This combined with Lemma \ref{lem:UnfLinDet} completes the proof of
Theorem \ref{thm:LinDetNet}\footnote{An alternate proof of the same result was
given in \cite{ADTAllerton07P2}. In that proof, only the previous received block was used
by the relays, instead of the larger number of blocks used above. However, we needed
to use the sub-modularity properties of entropy to demonstrate the performance of that
scheme \cite{ADTAllerton07P2}.}.

\section{Gaussian relay networks}
\label{sec:GaussCapacity}
So far, we have focused on deterministic relay networks. As we illustrated
in Sections \ref{sec:model} and \ref{sec:motivation}, linear
finite-field deterministic model captures some (but not all) aspects
of the high $\SNR$ behavior of the Gaussian model. Therefore we have
some hope to be able to translate the intuition and the techniques
used in the deterministic analysis to obtain approximate results for
 Gaussian relay networks. This is what we will accomplish in this
section.

Theorem \ref{thm:GaussMain} is the main result for Gaussian relay
networks and this section is devoted to proving it.  The proof of the
result for layered network is done in Section \ref{sec:Lay}. We extend
the result to an arbitrary network by expanding the network over time,
as done in Section \ref{sec:detCapacity}.  We first prove the theorem
for the single antenna case, then at the end we extend it to the
multiple antenna scenario.

\subsection{Layered Gaussian relay networks}
\label{sec:Lay}
In this section we prove Theorem \ref{thm:GaussMain} for the special
case of layered networks, where all paths from the source to the
destination in $\mathcal{G}$ have equal length. 

\subsubsection{Proof illustration}
\label{subsec:PfIdea}
Our proof has two steps. In the first step we propose a relaying
strategy, which is similar to our strategy for deterministic networks, and
show that by operating over a large block, it is possible to achieve
an end-to-end mutual information which is within a constant gap to the
cut-set upper bound. Therefore, the relaying strategy creates an inner code which provides certain end-to-end mutual
information between the transmit signal at the source and the received signal at the destination. Each symbol of
this inner code is a block. In the next step, we use an outer code to
map the message to multiple inner code symbols and send them to the
destination. By coding over many such symbols, it is possible to
achieve a reliable communication rate arbitrarily close to the mutual
information of the inner code, and hence the proof is complete. The
system diagram of our coding strategy is illustrated in Figure
\ref{fig:proofDiag}.

\begin{figure}
     \centering

       \scalebox{0.5}{\includegraphics{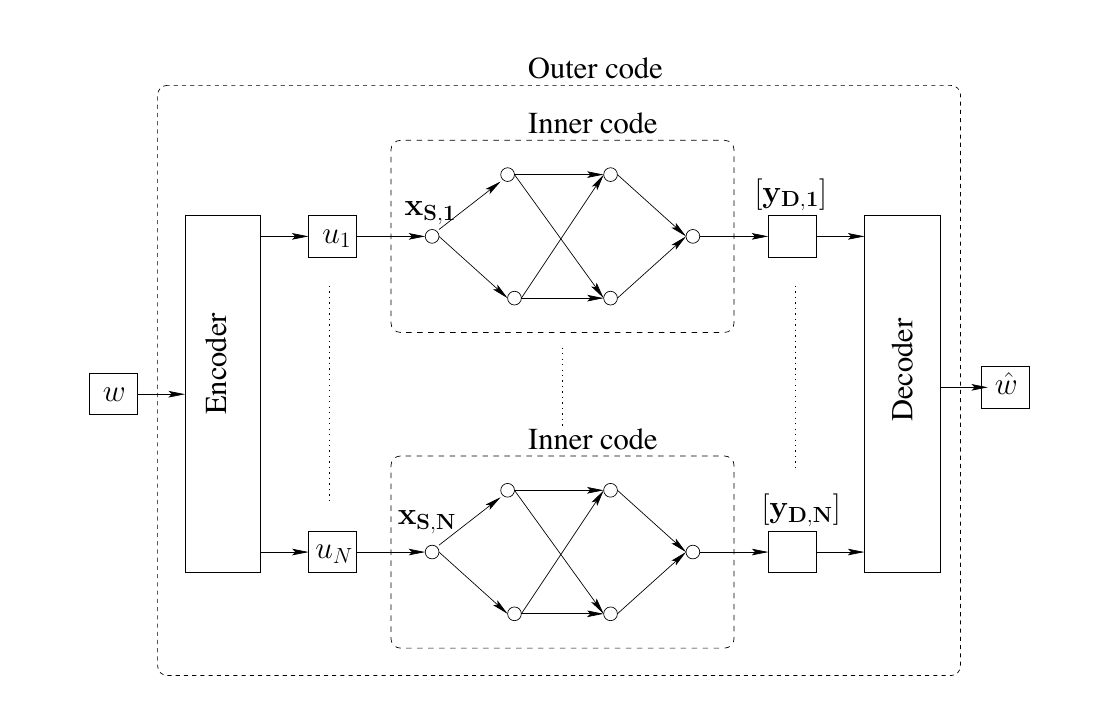} }
 \caption{System diagram. \label{fig:proofDiag}}
\end{figure}

We now explicitly describe our encoding strategy

\subsubsection{Encoding for layered Gaussian relay networks}
\label{subsec:EncLay}

We first define a quantization operation.
\begin{defn}
\label{defn:Quantization}
The quantization operation $[.]: \CC \rightarrow \ZZ \times \ZZ$ maps
a complex number $c=x+iy$ to $[c]=([x], [y])$, where $[x]$ and $[y]$
are the closest integers to $x$ and $y$, respectively. Since the Gaussian noise at all receive antennas has variance 1, this operation is basically scalar quantization at noise-level.
\end{defn}

As shown in Figure \ref{fig:proofDiag}, the encoding consists of an
inner code and an outer code:
\paragraph{Inner code} Each symbol of the inner code is represented by
$u\in\{1,\ldots,2^{R_{\text{in}}T}\}$, where $T$ and $R_{\text{in}}$
are respectively the block length and the rate of the inner code. The
source node $S$ generates a set of $2^{R_{\text{in}}T}$ independent
complex Gaussian codewords of length $T$ with components distributed
as i.i.d. $\mathcal{CN}(0,1)$, denoted by
$\mathcal{T}_{x_S}$. At relay node $i$, there is also a random mapping $F_i:
(\ZZ^{T},\ZZ^{T}) \rightarrow \mathcal{T}_{x_i}$ which maps each
quantized received signal vector of length $T$ independently into an
i.i.d. $\mathcal{CN}(\mathbf{0},1)$ random vector of
length $T$. A particular realization of $F_i$ is
denoted by $f_i$. Summarizing:
\begin{itemize}
\item Source: maps each inner code symbol
  $u\in\{1,\ldots,2^{R_{\text{in}}T}\}$ to $F_S(u) \in
  \mathcal{T}_{x_S}$.
\item Relay $i$: receives $\ybf_i$ of length $T$. Quantizes it to
  $[\ybf_i]$. Then maps it to $F_i([\ybf_i]) \in\mathcal{T}_{x_i}$.
\end{itemize}

\paragraph{Outer code}
The message is encoded by the source into $N$ inner
code symbols, $u_1,\ldots,u_N$. Each inner code symbol is then sent
via the inner code over $T$ transmission times, giving an overall
transmission rate of $R$. The received signal at the destination, corresponding to inner code symbol $u_i$, is denoted by $\ybf_{D,i}$, $i=1,\ldots,N$.

Now, given the knowledge of all the encoding functions $F_i$'s at the
relays and quantized received signals
$[\ybf_{D,1}],\ldots,[\ybf_{D,N}]$, the destination attempts to decode the
message sent by the source.

\subsubsection{Proof of Theorem \ref{thm:GaussMain} for layered networks}
\label{subsec:ProofLay}

Our first goal is to lower bound the average end-to-end mutual information, averaged over the random mappings
$F_\mathcal{V}=\{F_i: i \in \mathcal{V}\}$, achieved by the inner code
defined in Subsection \ref{subsec:EncLay}.

Note that
\beq \label{eq:MILB} \frac{1}{T} I(u
  ;[\ybf_D]|F_\mathcal{V})  \geq  \frac{1}{T} I(u ;[\ybf_D]
  |\zbf_{\mathcal{V}},F_\mathcal{V}) - \frac{1}{T}
  H([\ybf_D]|u,F_\mathcal{V}) \eeq where
$\zbf_{\mathcal{V}}$ is the vector of the channel noises at all
nodes in the network. The first term on the right hand side of
(\ref{eq:MILB}) is the average end-to-end mutual information
conditioned on the noise vector. Once we condition on a noise
vector, the network turns into a deterministic network. We use an
analysis technique similar to the one we used for linear deterministic
relay networks to upper bound the probability that the destination
will confuse an inner code symbol with another and then use Fano's
inequality to lower bound the end-to-end mutual information. This is
done in Lemma \ref{lem:MainScalar}. The second term on the RHS of
(\ref{eq:MILB}) is the average entropy of the received signal
conditioned on the source's transmit signal, and is upper bounded in
Lemma \ref{lem:condMI}. This term represents roughly the penalty due
to noise-forwarding at the relay, and is proportional to the number of
relay nodes.

\begin{defn} \label{def:cutSetiid}We define \beq \overline{C}_{i.i.d.}
  \stackrel{\triangle}{=} \min_{\Omega}
  I(x_{\Omega};y_{\Omega^c}|x_{\Omega^c}) \eeq where $x_i$, $i \in
  \mathcal{V}$, are i.i.d.  $\mathcal{CN}(0,1)$ random variables.
\end{defn}

\begin{lemma} \label{lem:MainScalar} Assume all nodes perform the operation
  described in subsection \ref{subsec:EncLay} (a) and the inner code
  symbol $U$ is distributed uniformly over
  $\{1,\ldots,2^{R_{\text{in}}T}\}$. Then
  \begin{align*} & I(u ;[\ybf_D] |\zbf_{\mathcal{V}},F_\mathcal{V})
    \geq R_{\text{in}} T- \\ & \quad (1+\min \{1, 2^{|\mathcal{V}|}
    2^{-T(\overline{C}_{iid}-|\mathcal{V}| -R_{\text{in}})}
    \}R_{\text{in}}T) \end{align*} where $\overline{C}_{iid}$ is
  defined in Definition \ref{def:cutSetiid}.

\end{lemma}
\begin{proof}

Consider a fixed noise realization in the network $\zbf_{\mathcal{V}}=\abf$. Suppose the destination attempts to detect the transmitted symbol $u$ at the source given the received signal, all the mappings,
channel gains, and $\abf$. A symbol value $u$ will be mistaken for
  another value $u'$ only if the received signal $[\ybf_D(u)]$ under
  $u$ is the same as what would have been received under $u'$. This
  leads to a notion of {\em distinguishability} for a fixed $\abf$,
  which is that symbol values $u,u'$ are distinguishable at any node
  $j$ if $[\ybf_j(u)]\neq [\ybf_j(u')]$. Hence,
{\footnotesize
\begin{align}
\nonumber & \prob{u\rightarrow u'|\zbf_{\mathcal{V}}=\abf} = \sum_{\Omega\in\Lambda_D}\\
&  \label{eq:PwErrA}  \underbrace{\prob{\mbox{Nodes in } \Omega
\mbox{ can distinguish } u,u' \mbox{ and nodes in } \Omega^c \mbox{
cannot} | \zbf_{\mathcal{V}}=\abf}}_{\mathcal{P}} ~
\end{align}
}
For any cut $\Omega \in \Lambda_D$, define the following sets:
\begin{itemize}
\item $L_l (\Omega)$: the nodes that are in $\Omega$ and are at layer
  $l$ (for example $S \in L_1 (\Omega)$),
\item $R_l (\Omega)$: the nodes that are in $\Omega^c$ and are at
  layer $l$ (for example $D \in R_{l_D} (\Omega)$).
\end{itemize}
We also define the following events:
\begin{itemize}
\item $\mathcal{L}_l$: Event that the nodes in $L_l$ can distinguish
  between $u$ and $u'$, \emph{i.e.}, $[\ybf_{L_l}(u)] \neq
  [\ybf_{L_l}(u')]$,
\item $\mathcal{R}_l$: Event that the nodes in $R_l$ can not
  distinguish between $u$ and $u'$, \emph{i.e.}, $[\ybf_{R_l}(u)] =
  [\ybf_{R_l}(u')]$.
\end{itemize}

Note that the source node by definition distinguishes between the two distinct messages $u,u'$,
\emph{i.e.} $\mathbb{P}\{\mathcal{L}_{1} \}=1$.
\begin{eqnarray}
\nonumber  \mathcal{P} &=& \mathbb{P} \{ \mathcal{R}_l,\mathcal{L}_{l-1} ,l=2,\ldots,l_D | \zbf_{\mathcal{V}}=\abf \} \\
\nonumber  &=& \prod_{l=2}^{l_D}  \mathbb{P} \{\mathcal{R}_l,\mathcal{L}_{l-1}|\mathcal{R}_j,\mathcal{L}_{j-1}, j=2,\ldots,l-1,\zbf_{\mathcal{V}}=\abf\} \\
\nonumber  &\leq& \prod_{l=2}^{l_D}  \mathbb{P} \{\mathcal{R}_l|\mathcal{R}_j,\mathcal{L}_{j}, j=2,\ldots,l-1,\zbf_{\mathcal{V}}=\abf\} \\
\nonumber  &\stackrel{(a)}{=}& \prod_{l=2}^{l_D}  \mathbb{P} \{\mathcal{R}_l|\mathcal{R}_{l-1},\mathcal{L}_{l-1},\zbf_{\mathcal{V}}=\abf\} \\
\label{eq:ineqPr1}  &=& \prod_{l=2}^{l_D}   \mathbb{P} \{[\ybf_{R_l}(u)] = [\ybf_{R_l}(u')]|\mathcal{R}_{l-1},\mathcal{L}_{l-1},\zbf_{\mathcal{V}}=\abf\} \quad \quad
 \end{eqnarray}
 where (a) is true due to the Markov structure in the layered network.

 Note that if $\Abf$ and $\Bbf$ are complex $m\times n$ matrices, then
 \beq 
\label{eq:quantS1}
[\Abf_{i,j}]=[\Bbf_{i,j}], \forall i,j\,\, \Rightarrow ||\Abf-\Bbf||_{\infty} \leq
 \sqrt{2}n.
\eeq 
Therefore by (\ref{eq:ineqPr1}) and (\ref{eq:quantS1}) we have
{\small
 \begin{eqnarray} \nonumber \mathcal{P} & \leq&   \prod_{l=2}^{l_D}   \mathbb{P} \{|| \ybf_{R_l}(u)- \ybf_{R_l}(u')||_{\infty} \leq \sqrt{2} |\mathcal{R}_{l-1},\mathcal{L}_{l-1},\zbf_{\mathcal{V}}=\abf\} \\
   \label{eq:ineqPr11}&\stackrel{(a)}{=}& \prod_{l=2}^{l_D}   \mathbb{P} \{|| \ybf_{R_l}(u)- \ybf_{R_l}(u')||_{\infty} \leq \sqrt{2} |\mathcal{R}_{l-1},\mathcal{L}_{l-1}\}\end{eqnarray}}
 where (a) is true since conditioned on $\mathcal{R}_{l-1},\mathcal{L}_{l-1}$ the distribution of $ \ybf_{R_l}(u)- \ybf_{R_l}(u')$ does not depend on the noise (due to the random mapping). 

By defining $\Hbf_l$ to be the transfer matrix from the left side of the cut
at stage $l-1$ to the right side of the cut at stage $l$ (\emph{i.e.},
the MIMO channel from $L_{l-1}$ to $R_l$), we have

{\footnotesize
\begin{align} 
\nonumber &  \mathcal{P}  \stackrel{(\ref{eq:ineqPr11})}{\leq}
  \prod_{l=2}^{l_D}   \mathbb{P} \{|| \ybf_{R_l}(u)- \ybf_{R_l}(u')||_{\infty} \leq 
\sqrt{2} |\mathcal{R}_{l-1},\mathcal{L}_{l-1}\} \\
\nonumber  &= \prod_{l=2}^{l_D}   \mathbb{P} \{\forall 1 \leq j \leq T: || \Hbf_l 
\lp \xbf_{L_{l-1},j}(u)- \xbf_{L_{l-1},j}(u')\rp||_{\infty} 
\leq \sqrt{2} |\mathcal{R}_{l-1},\mathcal{L}_{l-1}\} \\
\label{eq:quantBound1} 
&\stackrel{(b)}{=} \prod_{l=2}^{l_D}   \mathbb{P} \{\forall 1 \leq j \leq T: 
|| \Hbf_l \lp \xbf_{L_{l-1},j}(u)- \xbf_{L_{l-1},j}(u')\rp||_{\infty} 
\leq \sqrt{2} |\mathcal{L}_{l-1}\}
\end{align}
}
where (b) is true since the nodes in $R_{l-1}(\Omega)$ transmit the
same codeword under both $u$ and $u'$.

Since $\xbf_{L_{l-1}}(u) \neq \xbf_{L_{l-1}}(u')$, due to the random mapping, $\xbf_{L_{l-1}}(u)$ and $\xbf_{L_{l-1}}(u')$ are two independent random vectors with i.i.d. $\mathcal{CN}(0,1)$ elements. Therefore, their difference is a random vector with i.i.d. $\mathcal{CN}(0,2)$ elements. Now, we state the following Lemma
which is proved in Appendix \ref{app:proofLemQuantScalar}.

\begin{lemma} \label{lem:quantScalar} Assume
  $[\tilde{x}_{i,1},\cdots,\tilde{x}_{i,T}]$, $i=1,\ldots,m$, are
  i.i.d. vectors of length $T$ with i.i.d. $\mathcal{CN}(0,2)$ elements, and $\Hbf \in
  \CC^{n \times m}$ is an $n \times m$ matrix. Then 
  \begin{align} \nonumber & \prob{ \forall
    1 \leq j \leq T: ||\Hbf[\tilde{x}_{1,j},\cdots,\tilde{x}_{m,j}]^t||_{\infty} \leq    \sqrt{2}} \leq  \\ &\quad \quad 2^{-T\lp I \lp \xbf; \Hbf\xbf+\zbf\rp-\min(m,n) \rp } 
    \end{align} where $\xbf$ and $\zbf$ are i.i.d.
  complex unit variance Gaussian vectors of length $m$ and $n$
  respectively.
\end{lemma}

By applying Lemma \ref{lem:quantScalar} to (\ref{eq:quantBound1}) we get
{\footnotesize
\begin{align} \nonumber & \mathbb{P} \{\forall 1 \leq j \leq T: ||
  \Hbf_l \lp \xbf_{L_{l-1},j}(u)- \xbf_{L_{l-1},j}(u')\rp||_{\infty} \leq
  \sqrt{2} |\mathcal{L}_{l-1}\} \leq  \\ &\quad \quad 2^{-T\lp I \lp
    \xbf_{L_{l-1}}; \ybf_{R_l}|\xbf_{R_{l-1}} \rp -\min \lp |L_{l-1}| ,|R_l|
    \rp  \rp } \end{align}
}
where $x_i$, $i \in \mathcal{V}$, are i.i.d. with Gaussian
distribution. Hence
{\footnotesize
\beq \mathcal{P}   \leq  \prod_{l=2}^{l_D} 2^{-T\lp I \lp \xbf_{L_{l-1}}; \ybf_{R_l}|\xbf_{R_{l-1}} \rp -\min \lp |L_{l-1}| ,|R_l| \rp   \rp }  \leq 2^{-T(\overline{C}_{iid}-|\mathcal{V}| )}
\eeq }
where $\overline{C}_{iid}$ is defined in Definition
\ref{def:cutSetiid}.

The average probability of symbol detection error at the
destination can be upper bounded as \beq P_e=\prob{ \hat{u} \neq u |
  \zbf_{\mathcal{V}}=\abf } \leq 2^{R_{\text{in}}T} \prob{u\rightarrow
  u'|\zbf_{\mathcal{V}}=\abf}.  \eeq By the union bound we have \beq
P_e \leq \sum_{\Omega} 2^{-T(\overline{C}_{iid}-|\mathcal{V}|-R_{\text{in}})} \leq 2^{|\mathcal{V}|}
2^{-T(\overline{C}_{iid}-|\mathcal{V}|-R_{\text{in}})}.  \eeq Now, using Fano's inequality we get
\begin{align*}
&  I(u ;[\ybf_D] |\zbf_{\mathcal{V}}=\abf,F_\mathcal{V}) = H(u)-H(u|[\ybf_D] , \zbf_{\mathcal{V}}=\abf,F_\mathcal{V}) \\
  &= R_{\text{in}}T-H(u|[\ybf_D] , \zbf_{\mathcal{V}}=\abf,F_\mathcal{V}) \\
  &= R_{\text{in}}T-\EE_{F_\mathcal{V}} [H(u|[\ybf_D] , \zbf_{\mathcal{V}}=\abf,F_\mathcal{V}=f_\mathcal{V})] \\
  & \stackrel{\text{Fano}}{\geq}  R_{\text{in}}T- (1+\EE_{F_\mathcal{V}}[\prob{ \hat{u} \neq u | \zbf_{\mathcal{V}}=\abf ,F_\mathcal{V}=f_\mathcal{V}) } ] R_{\text{in}}T) \\
  & =  R_{\text{in}}T- (1+P_e R_{\text{in}}T) \\
  & \geq  R_{\text{in}}T- (1+\min \{1,  2^{|\mathcal{V}|} 2^{-T(\overline{C}_{iid}-|\mathcal{V}| )-R_{\text{in}})}  \}R_{\text{in}}T) \hspace{.5in}
 \end{align*}
Hence, the proof is complete.
\end{proof}

The following lemma, which is proved in Appendix \ref{app:proofLemCondMI}, bounds the second term on the RHS of (\ref{eq:MILB}).
\begin{lemma} \label{lem:condMI} Assume all nodes perform the operation
  described in subsection \ref{subsec:EncLay} (a). Then \beq
  H([\ybf_D]|u,F_\mathcal{V})] \leq 12T |\mathcal{V}| \eeq
\end{lemma}

The next lemma, which is proved in Appendix \ref{app:proofLemBeamForming}, bounds the gap between $\overline{C}$ and  $\overline{C}_{iid}$.
\begin{lemma} \label{lem:BeamForming} For a Gaussian relay network
  $\mathcal{G}$, \beq \overline{C}-\overline{C}_{iid}< 2|\mathcal{V}|
  \eeq where $\overline{C}$ is the cut-set upper bound on the capacity
  of $\mathcal{G}$ and $\overline{C}_{iid}$ is defined in Definition
  \ref{def:cutSetiid}.
\end{lemma}

Finally, using Lemmas  \ref{lem:MainScalar}, \ref{lem:condMI} and \ref{lem:BeamForming}, we have
\begin{lemma} \label{lem:mainGauss} Assume all nodes perform the operation
  described in subsection \ref{subsec:EncLay} (a) and the inner code
  symbol $U$ is distributed uniformly over
  $\{1,\ldots,2^{R_{\text{in}}T}\}$. Then

  {\footnotesize
\begin{equation}\frac{1}{T}  I(u ;[\ybf_D]|F_\mathcal{V}) \geq  R_{\text{in}} - 12 |\mathcal{V}| - (\frac{1}{T}+\min \{1,  2^{|\mathcal{V}|} 2^{-T(\overline{C}-3|\mathcal{V}| -R_{\text{in}})}  \}R_{\text{in}}) \end{equation} }
where $\overline{C}$ is the cut-set upper bound on the capacity of
$\mathcal{G}$.
\end{lemma}
\begin{proof}
By using Equation (\ref{eq:MILB}) and Lemmas \ref{lem:MainScalar}, \ref{lem:condMI} and \ref{lem:BeamForming}, we have
\begin{align*}
  & \frac{1}{T}  I(u ;[\ybf_D]|F_\mathcal{V})  \geq  \frac{1}{T} I(u ;[\ybf_D] |\zbf_{\mathcal{V}},F_\mathcal{V}) - \frac{1}{T}  H([\ybf_D]|u,F_\mathcal{V}) \\
  &  \stackrel{\text{Lemma \ref{lem:MainScalar} and \ref{lem:condMI} }}{\geq} \frac{1}{T} ( R_{\text{in}} T- [1+\min \{1,  2^{|\mathcal{V}|} 2^{-T(\overline{C}_{iid}-|\mathcal{V}|) -R_{\text{in}})}  \}\\ & \quad \quad \quad \quad \quad \quad \quad R_{\text{in}}T]- \quad 12T |\mathcal{V}| ) \\
  & \stackrel{\text{Lemma \ref{lem:BeamForming}} }{\geq} R_{\text{in}} - 12 |\mathcal{V}| - (\frac{1}{T}+\min \{1,  2^{|\mathcal{V}|} 2^{-T(\overline{C}-3|\mathcal{V}| -R_{\text{in}})}  \}R_{\text{in}}).
\end{align*}
\end{proof}

An immediate corollary of this lemma is that by choosing
$R_{\text{in}}$ arbitrarily close to $\overline{C}-2|\mathcal{V}|$, and letting $T$ be arbitrary large,
for any $\delta>0$ we get
\begin{equation} \label{eq:endToEndMI}\frac{1}{T}  I(u ;[\ybf_D]|F_\mathcal{V}) \geq  \overline{C}-15 |\mathcal{V}| - \delta. \end{equation}

Therefore, there exists a choice of mappings that provides an
end-to-end mutual information  close to
$\overline{C}-15|\mathcal{V}| $. Hence,
we have  created a point-to-point channel from $u$ to
$[\ybf_D]$ with at least this mutual information. We can now use a good
outer code to reliably send a message over $N$ uses of this channel
(as illustrated in Figure \ref{fig:proofDiag}) at any rate up to $\overline{C}-15|\mathcal{V}|$.

Hence we get an intermediate proof of Theorem \ref{thm:GaussMain} for
the special case of layered Gaussian relay networks, with single
antennas in the network.  This is stated below for convenience, and
its generalization to arbitrary networks with multiple antennas is
given in Section \ref{sec:Gen}.

\begin{theorem} \label{thm:MainLay} Given a Gaussian relay network
  $\mathcal{G}$ with a layered structure and single antenna at each
  node, all rates $R$ satisfying the following condition are
  achievable, \beq R < \overline{C}-\kappa_{\text{Lay}} \eeq where
  $\overline{C}$ is the cut-set upper bound on the capacity of
  $\mathcal{G}$ as described in Equation (\ref{eq:CutSetRef}),
  $\kappa_{\text{Lay}}=15|\mathcal{V}|$
  is a constant not depending on the channel gains.
\end{theorem}

\subsubsection{Vector quantization and network operation}
\label{subsec:TypPf}

The network operation can easily be generalized to include vector
quantization at each node.  Each node in the network generates a
transmission Gaussian codebook of length $T$ with components
distributed as i.i.d. $\mathcal{CN}(0,1)$.  The source operation is as
before, it produces a random mapping from messages
$w\in\{1,\ldots,2^{RT}\}$ to its transmit codebook
$\mathcal{T}_{x_S}$. We denote this codebook by
$\xbf_S^{(w)},w\in\{1,\ldots,2^{TR}\}$. Each received sequence
$\ybf_i$ at node $i$ is quantized to $\hat{\ybf}_i$ through a Gaussian
vector quantizer, with quadratic distortion set to the
noise-level. This quantized sequence is randomly mapped onto a
transmit sequence $\xbf_i$ using a random function
$\xbf_i=f_i(\hat{\ybf}_i)$.  This mapping as before is chosen such
that each quantized sequence is mapped uniformly at random to a
transmit sequence.  These transmit sequences are chosen to be in
$\mathcal{T}_{x_i}$, which are i.i.d. Gaussian $\mathcal{CN}(0,1)$.
We denote the $2^{TR_i}$ sequences of $\hat{\ybf}_i$ as
$\hat{\ybf}_i^{(k_i)}, k_i\in\{1,\ldots,2^{TR_i}\}$. Standard
rate-distortion theory tells us that we need $R_i>I(Y_i;\hat{Y}_i)$
for this quantization to be successful, where the reconstruction is
chosen such that the quadratic distortion is at the
noise-level\footnote{Note that we can be conserative and assume the
maximal received power, depending on the maximal channel gains. Since
we do not directly convey this quantization index, but just map it
forward, this conservative quantization suffices.}.  Since the uniform
random mapping produces $\xbf_i=f_i(\hat{\ybf}_i)$, for a quantized
value of index $k_i$, we will denote it by $\hat{\ybf}_i^{(k_i)}$ and
the sequence it is mapped to by
$\xbf_i^{(k_i)}=f_i(\hat{\ybf}_i^{(k_i)})$.  At the destination, we
can either employ a maximum-likelihood decoder (for which the mutual
information is evaluated), or a typicality decoder (see \cite{OD10}
for more details).

\subsection{General Gaussian relay networks (not necessarily layered)}
\label{sec:Gen}

Given the proof for layered networks, we are ready to tackle the proof
of Theorem \ref{thm:GaussMain} for general Gaussian relay networks.

Similar to the deterministic case, we first unfold the network
$\mathcal{G}$ over $K$ stages to create a layered network
$\mathcal{G}_{\text{unf}}^{(K)}$. The details of the construction are
described in Section \ref{sec:TimExpDet}, except now the linear
finite-field channels are replaced by Gaussian channels and the wired
links of capacity $K \overline{C}$ are replaced by orthogonal
point-to-point Gaussian links of capacity $K \overline{C}$ that do
not interfere with the other links in the network, where $\overline{C}$ is defined in (\ref{eq:CutSetRef}). 
 We now state the following lemma which is a corollary of Theorem
 \ref{thm:MainLay}.

\begin{lemma}
\label{lem:UnfGen}
All rates $R$
satisfying the following condition are achievable in $\mathcal{G}$:
\begin{eqnarray}
\label{eq:GenNetAchiRate}
R < \frac{1}{K}  \overline{C}_{\text{unf}}^{(K)} - \kappa
\end{eqnarray}
where $\overline{C}_{\text{unf}}^{(K)}$ is the cut-set upper bound on
the capacity of $\mathcal{G}_{\text{unf}}^{(K)}$, and $\kappa=15 (|\mathcal{V}|+\frac{2}{K})$.
\end{lemma}
\begin{proof}
  $\mathcal{G}_{\text{unf}}^{(K)}$ is a layered network. Therefore, by
  Theorem \ref{thm:MainLay}, all rates $R_{\text{unf}}$, satisfying
  the following condition are achievable in
  $\mathcal{G}_{\text{unf}}^{(K)}$: \beq \displaystyle R_{\text{unf}}
  < \overline{C}_{\text{unf}}^{(K)}- \kappa_{\text{unf}} \eeq where
  $\kappa_{\text{unf}}=15  |\mathcal{V}_{\text{unf}}^{(K)}|$. But the number of nodes at
  each stage of $\mathcal{G}_{\text{unf}}^{(K)}$ is exactly
  $|\mathcal{V}|$ (other than stage $0$ and $K+1$ which respectively contain the source, $S[0]$, and the destination, $D[K+1]$). Hence, $\kappa_{\text{unf}}=15 (K
  |\mathcal{V}|$+2). Now, similar to the proof of Lemma
  \ref{lem:UnfLinDet}, our achievability scheme (described in Section
  \ref{subsec:EncLay}) can be implemented in $\mathcal{G}$ by using
  $K$ blocks of size $T$ symbols. Therefore, we can achieve
  $\frac{1}{K} R_{\text{unf}}$ in $\mathcal{G}$ and the proof is
  complete.
\end{proof}

Similar to the deterministic case, it is easy to see that \beq
\label{eq:cutSetRelation} \overline{C}_{\text{unf}}^{(K)} \geq (K-|\mathcal{V}|)\overline{C}.
  \eeq

  Hence, by Lemma \ref{lem:UnfGen} and (\ref{eq:cutSetRelation}), we
  can achieve all rates up to \beq R < \frac{K-|\mathcal{V}|}{K}
  \overline{C} - \kappa \eeq where
  $\kappa=15 (|\mathcal{V}|+\frac{2}{K})$. By letting $K
  \rightarrow \infty$ the proof of Theorem \ref{thm:GaussMain} is
  complete.

  To prove Theorem \ref{thm:GaussMainMulticast}, {\em i.e.,} the
  multicast scenario, we just need to note that if all relays will
  perform exactly the same strategy then by our theorem, each
  destination, $D \in \mathcal{D}$, will be able to decode the message
  with low error probability as long as the rate of the message
  satisfies 
  \beq
  R< \min_{D\in\mathcal{D}}\overline{C}_{i.i.d.,D}-\kappa' 
  \eeq where
  $\kappa'<15|\mathcal{V}|$ is a constant and as in
  Definition \ref{def:cutSetiid} we have
  $ \overline{C}_{i.i.d.,D}= \min_{\Omega\in\Lambda_D}\log
  |\Ibf+P\Gbf_{\Omega}\Gbf_{\Omega}^*|$ is the cut-set bound evaluated
  for i.i.d. input distributions. Therefore as long as
  $R< \overline{C}_{\mathrm{mult}}-\kappa$,  where
  $\kappa<15|\mathcal{V}|$, all destinations can decode the
  message and hence the theorem is proved.

  In the case that we have multiple antennas at each node, the
  achievability strategy remains the same, except now each node
  receives a vector of observations from different antennas. We first
  quantize the received signal of each antenna at the noise level and
  then map it to another transmit codeword, which is joint across all
  antennas. The error probability analysis is exactly the same as
  before. However, the gap between the achievable rate and the cut-set
  bound will be larger. We can upper bound the gap between
  $\overline{C}$ and $\overline{C}_{iid}$ by twice the maximum number of
  degrees of freedom of the cuts, which due
  to \eqref{eq:TrivMultAntCutBnd} is at most
  $2 \sum_{i=1}^{|\mathcal{V}|}M_i$
  (see the last paragraph in
  Appendix \ref{app:proofLemBeamForming}). Also, by treating each
  receive antenna as a separate node and applying
  Lemma \ref{lem:condMI}, we get that
  $H([\ybf_D]|u,F_\mathcal{V})] \leq 12T \sum_{i=1}^{|\mathcal{V}|}
  N_i$. Therefore, from our previous analysis we know that the gap is
  at most $12 \sum_{i=1}^{|\mathcal{V}|}
  N_i+ 3\sum_{i=1}^{|\mathcal{V}|}M_i$
  and the theorem is proved when we have multiple antennas at each
  node.

\section{Connections between models}
\label{sec:connections}

In Section II, we showed that while the linear finite-field channel
model captures certain high SNR behaviors of the Gaussian model, it does
not capture all aspects. In particular, its capacity is not within a
constant gap to the Gaussian capacity for all MIMO channels. A natural
question is: is there a deterministic channel model which approximates
the Gaussian relay network capacity to within a constant gap?

The proof of the approximation theorem for the Gaussian network
capacity in the previous section already provides a partial answer
to this question. We showed that, after quantizing all the output at
the relays as well as the destination, the end-to-end mutual
information achieved by the relaying strategy in the noisy network
is close to that achieved when the noise sequences are known at the
destination, uniform over all realizations of the noise sequences.
In particular, this holds true when the noise sequences are all
zero. Since the former has been proved to be close to the capacity
of the Gaussian network,  this implies that the capacity of the {\em
quantized} deterministic model with
\begin{equation}
  \label{eq:TrunDetModel} {\bf y}_j[t]=\left [ \sum_{i\in \mathcal{V}}
    {\bf H}_{ij} {\bf }x_i[t] \right ], \qquad j=1, \ldots,
  |\mathcal{V}| \end{equation} 
must be {\em at least} within a constant gap to the capacity of the
Gaussian network. It is not too difficult to show that the
deterministic model capacity cannot be much larger. We establish all
this more formally in the next section, where we call the model in
\eqref{eq:TrunDetModel} as the {\em truncated deterministic model}.

\subsection{Connection between the truncated deterministic model and the Gaussian model}
\label{subsec:connTrunc}

\begin{theorem}
\label{thm:connGaussTrun}
The capacity of any Gaussian relay network, $C_{\text{Gaussian}}$, and
the capacity of the corresponding truncated deterministic model,
$C_{\text{Truncated}}$, satisfy the following relationship:
\beq \label{eq:connGaussTrun}
|C_{\text{Gaussian}}-C_{\text{Truncated}}| \leq
33 |\mathcal{V}|. \eeq
\end{theorem}


To prove this theorem  we need the following lemma which is proved in Appendix \ref{app:lemconnGaussTrunMIMO}. 

\begin{lemma} \label{lem:connGaussTrunMIMO}
Let $G$ be the channel gains matrix of a $m \times n$ MIMO system. Assume that there is an average power constraint equal to one at each node. Then for any input distribution $P_{\xbf}$,
\beq  |I(\xbf;G \xbf+Z)- I(\xbf;[G \xbf])| \leq 19n \eeq
where $Z=[z_1, \ldots, z_n]$ is a vector of $n$ i.i.d. $\mathcal{CN}(0,1)$ random variables.
\end{lemma}

\begin{proof} \textbf{(proof of Theorem \ref{thm:connGaussTrun})} \\
  First note that the value of any cut in the network is the same as
  the mutual information of a MIMO system. Therefore from Lemma
  \ref{lem:connGaussTrunMIMO} we have 
\beq \label{eq:connIneq1}
  |\overline{C}_{\text{Gaussian}}-\overline{C}_{\text{Truncated}}|
  \leq 19 |\mathcal{V}|. 
\eeq 
Now pick i.i.d. normal $\mathcal{CN}(0,1)$ distribution for
$\{x_i\}_{i \in \mathcal{V}}$. By applying Theorem \ref{thm:GenDetNet}
to the truncated deterministic relay network, we find
\beq
C_{\text{Truncated}} \geq \min_{\Omega\in\Lambda_D}
I(y_{\Omega^c}^{\text{truncated}};x_{\Omega}|x_{\Omega^c})
\stackrel{(a)}=H(y_{\Omega^c}^{\text{truncated}}|x_{\Omega^c}),  
\eeq
where $(a)$ is because we have a deterministic network.
By Lemma \ref{lem:BeamForming} and Lemma \ref{lem:connGaussTrunMIMO} we have
\begin{eqnarray}
\nonumber \min_{\Omega\in\Lambda_D}
I(y_{\Omega^c}^{\text{truncated}};x_{\Omega}|x_{\Omega^c}) & \geq & I(y_{\Omega^c}^{\text{Gaussian}};x_{\Omega}|x_{\Omega^c})-19 |\mathcal{V}| \\
\label{eq:connIneq2} & \geq & \overline{C}_{\text{Gaussian}}-20 |\mathcal{V}|.
\end{eqnarray}

Then from Equations (\ref{eq:connIneq1}) and (\ref{eq:connIneq2}) we have
\beq
\overline{C}_{\text{Gaussian}}-20 |\mathcal{V}| \leq C_{\text{Truncated}} \leq \overline{C}_{\text{Gaussian}}+19 |\mathcal{V}|.
\eeq
Also from Theorem \ref{thm:GaussMain} we know that
\beq \overline{C}_{\text{Gaussian}}- 15 |\mathcal{V}| \leq C_{\text{Gaussian}} \leq \overline{C}_{\text{Gaussian}}. \eeq

Therefore
\beq |C_{\text{Gaussian}}-C_{\text{Truncated}} | \leq  34 |\mathcal{V}|. \eeq

\end{proof}

\section{Extensions}
\label{sec:extensions}

In this section we extend our main result for Gaussian relay
networks (Theorem \ref{thm:GaussMain}) to the following scenarios:
\begin{enumerate}
  \item Compound relay network
  \item Frequency selective relay network
  \item Half-duplex relay network
  \item Quasi-static fading relay network (underspread regime)
  \item Low rate capacity approximation of Gaussian relay network
\end{enumerate}

\subsection{Compound relay network}
\label{subsec:compRelay}
The relaying strategy we proposed for general Gaussian relay networks does not require any channel
information at the relays; relays just quantize at noise level and
forward through a random mapping. The approximation gap also does
not depend on the channel gain values. As a result our main result
for Gaussian relay networks (Theorem \ref{thm:GaussMain}) can be
extended to compound relay networks where we allow each channel gain
$h_{i,j}$ to be from a set $\mathcal{H}_{i,j}$, and the particular
chosen values are unknown to the source node $S$, the relays, and the
destination node $D$. A communication rate $R$ is achievable if
there exists a scheme such that for any channel gain realizations,
 the source can communicate to the destination at rate $R$.

\begin{theorem}
\label{thm:compound}
The capacity $C_{cn}$ of the compound Gaussian relay network satisfies
\begin{equation}
\overline{C}_{cn}-\kappa \leq C_{cn} \leq \overline{C}_{cn},
\end{equation}
where $\overline{C}_{cn}$ is the cut-set upper bound on the compound capacity of
$\mathcal{G}$, \emph{i.e.}
\beq \overline{C}_{cn}=\max_{p(\{\xbf_i\}_{j \in \mathcal{V}})} \inf_{h \in \mathcal{H}} \min_{\Omega \in \Lambda_D} I(\ybf_{\Omega^c};\xbf_{\Omega}|\xbf_{\Omega^c}), \eeq
and $\kappa$
is a constant and is upper bounded by $13 \sum_{i=1}^{|\mathcal{V}|} N_i+ 3 \sum_{i=1}^{|\mathcal{V}|}M_i$, where $M_i$ and $N_i$ are respectively the number of transmit and receive antennas at node $i$.
\end{theorem}

\begin{proofOutline}
We sketch the proof for the case that nodes have single antenna; its extension to the multiple antenna scenario is straightforward. As we mentioned earlier, the relaying strategy that we used in Theorem \ref{thm:GaussMain} does not require any channel information. However, if all channel gains are known at the final destination, all rates within a constant gap to the cut-set upper bound are achievable. We first evaluate how much we lose if the final destination only knows a quantized version of the channel gains. In particular assume that each channel gain is bounded $|h_{ij}|\in [h_{\min},h_{\max}]$, and
final destination only knows the channel gain values quantized at level $\frac{1}{\sqrt{d_{\max}}}$, where $d_{\max}$ is the maximum degree of nodes in $\mathcal{G}$. Then since there is a transmit power constraint equal to one at each node, the effect of this channel uncertainty can be mimicked by adding a Gaussian noise of variance $d_{\max}\times \lp \frac{1}{\sqrt{d_{\max}}} \rp^2=1$  at each relay node (i.e., doubling the noise variance at each node), which will result in a reduction of at most $|\mathcal{V}|$  bits from the cut-set upper bound. Therefore with access to only quantized channel gains, we will lose at most $|\mathcal{V}|$ more bits, which means the gap between the achievable rate and the cut-set bound is at most $16 |\mathcal{V}|$.

Furthermore, as shown in \cite{FederLapidothUnivDecod} there exists a universal decoder for this finite set of channel sets. Hence we can use this decoder at the final destination and decode the message as if we knew the channel gains quantized at the noise level, for all rates up to
\beq \label{eq:compoundAchi} R < \max_{p(\{x_i\}_{j \in \mathcal{V}})} \inf_{\hat{h} \in \hat{\mathcal{H}}} \min_{\Omega \in \Lambda_D} I(y_{\Omega^c};x_{\Omega}|x_{\Omega^c}) \eeq
where $\hat{\mathcal{H}}$ is representing the quantized state space. Now as we showed earlier, if we restrict the channels to be quantized at noise level the cut-set upper bound changes at most by $|\mathcal{V}|$, therefore
\beq \label{eq:compoundCutLoss}\overline{C}_{cn}-|\mathcal{V}| \leq \max_{p(\{x_i\}_{j \in \mathcal{V}})} \inf_{\hat{h} \in \hat{\mathcal{H}}} \min_{\Omega \in \Lambda_D} I(y_{\Omega^c};x_{\Omega}|x_{\Omega^c}). \eeq
Therefore from Equations (\ref{eq:compoundAchi}) and (\ref{eq:compoundCutLoss}) all rates up to $\overline{C}_{cn}- 16 |\mathcal{V}|$ are achievable and the proof can be completed.

Now by using the ideas in  \cite{BlackwellCapGaussian} and \cite{VaraiyaCapGaussian}, we believe that an infinite state universal decoder can also be analysed to give ``completely oblivious to channel'' results.
\end{proofOutline}

\subsection{Frequency selective Gaussian relay network }
\label{subsec:mimo}
In this section we generalize our main result to
the case that the channels are frequency selective. Since one can
present a frequency selective channel as a MIMO link, where each
antenna is operating at a different frequency band\footnote{This can
be implemented in particular by using OFDM and appropriate spectrum
shaping or allocation.}, this extension is just a straightforward
corollary of the case that nodes have multiple antennas.
\begin{theorem}
\label{thm:MainFS}
The capacity $C$ of the frequency selective Gaussian relay network with $F$ different frequency bands satisfies
\begin{equation}
\label{eq:MainFS}
\overline{C}-\kappa \leq C \leq \overline{C}
\end{equation}
where $\overline{C}$ is the cut-set upper bound on the capacity of $\mathcal{G}$ as described in Equation (\ref{eq:CutSetRef}), and $\kappa$ is a constant and is upper bounded by $12 F \sum_{i=1}^{|\mathcal{V}|} N_i+ 3F \sum_{i=1}^{|\mathcal{V}|}M_i$, where $M_i$ and $N_i$ are respectively the number of transmit and receive antennas at node $i$.
\end{theorem}

\subsection{Half duplex relay network (fixed transmission scheduling)}
One of the practical constraints on wireless networks is that the
transceivers cannot transmit and receive at the same time on the
same frequency band, known as the half-duplex constraint. As a
result of this constraint, the achievable rate of the network will
in general be lower.  The model
that we use to study this problem is the same as
\cite{khojastepour}. In this model the network has finite modes of
operation. Each mode of operation (or state of the network), denoted
by $m \in \{1,2,\ldots,M \}$, is defined as a valid partitioning of
the nodes of the network into two sets of ``sender'' nodes and
``receiver'' nodes such that there is no active link  that arrives at
a sender node\footnote{Active link is defined as a link which is
departing from the set of sender nodes}. For each node $i$, the
transmit and the receive signal at mode $m$ are respectively shown
by $x_i^m$ and $y_i^m$. Also $t_m$ defines the fraction of the time
that network will operate in state $m$, as the network use goes to
infinity. The cut-set upper bound on the capacity of the Gaussian
relay network with half-duplex constraint, $C_{hd}$, is shown to be
\cite{khojastepour}

 {\footnotesize \beq \label{eq:MinCutHDRe}C_{hd} \leq
\overline{C}_{hd}= \max_{\substack{p(\{x_j^m\}_{j\in\mathcal{V}, m
\in \{1,\ldots,M\}})\\ t_m:~ 0 \leq t_m \leq 1, ~\sum_{m=1}^M t_m=1
}} \min_{\Omega\in\Lambda_D} \sum_{m=1}^M t_m
I(y_{\Omega^c}^m;x_{\Omega}^m|x_{\Omega^c}^m).  \eeq}


\begin{theorem}
The capacity $C_{hd}$ of the Gaussian relay network with half-duplex constraint satisfies
\begin{equation}
\overline{C}_{hd}-\kappa \leq C_{hd} \leq \overline{C}_{hd}
\end{equation}
where $\overline{C}_{hd}$ is the cut-set upper bound on the capacity as described in equation (\ref{eq:MinCutHDRe}) and $\kappa$ is a constant and is upper bounded by $12 \sum_{i=1}^{|\mathcal{V}|} N_i+ 3 \sum_{i=1}^{|\mathcal{V}|}M_i$, where $M_i$ and $N_i$ are respectively the number of transmit and receive antennas at node $i$.
\end{theorem}
\begin{proof}
We prove the result for the case that nodes have single antenna; its extension to the multiple antenna scenario is straightforward. Since each relay can be either in a transmit or receive mode, we have a total of $M=2^{|\mathcal{V}|-2}$ number of modes. An example of a network with two relay and all four modes of half-duplex operation of the relays are shown in Figure \ref{fig:exHD}.
\begin{figure}
     \centering \subfigure[A two-relay network]{ \scalebox{0.75}{
       \includegraphics{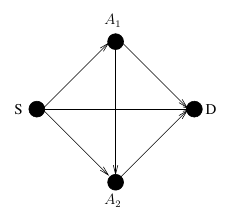}}
}
     \subfigure[Mode 1]{
       \scalebox{0.75}{\includegraphics{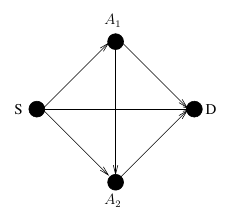}}
}
\hspace{0in}
     \subfigure[Mode 2]{
       \scalebox{0.75}{\includegraphics{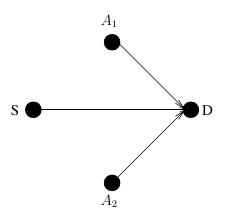}}
}
\hspace{-0.25in}
     \subfigure[Mode 3]{
    \scalebox{0.75}{\includegraphics{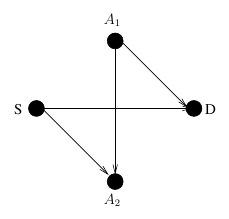}}
}
\hspace{-0.25in}
     \subfigure[Mode 4]{
    \scalebox{0.75}{\includegraphics{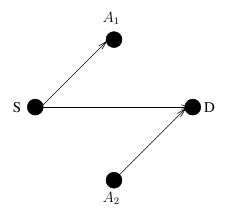}}
}
\caption{An example of a relay network with two relays is shown in $(a)$. All four modes of half-duplex operation of the relays are shown in $(b)-(e)$.
\label{fig:exHD}}
\end{figure}

Consider the $t_i$'s that maximize $\overline{C}_{hd}$ in
(\ref{eq:MinCutHDRe}). Assume that they are rational numbers
(otherwise look at the sequence of rational numbers approaching
them) and set $W$ to be the LCM (least common divisor) of the
denominators. Now increase the bandwidth of system by $W$ and
allocate $Wt_i$ of bandwidth to mode $i$, $i=1,\ldots,M$. Each
mode is running at a different frequency band. Therefore, as shown in
Figure \ref{fig:HDComb}, we can combine all these modes and create a
frequency selective relay network. Since the links are orthogonal to
each other, the cut-set upper bound on the capacity of this
frequency selective relay network (in bits/sec/Hz) is the same as
(\ref{eq:MinCutHDRe}). By theorem \ref{thm:MainFS} we know that
our quantize-map-and-forward scheme achieves, within a constant gap,
$\kappa$, of $\overline{C}_{hd}$ for all channel gains. In this
relaying scheme, at each block, each relay transmits a signal that
is only a function of its received signal in the previous block and
hence does not have memory over different blocks. We will
translate this scheme to a scheme in the original network that modes
are just at different times (not different frequency bands). The
idea is that we can expand exactly communication block of the
frequency selective network into $W$ blocks of the original network
and allocating $Wt_i$ of these blocks to mode $i$. In the
$Wt_i$ blocks that are allocated to mode $i$, all relays do exactly
what they do in frequency band $i$. This is described in
Figure \ref{fig:HDTimeExp} for the network of Figure
\ref{fig:HDComb}. This figure shows how one communication block of
the frequency selective network (a) is expanded over $W$ blocks of
the the original half-duplex network (b). Since the transmitted
signal at each frequency band is only a function of the data
received in the previous block of the frequency selective network,
the ordering of the modes inside the $W$ blocks of the original
network is not important at all. Therefore with this strategy we can
achieve within a constant gap, $\kappa$, of the cut-set bound of the
half-duplex relay network and the proof is complete.

One of the differences between this strategy and our original strategy
for full duplex networks is that now the relays might be required to
have a much larger memory. In the full duplex scenario, in the layered
case the relays had only memory over one block\footnote{This could be
  also done in the arbitrary networks but requires an alternative
  analysis. See footnote in Section \ref{sec:TimExpDet}.} (what they
sent was only a function of the previous block). However for the
half-duplex scenario the relays are required to have a memory over $W$
blocks and $W$ can be arbitrarily large.

\begin{figure}
\centering
\scalebox{0.7}{\includegraphics{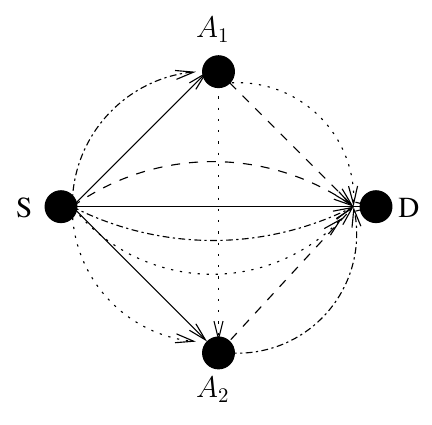}}
\caption{\label{fig:HDComb} Combination of all half-duplex modes of the network shown in figure \ref{fig:exHD}. Each mode operates at a different frequency band.}
\end{figure}

\begin{figure*}
     \centering
     \subfigure[]{
       \scalebox{0.7}{\includegraphics{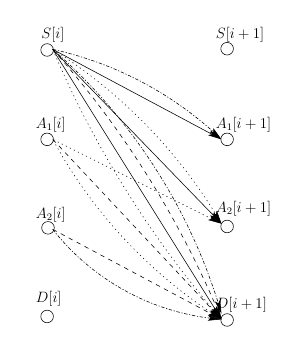}}
}
\vspace{0.2in}
     \subfigure[]{
        \scalebox{0.7}{\includegraphics{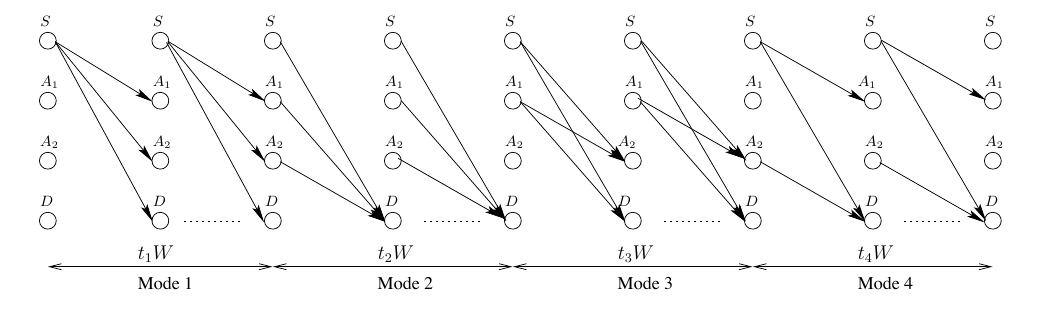}}
}
\caption{One communication block of the frequency selective network (a), and its expansion over $W$ blocks of the original half-duplex network (b).
\label{fig:HDTimeExp}}
\end{figure*}

\end{proof}

\subsection{Quasi-static fading relay network (underspread regime)}

In a wireless environment channel gains are not fixed and can change. In this section we consider a typical scenario in which
although the channel gains change, they can be considered time
invariant over a long time scale (for example during the
transmission of a block). This happens when the coherence time of
the channel ($T_c$) is much larger than the delay spread ($T_d$).
Here the delay spread is the largest extent of the unequal path
lengths, which is in some sense corresponding to inter-symbol
interference. Now, depending on how fast the channel gains are
changing compared to the delay requirements, we have two different
regimes: fast fading or slow fading scenarios. We consider each case
separately.

\subsubsection{Fast fading}
In the fast fading scenario the channel gains are changing much
faster compared to the delay requirement of the application
(\emph{i.e.}, coherence time of the channel, $T_c$, is much smaller
than the delay requirements). Therefore, we can interleave data and
encode it over different coherence time periods. In this scenario,
ergodic capacity of the network is the relevant capacity measure to
look at.


\begin{theorem}
The ergodic capacity $C_{\text{ergodic}}$ of the quasi-static fast fading Gaussian relay network satisfies
\begin{equation}
\mathcal{E}_{h_{ij}} \left [ \overline{C}(\{h_{ij}\}) \right ]- \kappa \leq C_{\text{ergodic}} \leq \mathcal{E}_{h_{ij}} \left [ \overline{C}(\{h_{ij}\}) \right ]
\end{equation}
where $\overline{C}$ is the cut-set upper bound on the capacity as
described in Equation (\ref{eq:CutSetRef}) and the expectation is
taken over the channel gain distribution. Also, the constant $\kappa$ is  upper bounded by $12 \sum_{i=1}^{|\mathcal{V}|} N_i+ 3 \sum_{i=1}^{|\mathcal{V}|}M_i$, where $M_i$ and $N_i$ are respectively the number of
transmit and receive antennas at node $i$.
\end{theorem}

\begin{proof}
We prove the result for the case that nodes have single antenna. Its extension to the multiple antenna scenario is straightforward. An upper bound is just the cut-set upper bound. For the achievability note that the relaying strategy we proposed for general relay networks does not depend on the channel realization, relays just quantize at noise level and forward through a random mapping. The approximation gap also does not depend on the channel parameters. As a result by coding data over $L$ different channel realizations the following rate is achievable
\beq \frac{1}{L} \sum_{l=1}^L \lp \overline{C}(\{h_{ij}\}^l) - \kappa \rp. \eeq
Now as $L \rightarrow \infty$,
\beq \frac{1}{L}  \sum_{l=1}^L \overline{C}(\{h_{ij}\}^l)  \rightarrow \mathcal{E}_{h_{ij}} \left [ \overline{C} \right ] \eeq and the theorem is proved.
\end{proof}

\subsubsection{Slow fading}
In a slow fading scenario the delay requirement does not allow us to
interleave data and encode it over different coherence time periods.
We assume that there is no channel gain information available at the
source, therefore there is no definite capacity and for a fixed
target rate $R$ we should look at the outage probability, \beq
\mathcal{P}_{out}(R) = \prob{ C(\{h_{ij}\}) < R} \eeq where the
probability is calculated over the distribution of the channel gains
and the $\epsilon$-outage capacity is defined as \beq
C_{\epsilon}=\mathcal{P}_{out}^{-1} (\epsilon). \eeq Here is our main
result to approximate the outage probability.

\begin{theorem}
\label{thm:outage}
The outage probability $\mathcal{P}_{out}(R)$ of the quasi-static slow fading Gaussian relay network satisfies
\begin{equation}
\label{eq:outageprob}
\prob{ \overline{C}(\{h_{ij}\}) < R}  \leq \mathcal{P}_{out}(R) \leq \prob{ \overline{C}(\{h_{ij}\}) < R+\kappa}
\end{equation}
where $\overline{C}$ is the cut-set upper bound on the capacity as described in Equation (\ref{eq:CutSetRef}) and the probability is calculated over the distribution of the channel gains. The constant $\kappa$ is upper bounded by $12 \sum_{i=1}^{|\mathcal{V}|} N_i+ 3 \sum_{i=1}^{|\mathcal{V}|}M_i$, where $M_i$ and $N_i$ are respectively the number of
transmit and receive antennas at node $i$.
\end{theorem}
\begin{proof}
Lower bound is just based on the cut-set upper bound on the
capacity. For the upper bound we use the compound network result.
Therefore, based on Theorem \ref{thm:compound} we know that as long
as $  \overline{C}(\{h_{ij}\})-\kappa < R $ there will not be an
outage.
\end{proof}

\subsection{Low rate capacity approximation of Gaussian relay network}
In the low data rate regime, a constant-gap approximation of the
capacity may not be useful any more. A more useful kind of
approximation in this regime would be a universal multiplicative
approximation, where the multiplicative factor
does not depend on the channel gains in the network.

\begin{theorem}
The capacity $C$ of the Gaussian relay network satisfies
\begin{equation}
\lambda \overline{C} \leq C \leq \overline{C}
\end{equation}
where $\overline{C}$ is the cut-set upper bound on the capacity, as
described in Equation (\ref{eq:CutSetRef}), and $\lambda$ is a
constant and is lower bounded by $\frac{1}{2d\lp d+1 \rp}$, where $d$
is the maximum degree of nodes in $\mathcal{G}$.
\end{theorem}

\begin{proof}
First we use a time-division scheme and make all links in the network orthogonal to each other. By Vizing's theorem (e.g., see \cite{ModernGraphTheory} p.153) any simple undirected graph can be edge colored with at most $d+1$ colors, where $d$ is the maximum degree of nodes in $\mathcal{G}$. Since our graph $\mathcal{G}$ is a directed graph we need at most $2 (d+1)$ colors. Therefore we can generate $2(d+1)$ time slots and assign the slots to directed graphs such that at any node all the links are orthogonal to each other. Therefore each link is used a $\frac{1}{2 (d+1)}$ fraction of the time. We further impose the constraint that each of these links uses a total $\frac{1}{2 d(d+1)}$ of the time, but with a factor of $d$ more power. By coding we can convert each links $h_{i,j}$ into a  noise free link with capacity
\beq c_{i,j}=\frac{1}{2 d(d+1)} \log(1+d |h_{i,j}|^2). \eeq
By Ford-Fulkerson theorem we know that the capacity of this network is
\beq C_{\text{orthogonal}}=\min_{\Omega} \sum_{i,j: i \in \Omega, j \in \Omega^c} c_{i,j}\eeq
and this rate is achievable in the original Gaussian relay network. Now we will prove that
\beq C_{\text{orthogonal}} \geq \frac{1}{2 d(d+1)} \overline{C}. \eeq

To show this, assume that in the orthogonal network each node transmits the same signal on its outgoing links. Furthermore, each node $j$ takes the summation of all incoming signals (normalized by $\frac{1}{\sqrt{d}}$) and denotes it as its received signal $y_j$, \emph{i.e.} 
\begin{eqnarray} y_{j}[t]&=&\frac{1}{\sqrt{d}}  \sum_{i=1}^d \lp h_{ij} \sqrt{d} x_i[t] + z_{ij}[t] \rp \\
&=& \sum_{i=1}^d h_{ij}  x_i[t] + \tilde{z}_j[t] \end{eqnarray}
where
\beq \tilde{z}_j[t] = \frac{\sum_{i=1}^d z_{ij}[t] }{\sqrt{d}} \sim \mathcal{CN}(0,1). \eeq
Therefore we get a network which is statically similar to the original non-orthogonal network, however each time-slot is only a $\frac{1}{d(d+1)}$ fraction of the time slots in the original network. Therefore without this restriction the cut-set of the orthogonal network can only increase. Hence
\beq C_{\text{orthogonal}} \geq \frac{1}{2 d(d+1)} \overline{C}. \eeq
\end{proof}

\section{Conclusions}
\label{sec:conclusions}

In this paper we presented a new approach to analyze the capacity of
Gaussian relay networks. We start with deterministic models to build
insights and use them as the foundation to analyze Gaussian models. The
main results are a new scheme for general Gaussian relay networks
called quantize-map-and-forward and a proof that it can achieve to within
a constant gap to the cutset bound. The gap does not depend on the SNR
or the channel values of the network. No other scheme in the
literature has this property.

One limitation of these results is that the gap grows with the number
of nodes in the network. This is due to the noise accumulation
property of the quantize-map-and-forward scheme. It is an interesting
question whether there is another scheme that can circumvent this to
achieve a universal constant gap to the cutset bound, independent of
the number of nodes, or if this is an inherent feature of any
scheme. In this case a better upper bound than the
cutset bound is needed.

\appendices

\section{Proof of Theorem \ref{thm:oneRelay1Bit}}
\label{app:oneRelay1Bit}
If $|h_{SR}|<|h_{SD}|$ then the relay is ignored and a communication
rate equal to $R=\log(1+|h_{SD}|^2)$ is achievable. If
$|h_{SR}|>|h_{SD}|$ the problem becomes more interesting. In this case
by using the decode-forward scheme described in \cite{coverElgamal} we
can achieve \beq \nonumber R = \min \lp \log \lp 1+|h_{SR}|^2 \rp , \log \lp
1+|h_{SD}|^2+|h_{RD}|^2 \rp \rp. \eeq Therefore, overall the following
rate is always achievable
\begin{eqnarray*} \nonumber R_{\text{DF}} &=& \max \{ \log(1+|h_{SD}|^2),
  \min [ \log \lp 1+|h_{SR}|^2 \rp \\ \nonumber  && \quad \quad , \log \lp
  1+|h_{SD}|^2+|h_{RD}|^2  \rp ] \}. \end{eqnarray*} Now we compare
this achievable rate with the cut-set upper bound on the capacity of
the Gaussian relay network

\begin{eqnarray*}
\nonumber C \leq  \overline{C} &=& \max_{|\rho| \leq 1} \min \{ \log
\lp 1+(1-\rho^2)(|h_{SD}|^2+|h_{SR}|^2) \rp \\ \nonumber &&  \quad  ,  \log \lp
1+|h_{SD}|^2+|h_{RD}|^2+2\rho |h_{SD}||h_{RD}| \rp \}.
\end{eqnarray*}

Note that if $|h_{SR}|\geq |h_{SD}|$ then
 \beq \nonumber R_{DF}= \min \lp \log \lp
1+|h_{SR}|^2 \rp , \log \lp 1+|h_{SD}|^2+|h_{RD}|^2 \rp \rp \eeq and
for all $|\rho|\leq 1$ we have  \beq \nonumber \log \lp
1+(1-\rho^2)(|h_{SD}|^2+|h_{SR}|^2) \rp \leq \log \lp 1+|h_{SR}|^2
\rp+1 \eeq 
 \begin{eqnarray*}
&& \log \lp 1+|h_{SD}|^2+|h_{RD}|^2+2\rho |h_{SD}||h_{RD}| \rp \\ \leq
&& \quad \quad \quad \quad \quad \quad \quad \quad \quad \quad   \log \lp 1+|h_{SD}|^2+|h_{RD}|^2 \rp +1
\end{eqnarray*}
Hence \beq  \nonumber R_{\text{DF}} \geq \overline{C}_{\text{relay}}-1. \eeq Also
if $|h_{SD}|>|h_{SR}|$, \beq \nonumber R_{DF}=\log(1+|h_{SD}|^2) \eeq and {\small \beq
 \nonumber \log \lp 1+(1-\rho^2)(|h_{SD}|^2+|h_{SR}|^2) \rp \leq \log \lp
1+|h_{SD}|^2 \rp+1 \eeq } therefore again, \beq \nonumber R_{\text{DF}} \geq
\overline{C}_{\text{relay}}-1. \eeq

\section{Proof of Theorem \ref{thm:twoRelay1Bit}}
\label{app:twoRelay1Bit}
The cut-set upper bound on the capacity of diamond network is
\begin{eqnarray*}
\nonumber C_{\text{diamond}} & \leq & \overline{C} \leq \min \{ \log
\lp 1+ |h_{SA_1}|^2+|h_{SA_2}|^2 \rp, \\ \nonumber & &  \log \lp 1+
(|h_{A_1D}|+|h_{A_2D}|)^2 \rp, \\ \nonumber & &  \log
(1+|h_{SA_1}|^2)+\log (1+|h_{A_2D}|^2), \\ & &  \log
(1+|h_{SA_2}|^2)+\log (1+|h_{A_1D}|^2) \}. ~ ~~~~
\end{eqnarray*}
Without loss of generality assume $|h_{SA_1}| \geq  |h_{SA_2}|$. Then we have the following cases:
\begin{enumerate}
  \item $|h_{SA_1}| \leq  |h_{A_1D}|$:  In this case
  \beq \nonumber R_{PDF} \geq \log(1+|h_{SA_1}|^2) \geq \overline{C}-1. \eeq
  
\item $|h_{SA_1}| > |h_{A_1D}|$:\\
  Let $\alpha=\frac{|h_{A_1D}|^2}{|h_{SA_1}|^2}$, then 
  \begin{align*} R_{PDF} =
\nonumber &  \log(1+|h_{A_1D}|^2)+ \min  \{ \log \lp 1+\frac{\lp 1- \alpha
      \rp |h_{SA_2}|^2}{\alpha |h_{SA_2}|^2+1} \rp  \\ & \quad , \log \lp 1+\frac{
      |h_{A_2D}|^2}{1+ |h_{A_1D}|^2} \rp  \} \end{align*} or
  \begin{align} \nonumber &  R_{PDF} = \min  \{ \log \lp \frac{(1+
      |h_{SA_2}|^2)(1+|h_{A_1D}|^2)}{\alpha |h_{SA_2}|^2+1} \rp \\ & \label{eq:PDFRate} \quad \quad \quad \quad , \log
    \lp 1+|h_{A_1D}|^2+|h_{A_2D}|^2 \rp  \}. \end{align} Now if \begin{align*} &
  \log \lp \frac{(1+ |h_{SA_2}|^2)(1+|h_{A_1D}|^2)}{\alpha
    |h_{SA_2}|^2+1} \rp \geq \\ & \quad \quad \quad \quad \quad \quad \quad \log \lp 1+|h_{A_1D}|^2+|h_{A_2D}|^2 \rp
  \end{align*} we have
\begin{eqnarray*}
 R_{PDF} & = & \log \lp 1+|h_{A_1D}|^2+|h_{A_2D}|^2 \rp \\
 & \geq & \log \lp 1+\lp |h_{A_1D}|+|h_{A_2D}|\rp^2 \rp -1 \geq  \overline{C}-1.
\end{eqnarray*}
Therefore, the achievable rate of partial decode-forward scheme is
within one bit of the cut-set bound. So we just need to look at the
case that \beq \nonumber R_{PDF} = \log \lp \frac{(1+
  |h_{SA_2}|^2)(1+|h_{A_1D}|^2)}{\alpha |h_{SA_2}|^2+1} \rp. \eeq In
this case consider two possibilities:
\begin{itemize}
    \item $\alpha |h_{SA_2}|^2 \leq 1$: Here we have
    \begin{align*}
 &   R_{PDF}  =   \log \lp \frac{(1+ |h_{SA_2}|^2)(1+|h_{A_1D}|^2)}{\alpha |h_{SA_2}|^2+1}  \rp\\
    & \geq  \log \lp \frac{(1+ |h_{SA_2}|^2)(1+|h_{A_1D}|^2)}{2}  \rp \\
    &=\log(1+|h_{SA_2}|^2)+\log (1+|h_{A_1D}|^2)-1 \geq  \overline{C}-1.
    \end{align*}
    
    \item $\alpha |h_{SA_2}|^2 \geq 1$:\\
    In this case we will show that
    \begin{eqnarray*}
    R_{PDF} & = &  \log \lp \frac{(1+ |h_{SA_2}|^2)(1+|h_{A_1D}|^2)}{\alpha |h_{SA_2}|^2+1}  \rp\\
    & \geq & \log \lp 1+ |h_{SA_1}|^2+|h_{SA_2}|^2  \rp -1 \\
    & \geq & \overline{C}-1.
    \end{eqnarray*}
    To show this we just need to prove
   {\footnotesize  \beq  \nonumber \frac{(1+ |h_{SA_2}|^2)(1+|h_{A_1D}|^2)}{\alpha |h_{SA_2}|^2+1} \geq \frac{1}{2}\lp 1+ |h_{SA_1}|^2+|h_{SA_2}|^2  \rp. \eeq }
By replacing $\alpha=\frac{|h_{A_1D}|^2}{|h_{SA_1}|^2}$, we get
{\small \begin{align*} & 2 |h_{SA_1}|^2 (1+ |h_{SA_2}|^2)(1+|h_{A_1D}|^2) \geq \\ & \lp 1+ |h_{SA_1}|^2+|h_{SA_2}|^2  \rp \lp|h_{SA_1}|^2+|h_{SA_2}|^2|h_{A_1D}|^2 \rp. \end{align*}}
But note that
{\footnotesize
\begin{align*}
\nonumber & 2 |h_{SA_1}|^2 (1+ |h_{SA_2}|^2)(1+|h_{A_1D}|^2) - \\ & \lp 1+ |h_{SA_1}|^2+|h_{SA_2}|^2  \rp \lp |h_{SA_1}|^2+|h_{SA_2}|^2|h_{A_1D}|^2 \rp \\
 &= |h_{SA_1}|^2+|h_{SA_1}|^2|h_{A_1D}|^2 + (|h_{SA_1}|^2|h_{SA_2}|^2 - \\ & |h_{SA_2}|^4|h_{A_1D}|^2)
+(|h_{SA_1}|^2|h_{A_1D}|^2-|h_{SA_2}|^2|h_{A_1D}|^2)+ \\ & (|h_{SA_1}|^2|h_{SA_2}|^2|h_{A_1D}|^2-|h_{SA_1}|^4)\\
 &= |h_{SA_1}|^2+|h_{SA_1}|^2|h_{A_1D}|^2 + |h_{SA_2}|^2(|h_{SA_1}|^2- \\ & |h_{SA_2}|^2|h_{A_1D}|^2)
+|h_{A_1D}|^2(|h_{SA_1}|^2-|h_{SA_2}|^2)+ \\ & |h_{SA_1}|^2(|h_{SA_2}|^2|h_{A_1D}|^2-|h_{SA_1}|^2)\\ 
&= |h_{SA_1}|^2+|h_{SA_1}|^2|h_{A_1D}|^2 +(|h_{SA_1}|^2-|h_{SA_2}|^2) \\ & (|h_{SA_2}|^2|h_{A_1D}|^2-|h_{SA_1}|^2+|h_{A_1D}|^2) \geq 0
\end{align*}
}
where the last step is true since 
\begin{align*}
& |h_{SA_1}|^2  \geq   |h_{SA_2}|^2\\
& |h_{SA_2}|^2|h_{A_1D}|^2  \geq |h_{SA_1}|^2 \quad (\hbox{since } \alpha |h_{SA_2}|^2 \geq 1).
\end{align*}

\end{itemize}

\end{enumerate}

\section{Proof of Theorems \ref{thm:GenDetNet} and \ref{thm:GenDetNetMulticast}}
\label{app:GenDet}

In this appendix we prove Theorems \ref{thm:GenDetNet} and \ref{thm:GenDetNetMulticast}. We first generalize the encoding scheme to accommodate arbitrary deterministic functions of (\ref{eq:GenDetModel}) in Section \ref{subsec:EncGenDetLay}.
We then illustrate the ingredients of the proof using the same example
as in Section \ref{subsec:PfIdeaDet}. The complete proof of our result for
layered networks is proved in Section \ref{subsec:LayGen}. The extension to the non-layered case is very similar to the proof for linear finite-field model discussed in Section \ref{sec:TimExpDet}, hence is omitted.

\subsection{Encoding for layered general deterministic relay network}
\label{subsec:EncGenDetLay}
We have a single source $S$ with a sequence of messages $w_k \in\{1,2,\ldots,2^{TR}\}$, $k=1,2,\ldots$.
Each message is encoded by the source $S$ into a signal over $T$
transmission times (symbols), giving an overall transmission rate of
$R$.
We will use strong (robust) typicality as defined in \cite{OR01}.
The notion of joint typicality is naturally extended from Definition
\ref{def:RobTyp}.
\begin{defn}
\label{def:RobTyp}
We define $\underline{x}$ as $\delta$-typical with respect to distribution $p$, and denote it by $\underline{x}\in T_{\delta}$, if
\[
|\nu_{\underline{x}}(x)-p(x)| \leq \delta p(x), \quad \forall x
\]
where $\delta \in \mathbb{R}^+$  and $\nu_{\underline{x}}(x)=\frac{1}{T}|\{t: x_t=x\}|$, is the
empirical frequency.
\end{defn}

Each relay operates over blocks of time $T$ symbols, and uses a
mapping $f_j:\mathcal{Y}_j^T\rightarrow \mathcal{X}_j^T$ from its
previous block of received $T$ symbols to transmit
signals in the next block. In particular, block $k$ of $T$ received
symbols is denoted by
$\ybf_j^{(k)}=\{y[(k-1)T+1],\ldots,y[kT]\}$ and the transmit
symbols by $\xbf_j^{(k)}$.  Choose some product distribution
$\prod_{i\in\mathcal{V}}p(x_i)$. At the source $S$, map each of the
indices in $w_k\in\{1,2,\ldots,2^{TR}\}$, choose $f_S(w_k)$ onto a
sequence uniformly drawn from $T_{\delta}(x_S)$, which is the typical
set of sequences in $\mathcal{X}_S^T$.  At any relay node $j$ choose
$f_j$ to map each typical sequence in $T_{\delta}(y_j)$ onto the typical set of transmit sequences $T_{\delta}(x_j)$, as
\begin{equation}
\label{eq:GenEncFnLay}
\xbf_j^{(k)} = f_j(\ybf_j^{(k-1)}),
\end{equation}
where $f_j$ is chosen to map uniformly randomly each sequence in
$T_{\delta}(y_j)$ onto $T_{\delta}(x_j)$.  Each relay does the encoding prescribed by
(\ref{eq:GenEncFnLay}).  

\subsection{Proof illustration}
\label{subsec:PfIdeaGen}

Now, we illustrate the ideas behind the proof of Theorem
\ref{thm:GenDetNet} for layered networks using the same example as in
Section \ref{subsec:PfIdeaDet}, which was done for the linear
deterministic model.  Since we are dealing with deterministic
networks, the logic up to (\ref{eq:PwErr}) in Section
\ref{subsec:PfIdeaDet} remains the same. We will again illustrate the
ideas using the cut $\Omega=\{S,A_1,B_1\}$. As in Section
\ref{subsec:PfIdeaDet}, we can write
\begin{eqnarray*}
 \mathcal{P} &=& \mathbb{P} \{\mathcal{A}_2,\mathcal{B}_2,\mathcal{D},\mathcal{A}_1^c,\mathcal{B}_1^c\} \\
 &=& \mathbb{P} \{\mathcal{A}_2 \} \times \mathbb{P} \{\mathcal{B}_2,\mathcal{A}_1^c|\mathcal{A}_2 \}  \times \mathbb{P} \{ \mathcal{D},\mathcal{B}_1^c| \mathcal{A}_2,\mathcal{B}_2,\mathcal{A}_1^c \} \\
 & \leq &  \mathbb{P} \{\mathcal{A}_2 \} \times \mathbb{P} \{\mathcal{B}_2|\mathcal{A}_1^c,\mathcal{A}_2 \} \times \mathbb{P} \{ \mathcal{D}|\mathcal{B}_1^c, \mathcal{A}_2,\mathcal{B}_2,\mathcal{A}_1^c \} \\
 & = & \mathbb{P} \{\mathcal{A}_2 \} \times \mathbb{P} \{\mathcal{B}_2|\mathcal{A}_1^c,\mathcal{A}_2 \}  \times \mathbb{P} \{ \mathcal{D}|\mathcal{B}_1^c,\mathcal{B}_2\} 
\end{eqnarray*}
where the events $\{ \mathcal{A}_1,\mathcal{A}_2,\mathcal{B}_1,\mathcal{B}_2,\mathcal{D}\}$ are defined in (\ref{eq:EventsDetDef}), and the last step is true since there is an independent random mapping at each node and we have a Markovian layered structure in the network.

Note that since $\ybf_j\in T_{\delta}(y_j)$ with high probability, we
can focus only on the typical received signals. Let us first examine
the probability that $\ybf_{A_2}(w)=\ybf_{A_2}(w')$. Since $S$ can
distinguish between $w,w'$, it maps these messages independently
to two transmitted signals $\xbf_S(w),\xbf_S(w')\in T_{\delta}(x_S)$,
hence we can see that
{\small
\begin{equation}
\label{eq:ExGenEP1}
 \mathbb{P} \{\mathcal{A}_2 \} = \prob{(\xbf_S(w'),\ybf_{A_2}(w))\in T_{\delta}(x_S,y_{A_2})} \stackrel{\cdot}{=}
2^{-TI(x_S;y_{A_2})},
\end{equation}}
where $\stackrel{\cdot}{=}$ indicates exponential equality (where we
neglect subexponential constants).

 Now, in order to analyze the second probability, as seen in the linear model analysis, $\mathcal{A}_2$ implies
$\xbf_{A_2}(w)=\xbf_{A_2}(w')$, {\em i.e.,} the {\em same} signal is
sent under both $w,w'$. Therefore, since
$(\xbf_{A_2}(w),\ybf_{B_2}(w))\in T_{\delta}(x_{A_2},y_{B_2})$,
obviously, $(\xbf_{A_2}(w'),\ybf_{B_2}(w))\in
T_{\delta}(x_{A_2},y_{B_2})$ as well.  Therefore, under $w'$, we
already have $\xbf_{A_2}(w')$ to be jointly typical with the signal
that is received under $w$. However, since $A_1$ can distinguish
between $w,w'$, it will map the transmit sequence $\xbf_{A_1}(w')$ to
a sequence which is independent of $\xbf_{A_1}(w)$ transmitted under
$w$. Since an error occurs when
$(\xbf_{A_1}(w'),\xbf_{A_2}(w'),\ybf_{B_2}(w))\in
T_{\delta}(x_{A_1},x_{A_2},y_{B_2})$, and since $A_2$ cannot
distinguish between $w,w'$, we also have
$\xbf_{A_2}(w)=\xbf_{A_2}(w')$, we require that
$(\xbf_{A_1},\xbf_{A_2},\ybf_{B_2})$ generated like
$p(\xbf_{A_1})p(\xbf_{A_2},\ybf_{B_2})$ behaves like a jointly typical
sequence. Therefore, this probability is given by
\begin{align}
\nonumber &   \mathbb{P} \{\mathcal{B}_2|\mathcal{A}_1^c,\mathcal{A}_2 \} = \mathbb{P} \{(\xbf_{A_1}(w'),\xbf_{A_2}(w),\ybf_{B_2}(w))\in \\ \label{eq:ExGenEP2} & 
T_{\delta}(x_{A_1},x_{A_2}y_{B_2}) \} \stackrel{\cdot}{=}  
2^{-TI(x_{A_1};y_{B_2},x_{A_2})} \stackrel{(a)}{=}
2^{-TI(x_{A_1};y_{B_2}|x_{A_2})},
\end{align}
where $(a)$ follows since we have generated the mappings $f_j$ independently, it induces an independent
distribution on $x_{A_1},x_{A_2}$. Another way to see this is that the
probability (\ref{eq:ExGenEP2}) is 
$\frac{|T_{\delta}(\xbf_{A_1}|\xbf_{A_2},\ybf_{B_2})|}{|T_{\delta}(\xbf_{A_1})|}$,
which by using properties of (robustly) typical sequences \cite{OR01}
yields the same expression as in (\ref{eq:ExGenEP2}). Note that the
calculation in (\ref{eq:ExGenEP2}) is similar to one of the error
event calculations in a multiple access channel.

Using a similar logic we can write
{\footnotesize
\begin{align}
\nonumber &  \mathbb{P} \{ \mathcal{D}|\mathcal{B}_1^c,\mathcal{B}_2\} =
\prob{(\xbf_{B_1}(w'),\xbf_{B_2}(w),\ybf_{D}(w))\in
T_{\delta}(x_{B_1},x_{B_2}y_{D})} \stackrel{\cdot}{=}
\\ \label{eq:ExGenEP3} & 2^{-TI(x_{B_1};y_{D},x_{B_2})} \stackrel{(a)}{=}
2^{-TI(x_{B_1};y_{D}|x_{B_2})}.
\end{align}}
Therefore, putting (\ref{eq:ExGenEP1})--(\ref{eq:ExGenEP3}) together as
done in (\ref{eq:ConfProb}) we get
\[
\mathcal{P}\leq
2^{-T\{I(x_S;y_{A_2})+I(x_{A_1};y_{B_2}|x_{A_2})+I(x_{B_1};y_{D}|x_{B_2})\}}.
\]
Note that, for this example, due to the Markovian structure of the
network we can see that\footnote{Though in the encoding
scheme there is a dependence between $x_{A_1},x_{A_2},x_{B_1},x_{B_2}$
and $x_S$, in the single-letter form of the mutual information, under
a product distribution, $x_{A_1},x_{A_2},x_{B_1},x_{B_2},x_S$ are
independent of each other. Therefore for example, $y_{B_2}$ is
independent of $x_{B_2}$ leading to
$H(y_{B_2}|x_{A_2},x_{B_2})=H(y_{B_2}|x_{A_2})$.  Using this argument for
the cut-set expression $I(y_{\Omega^c};x_{\Omega}|x_{\Omega^c})$, we get
the expansion.}  $I(y_{\Omega^c};x_{\Omega}|x_{\Omega^c})
=I(x_S;y_{A_2})+I(x_{A_1};y_{B_2}|x_{A_2})+I(x_{B_1};y_{D}|x_{B_2})$,
hence as in (\ref{eq:RateBndLinEx}) we get 
\begin{equation}
\label{eq:RateBndGenEx}
\displaystyle
P_e \leq 2^{RT} |\Lambda_D| 2^{-T\min_{\Omega\in\Lambda_D}
I(y_{\Omega^c};x_{\Omega}|x_{\Omega^c})},
\end{equation}
and hence the error probability can be made as small as desired if
$R<\min_{\Omega\in\Lambda_D}H(y_{\Omega^c}|x_{\Omega^c})$.

\subsection{Proof of Theorems \ref{thm:GenDetNet} and \ref{thm:GenDetNetMulticast} for layered networks}
\label{subsec:LayGen}

As in the example illustrating the proof in Section
\ref{subsec:PfIdeaGen}, the logic of the proof in the general
deterministic functions follows that of the linear model quite
closely. 

For any such cut $\Omega$, define the following sets:
\begin{itemize}
\item $L_l (\Omega)$: the nodes that are in $\Omega$ and are at layer $l$, (for example $S \in  L_1 (\Omega)$),
\item $R_l (\Omega)$: the nodes that are in $\Omega^c$ and are at layer $l$, (for example $D \in  R_{l_D} (\Omega)$).
\end{itemize}

As in Section \ref{sec:LayFF} we can define
the bi-partite network associated with a cut $\Omega$. Instead of a
transfer matrix $\Gbf_{\Omega,\Omega^c}(\cdot)$ associated with the cut, we
have a transfer function $\tilde{\Gbf}_{\Omega}$.  Since we are still
dealing with a layered network, as in the linear model case, this
transfer function breaks up into components corresponding to each of
the $l_D$ layers of the network. More precisely, we can create $d=l_D$
disjoint sub-networks of nodes corresponding to each layer of the
network,  with
the set of nodes $L_{l-1}(\Omega)$, which are at distance $l-1$ from $S$ and are in $\Omega$, on one side and  the set of nodes
$R_l(\Omega)$, which are at distance $l$ from $S$ that are in $\Omega^c$, on the other side, for
$l=2,\ldots,l_D$. Each of these clusters have a
transfer function $\Gbf_l(\cdot),l=1,\ldots,l_D$ associated with them.

As in the linear model, each node $i$ sees a signal
related to $w=w_1$ in block $l_i=l-1$, and therefore waits to receive
this block and then does a mapping using the general encoding function
given in (\ref{eq:GenEncFnLay}) as
\begin{equation}
\label{eq:GenEncFnRep}
\xbf_j^{(k)}(w) = f_j^{(k)}(\ybf_j^{(k-1)}(w)).
\end{equation}
The received signals in the nodes $j\in R_l(\Omega)$ are
deterministic transformations of the transmitted signals from nodes
$\mathcal{T}_l=\{u: (u,v)\in \mathcal{E}, v\in R_l(\Omega)\}$. As
in the linear model analysis of Section \ref{sec:LayFF}, the
dependence is on all the transmitting signals at distance $l-1$ from
the source, not just the ones in $L_l (\Omega)$.  Since all the receivers in $R_l(\Omega)$ are at distance $l$ from $S$, they form the receivers of the layer $l$.

We now define the following events:
\begin{itemize}
\item $\mathcal{L}_l$: Event that the nodes in $L_l$ can distinguish between $w$ and $w'$, i.e. $\ybf_{L_l}(w) \neq \ybf_{L_l}(w')$,
\item $\mathcal{R}_l$: Event that the nodes in $R_l$  can not distinguish between $w$ and $w'$, i.e. $\ybf_{R_l}(w) = \ybf_{R_l}(w')$.
\end{itemize}

Similar to Appendix \ref{subsec:PfIdeaGen} we can write
\begin{align*}
& \mathcal{P} = \mathbb{P} \{ \mathcal{R}_l,\mathcal{L}_{l-1} ,l=2,\ldots,l_D  \} \\
 &= \prod_{l=2}^{l_D}  \mathbb{P} \{\mathcal{R}_l,\mathcal{L}_{l-1}|\mathcal{R}_j,\mathcal{L}_{j-1}, j=2,\ldots,l-1\} \\
 &\leq \prod_{l=2}^{l_D}  \mathbb{P} \{\mathcal{R}_l|\mathcal{R}_j,\mathcal{L}_{j}, j=2,\ldots,l-1\} = \prod_{l=2}^{l_D}  \mathbb{P} \{\mathcal{R}_l|\mathcal{R}_{l-1},\mathcal{L}_{l-1}\} .
 \end{align*}

Note that  for all the transmitting nodes in $R_{l-1}$ which cannot distinguish between $w,w'$ the transmitted
signal would be the same under both $w$ and $w'$, \emph{i.e.}
\[
\xbf_j(w)=\xbf_j(w'),\,\,\,j\in R_{l-1}.
\]
Therefore, since $(\{\xbf_j(w)\}_{j\in
R_{l-1}},\ybf_{R_l}(w))\in T_{\delta}$, we have that
\[
(\{\xbf_j(w')\}_{j \in R_{l-1}},\ybf_{R_l} (w))\in
T_{\delta}.
\]
Therefore, just as in Appendix \ref{subsec:PfIdeaGen}, we see that 
\begin{align}
\nonumber &   \mathbb{P} \{\mathcal{R}_l|\mathcal{R}_{l-1},\mathcal{L}_{l-1}\}  = \mathbb{P}\{(\xbf_{L_{l-1}}(w'),\xbf_{R_{l-1}}(w),\ybf_{R_l}(w))\in \\ \label{eq:ClusterProb} &
T_{\delta}(x_{L_{l-1}},x_{R_{l-1}},y_{R_l})\}  \stackrel{\cdot}{=}
2^{-TI(x_{L_{l-1}};y_{R_l}|x_{R_{l-1}})}.
\end{align}
Therefore  
\begin{eqnarray}
\label{eq:GenPEP1} \mathcal{P} \leq    \prod_{l=2}^d
2^{-TI(x_{L_{l-1}};y_{R_l}|x_{R_{l-1}})}  = 2^{-T\sum_{l=2}^d H(y_{R_l}|x_{R_{l-1}})}. 
\end{eqnarray}
Due to the Markovian nature of the layered network, $\sum_{l=2}^d H(y_{R_l}|x_{R_{l-1}})=H(y_{\Omega^c}|x_{\Omega^c})$.  From this point the
proof closely follows the steps from
(\ref{eq:RateBndGenEx}) onwards.  Similarly, in a multicast scenario  we declare an error if {\em any} receiver $D\in\mathcal{D}$ makes
an error. Since we have $2^{RT}$ messages, from the union
bound we can drive the error probability to zero if we have
\begin{equation}
\label{eq:LayRateBndGen}
\displaystyle
R < \max_{\prod_{i\in\mathcal{V}} p(x_i)} \min_{D\in\mathcal{D}} \min_{\Omega\in\Lambda_D}
H(y_{\Omega^c}|x_{\Omega^c}).
\end{equation}

We can use an argument similar to Section \ref{sec:TimExpDet} in the
linear deterministic case, to show that the layered proof for the
general deterministic relay network can be extended to arbitrary
(non-layered) deterministic networks. We also had an alternate proof
for this conversion in \cite{ADTAllerton07P2}, which used
submodularity properties of entropy to show the same result.

\section{Proof of Lemma \ref{lem:quantScalar}}
\label{app:proofLemQuantScalar}

Consider the SVD decomposition of $\Hbf$: $\Hbf=\Ubf \Sigma \Vbf^\dag
$, with singular values $\sigma_1,\ldots,\sigma_{\min(m,n)}$. Let us
define $K=\min\{m,n\}$ and
$\tilde{\xbf}_j=[\tilde{x}_{1,j},\cdots,\tilde{x}_{m,j}]$, which is
i.i.d. (over $1\leq j\leq T$) $\mathcal{CN}(0,\Ibf_m)$.

Therefore, if $||\Hbf\tilde{\xbf}_j||_{\infty} \leq \sqrt{2}$, then $ ||\Sigma \Vbf
\tilde{\xbf}_j||_{2} \leq \sqrt{2K}$, which means,
\begin{eqnarray}
\nonumber \prob{||\Hbf\tilde{\xbf}_j||_{\infty} \leq \sqrt{2}} & \leq & \prob{||\Sigma
  \Vbf \tilde{\xbf}_j||_{2} \leq \sqrt{2K}} \\ &=& \prob{||\Sigma
  \tilde{\xbf}_j||_{2} \leq \sqrt{2K}}
\label{eq:normRel}
\end{eqnarray}
where the last step is true since the distribution of $\tilde{\xbf}$ and $\Vbf\tilde{\xbf}$ are
the same.

Now by using (\ref{eq:normRel}),  we get
\begin{align*}
\nonumber & \prob{\forall 1 \leq j \leq T:
    ||\Hbf[\tilde{x}_{1,j},\cdots,\tilde{x}_{m,j}]^t||_{\infty} \leq \sqrt{2}} \\ \nonumber & \leq
  \prob{\forall 1 \leq j \leq T: ||\Sigma
    [\tilde{x}_{1,j},\cdots,\tilde{x}_{m,j}]^t||_{2} \leq \sqrt{2K}} \\ & \leq
  \prob{\forall 1 \leq j \leq T: \sum_{i=1}^{\min\{m,n\}}\sigma_i^2
    |\tilde{x}_{i,j}|^2\leq 2K} \\ &=
  \prod_{j=1}^T\prob{\sum_{i=1}^{\min\{m,n\}}\sigma_i^2 |\tilde{x}_{i,j}|^2
    \leq 2K} \\ &\stackrel{(a)}{\leq} \prod_{j=1}^T e^{-\left (
    \sum_{i=1}^{\min\{m,n\}}\log(1+\frac{1}{2}2\sigma_i^2)-K\right
    )}\\ &= e^{- T\left (
    \sum_{i=1}^{\min\{m,n\}}\log(1+\sigma_i^2)-K\right
    )},
\end{align*}
where $(a)$ follows from the Chernoff bound\footnote{We would like to
acknowledge useful discussions with A. Ozgur on sharpening the proof of this
result. It is also related to the proof technique in \cite{OD10}.}.

Since, $\sum_{i=1}^{\min\{m,n\}}\log(1+\sigma_i^2)=\log
\mathrm{det}(\Ibf+\Hbf\Hbf^*)=I(\xbf;\Hbf\xbf+\zbf)$, for
$\xbf\sim\mathcal{CN}(0,\Ibf_m),\zbf\sim\mathcal{CN}(0,\Ibf_n)$, we
get the desired result.

\section{Proof of Lemma \ref{lem:condMI}}
\label{app:proofLemCondMI}

We first prove the following lemmas.

\begin{lemma} \label{lem:conEntLem}
Consider integer-valued random variables $x$, $r$ and $s$ such that 
\begin{eqnarray*}
x & \bot & r \\
s & \in & \{-L,\ldots,0,\ldots,L\} \\
\prob{|r| \geq k} & \leq & e^{-f(k)}, \quad \hbox{for all $k \in \ZZ^+$}
\end{eqnarray*}
for some integer $L$ and a function $f(.)$. Let \beq \nonumber y=x+r+s \eeq Then
{\footnotesize
\begin{align*} & H(y|x)  \leq  2 \log_2 e \lp \sum_{k=1}^{\infty}f(k)e^{-f(k)} \rp + \frac{2L+1}{2}+N_f \\ & H(x|y)  \leq  \log \lp 2L+1 \rp+ 2 \log_2 e \lp \sum_{k=1}^{\infty}f(k)e^{-f(k)} \rp + \frac{2L+1}{2}+N_f \end{align*}}
where
\beq N_f=\left |  \{n \in \mathcal{Z}^+ | e^{-f(n)} >\frac{1}{2} \}  \right |. \eeq
\end{lemma}
\begin{proof}
By definition we have
\begin{align*} & H(y|x)=H(x+r+s|x)=H(r+s|x) \\
& \leq  H(r+s) = - \sum_{k} \prob{r+s=k} \log \prob{r+s=k}. \end{align*}
Now since  $-p\log p \leq \frac{1}{2}$ for $0 \leq p \leq 1$,  we have
\beq \label{eqn:conEntLemPart1} -\sum_{k=-L}^{L}\prob{r+s=k} \log \prob{r+s=k} \leq \frac{2L+1}{2} . \eeq

For $|k|>L$ we have
\beq \prob{r+s=k} \leq  \prob{|r| \geq |k|-L}  \leq e^{-f(|k|-L)}. \eeq
Since $p \log p$ is decreasing in $p$ for $p<\frac{1}{2}$ we have
\begin{align} 
\nonumber & -\sum_{k=L+1}^{\infty} \prob{r+s=k} \log \prob{r+s=k} \\ \nonumber \quad  & = -\sum_{\substack{k>L \\ k-L \in N_f}} \prob{r+s=k} \log \prob{r+s=k} \\ \nonumber  & \quad \quad -\sum_{\substack{k>L \\ k-L \notin N_f}} \prob{r+s=k} \log \prob{r+s=k} \quad \quad \\ \label{eqn:conEntLemPart2}
& \leq  \frac{N_f}{2}+\sum_{k=L+1}^{\infty}e^{-f(k-L)}f(k-L) \log e
\end{align}
and similarly 
\begin{align} 
\nonumber & -\sum_{k=-\infty}^{-L} \prob{r+s=k} \log \prob{r+s=k} \\ &\nonumber \quad = -\sum_{\substack{k<-L\\ |k|-L \in N_f}} \prob{r+s=k} \log \prob{r+s=k}\\ \nonumber \quad \quad  & -\sum_{\substack{k<-L\\ |k|-L \notin N_f}} \prob{r+s=k} \log \prob{r+s=k} \quad \quad \\ \label{eqn:conEntLemPart3}
& \quad \leq  \frac{N_f}{2}+\sum_{k=L+1}^{\infty}e^{-f(k-L)}f(k-L) \log e.
\end{align}
By combining (\ref{eqn:conEntLemPart1}), (\ref{eqn:conEntLemPart2}) and (\ref{eqn:conEntLemPart3}) we get
\beq H(y|x)  \leq  2 \log_2 e \lp \sum_{k=1}^{\infty}f(k)e^{-f(k)} \rp + \frac{2L+1}{2}+N_f . \eeq
Now we prove the second inequality:
\begin{eqnarray*}
H(x|y) &=& H(x|x+r+s)=  H(x)-I(x;x+r+s) \\
&=& H(x)-H(x+r+s)+H(x+r+s|x) \\
& \leq & H(x)-H(x+r+s|s)+ H(y|x) \\
&=& H(x)-H(x+r|s)+ H(y|x) \\ 
&=& H(x)-H(x+r)+I(x+r;s)+ H(y|x) \\ 
&\leq& H(x)-H(x+r)+H(s)+ H(y|x) \\ 
&\leq& H(x)-H(x+r|r)+\log\lp 2L+1 \rp+ H(y|x)\\
&=& H(x)-H(x)+\log\lp 2L+1 \rp+ H(y|x)\\
&=& \log\lp 2L+1 \rp+ H(y|x).
\end{eqnarray*}
Therefore 
{\small \beq H(x|y)  \leq  \log\lp 2L+1 \rp + 2 \log_2 e \lp \sum_{k=1}^{\infty}f(k)e^{-f(k)} \rp + \frac{2L+1}{2}+N_f.  \eeq}
\end{proof}

\begin{corollary}
\label{cor:condEntropyQ}
Assume $v$ is a continuous complex random variable, then
\begin{eqnarray*} H([v+z]||[v]) & \leq & 12 \\
H([v]||[v+z]) & \leq & 12
\end{eqnarray*}
where $z$ is a $\mathcal{CN}(0,1)$ random variable independent of $v$ and $[.]$ is defined in Definition \ref{defn:Quantization}.
\end{corollary}
\begin{proof} We use lemma \ref{lem:conEntLem} with variables
\begin{eqnarray*}
x&=&[\text{Re}(v)] \\
r &=& [\text{Re}(z)] \\
s &=& [\{\text{Re}(v)\}+\{\text{Re}(z)\}]
\end{eqnarray*}
Then $L=1$ and since
\begin{eqnarray*}
\prob{|[\text{Re}(z)]|\geq k} & \leq & \prob{|[\text{Re}(z)]|-\frac{1}{2} \geq k}  =  2 Q(k-\frac{1}{2}) \\
& \leq & e^{-\frac{(k-\frac{1}{2})^2}{2}}  \end{eqnarray*}
We can use \beq \nonumber f(k)=\frac{(k-\frac{1}{2})^2}{2}.\eeq
Also since 
\beq \nonumber e^{-\frac{(k-\frac{1}{2})^2}{2}} < \frac{1}{2}, \quad \hbox{for $k\geq 2$}\eeq we have $N_f=1$.
Hence
\begin{eqnarray*}
&& \log\lp 2L+1 \rp + 2 \log_2 e \lp \sum_{k=1}^{\infty}f(k)e^{-f(k)} \rp + \frac{2L+1}{2}+N_f   \\ & &= 2 \log_2 e \lp \sum_{k=1}^{\infty}\frac{(k-\frac{1}{2})^2}{2} e^{-\frac{(k-\frac{1}{2})^2}{2}} \rp +2.5+\log_2 3  \\
&&\approx 5.89 < 6.
\end{eqnarray*}
As a result
\begin{eqnarray*} H([\text{Re}(v+z)]||[\text{Re}(v)]) & \leq & 6 \\
H([\text{Re}(v)]||[\text{Re}(v+z)]) & \leq & 6
\end{eqnarray*}
Similarly 
\begin{eqnarray*} H([\text{Im}(v+z)]||[\text{Im}(v)]) & \leq & 6 \\
H([\text{Im}(v)]||[\text{Im}(v+z)]) & \leq & 6
\end{eqnarray*}
Therefore
\begin{eqnarray*} H([v+z]||[v])&\leq & H([\text{Re}(v+z)]||[\text{Re}(v)])  \\ && + H([\text{Im}(v+z)]||[\text{Im}(v)])  \leq  12 \\
H([v]||[v+z]) & \leq & H([\text{Re}(v)]||[\text{Re}(v+z)])\\ && +H([\text{Im}(v)]||[\text{Im}(v+z)]) \leq 12
\end{eqnarray*}

\end{proof}

\begin{align*}
\nonumber &H([\ybf_D]|u,F_\mathcal{V})]  \leq  H([\ybf_{\mathcal{V}}]|w',F_\mathcal{V}) \\& = \sum_{l=2}^{l_D} H([\ybf_{\mathcal{V}_{l}}]|[\ybf_{\mathcal{V}_{l-1}},F_\mathcal{V}]) \\
\nonumber & = \sum_{l=2}^{l_D} H([\ybf_{\mathcal{V}_{l}}]|\xbf_{\mathcal{V}_{l-1}},F_\mathcal{V}) \\
\nonumber & = \sum_{l=2}^{l_D} H([\text{Re}( \ybf_{\mathcal{V}_{l}})]|\xbf_{\mathcal{V}_{l-1}},F_\mathcal{V})+H([\text{Im}( \ybf_{\mathcal{V}_{l}})]|\xbf_{\mathcal{V}_{l-1}},F_\mathcal{V}) \\
& \stackrel{\text{Corollary \ref{cor:condEntropyQ}}}{\leq}  \sum_{l=2}^{l_D} 12T |\mathcal{V}_{l}| \\
&= 12T |\mathcal{V}|.
\end{align*}

\section{Proof of Lemma \ref{lem:BeamForming}}
\label{app:proofLemBeamForming}

First note that $\overline{C}_{\Omega}$ is the capacity of the MIMO
channel that the cut $\Omega$ creates. Therefore intuitively we want
to prove that the gap between the capacity of a MIMO channel and its
capacity when it is restricted to have equal power allocation at the
transmitting antennas is upper bounded by a constant. Therefore
without loss of generality we just focus an $n \times m$ MIMO channel, with $K=\min\{m,n\}$,
\begin{equation} 
\ybf^n=\Gbf \xbf^m+ \zbf^n 
\end{equation} 
with average transmit power per antenna equal to $P$ and i.i.d complex
normal noise. We know that the capacity of this MIMO channel is
achieved with water filling, and
\begin{equation}
C=C_{wf}=\sum_{i=1}^K \log (1+\tilde{Q}_{ii} \lambda_i)
\end{equation}
where $\lambda_i$'s are the singular values of $\Gbf$ and
$\tilde{Q}_{ii}$ is given by water filling solution satisfying
\begin{equation} \sum_{i=1}^K \tilde{Q}_{ii}=mP. \end{equation}

Now with equal power allocation we have
\begin{equation} C_{ep}= \sum_{i=1}^K \log (1+P \lambda_i).  \end{equation}
Now note that

\begin{align*}
&C_{wf}-C_{ep}=   \log \lp \frac{\prod_{i=1}^K  (1+\tilde{Q}_{ii} \lambda_i)}{\prod_{i=1}^K  (1+P \lambda_i)}\rp  \\
& \leq   \log \lp \frac{\prod_{i=1}^K  (1+\tilde{Q}_{ii} \lambda_i)}{\prod_{i=1}^K  \max (1,P \lambda_i)}\rp \\
& =   \log \lp \prod_{i=1}^K \frac{1+\tilde{Q}_{ii} \lambda_i}{ \max (1,P \lambda_i)}\rp \\
& =   \log \lp \prod_{i=1}^K \lp \frac{1}{ \max (1,P \lambda_i)}+ \frac{\tilde{Q}_{ii} \lambda_i}{ \max (1,P \lambda_i)} \rp \rp \\
& \leq    \log \lp \prod_{i=1}^K \lp 1+ \frac{\tilde{Q}_{ii} \lambda_i}{P \lambda_i} \rp \rp =  \log    \lp \prod_{i=1}^K \lp 1+ \frac{\tilde{Q}_{ii} }{P} \rp \rp .
\end{align*}
Now note that 
\begin{equation} \sum_{i=1}^K (1+ \frac{\tilde{Q}_{ii} }{P})=K+m \end{equation} and therefore by arithmetic mean-geometric mean inequality we have
\begin{equation} \prod_{i=1}^K \lp 1+ \frac{\tilde{Q}_{ii} }{P} \rp \leq \lp \frac{ \sum_{i=1}^K (1+ \frac{\tilde{Q}_{ii} }{P})}{K} \rp ^K=(1+\frac{m}{K})^K  \end{equation}
and hence
{\small
\begin{eqnarray}  \label{eq:wfVep} 
C_{wf}-C_{ep} &\leq& K\log(1+\frac{m}{K}) \\ &=& 
K\log\left(\frac{m}{K}\right) + K\log\left(1+\frac{K}{m}\right)
\\ &\leq&  \underbrace{K\log\left(\frac{m}{K}\right)}_{\log\left(\frac{m}{K}\right)^K} + K 
\stackrel{(a)}{\leq}
\frac{m}{e} + K,
\end{eqnarray}
}
where $K = \min\{m,n\}$, and $(a)$ follows because
$\max_K\left(\frac{m}{K}\right)^K\leq e^{m/e}$ and we also take
natural logarithms.  Therefore the loss from restricting ourselves to
use equal transmit powers at each antenna of an $m \times n$ MIMO
channel is at most $\frac{m}{e} + \min\{m,n\}$ bits.

Now, let us apply \eqref{eq:wfVep} to prove Lemma
\ref{lem:BeamForming}. Note that the cut-set upper bound of
\eqref{eq:CutSetRef} when applied to the Gaussian network yields,
{\small
\begin{align}
\displaystyle
\label{eq:GaussCutSetRef0} 
&\overline{C}=
\max_{\substack{p(\{\Xbf_i\})\\\Qbf:\Qbf_{ii}\leq P,\forall i}} \min_{\Omega\in\Lambda_D}
\left \{ h(\Ybf_{\Omega^c}|\Xbf_{\Omega^c}) - 
h(\Ybf_{\Omega^c}|\Xbf_{\Omega^c},\Xbf_{\Omega}) \right \}\\ \nonumber
&=
\max_{\substack{p(\{\Xbf_i\}),\\\Qbf:\Qbf_{ii}\leq P,\forall i}} \min_{\Omega\in\Lambda_D}
\left \{ h(\Ybf_{\Omega^c}-\breve{\Gbf}_{\Omega^c,\Omega^c}\Xbf_{\Omega^c}|\Xbf_{\Omega^c}) - 
h(\Zbf_{\Omega^c}) \right \}\\ 
\label{eq:GaussCutSetRef} 
&\leq
\max_{\Qbf:\Qbf_{ii}\leq P,\forall i} \min_{\Omega\in\Lambda_D}
\log |\Ibf+\Gbf_{\Omega}\Qbf\Gbf_{\Omega}^*|,
\end{align} }
where $\Gbf_{\Omega}$ represents the network transfer matrix from
transmitting set $\Omega$ to receiving set $\Omega^c$ and 
$\breve{\Gbf}_{\Omega^c,\Omega^c}$ represents the transfer matrix
from set $\Omega^c$ to $\Omega^c$. The maximization
in \eqref{eq:CutSetRef} can be restricted to jointly Gaussian inputs represented
by covariance matrix $\Qbf$ with individual power constraints.
Now, clearly these constraints can be relaxed to the sum-power constraints yielding,
\begin{eqnarray}
\displaystyle
\nonumber
\overline{C} & \leq&
\max_{\Qbf:\Qbf_{ii}\leq P,\forall i} \min_{\Omega\in\Lambda_D}
\log |\Ibf+\Gbf_{\Omega}\Qbf\Gbf_{\Omega}^*|\\
\nonumber  &\leq &
 \min_{\Omega\in\Lambda_D}\max_{\Qbf:\mathrm{tr}(\Qbf)\leq |\Omega|P}
\log |\Ibf+\Gbf_{\Omega}\Qbf\Gbf_{\Omega}^*| \\ \label{eq:GaussCutSetRefSumPwr} &=&  \min_{\Omega\in\Lambda_D} \overline{C}_{\Omega}.
\end{eqnarray}

Now, let us define $\overline{C}_{\Omega}^{iid}$, to be the cut value
for $i.i.d.$ Gaussian inputs, {\em i.e.,} $\Qbf=\Ibf$. More precisely,
from Definition \ref{def:cutSetiid} we have for $p(\{\Xbf_i\})=\prod_i
p(\Xbf_i)$, and $\Xbf_i\sim \mathcal{CN}(0,1)$, {\em i.e.,} i.i.d.,
unit variance Gaussian variables, the cut value evaluated as
\begin{eqnarray}
\label{eq:IIDcutval}
\overline{C}_{i.i.d.} &=&
\min_{\Omega\in\Lambda_D}
\left \{ h(\Ybf_{\Omega^c}|\Xbf_{\Omega^c}) - 
h(\Ybf_{\Omega^c}|\Xbf_{\Omega^c},\Xbf_{\Omega}) \right \} \\ \nonumber
&=&
\min_{\Omega\in\Lambda_D}
\left \{ h(\Ybf_{\Omega^c}-\breve{\Gbf}_{\Omega^c,\Omega^c}\Xbf_{\Omega^c}|\Xbf_{\Omega^c}) - 
h(\Zbf_{\Omega^c}) \right \}\\ \nonumber
&\stackrel{(a)}{=}&\min_{\Omega} 
\underbrace{\log |\Ibf+P\Gbf_{\Omega}\Gbf_{\Omega}^*|}_{\overline{C}_{\Omega}^{iid}}= 
\min_{\Omega} \overline{C}_{\Omega}^{iid},
\end{eqnarray}
where $(a)$ follows because 
$\Ybf_{\Omega^c}-\breve{\Gbf}_{\Omega^c,\Omega^c}\Xbf_{\Omega^c}=\Gbf_{\Omega}\Xbf_{\Omega}+\Zbf_{\Omega^c}$ is independent of $\Xbf_{\Omega^c}$ due to i.i.d. choice of input
distrbutions.

By using \eqref{eq:wfVep}, we get,
\begin{eqnarray}
\label{eq:ConnectMIMOcut}
\overline{C}_{\Omega}-\overline{C}_{\Omega}^{iid}&\leq& 
\frac{|\Omega|}{e} + \min\{|\Omega|,|\Omega^c|\}\leq 2|\mathcal{V}|, \,\,\, \forall \Omega,\\ \nonumber
\mbox{or    }\overline{C}_{\Omega}\ &\leq& \overline{C}_{\Omega}^{iid} + 2|\mathcal{V}|, 
\,\,\, \forall \Omega.
\end{eqnarray}
Since 
$\displaystyle\min_{\Omega} \overline{C}_{\Omega} \leq \min_{\Omega} \overline{C}_{\Omega}^{iid} + 2|\mathcal{V}|$,
we get the claimed result in Lemma \ref{lem:BeamForming}, for the scalar case.

For the case with multiple antennas, we see that for any cut $\Omega$,
the number of degrees of freedom is
$\min\{\sum_{i\in\Omega}M_i,\sum_{i\in\Omega^c}N_i\}$.  Note that, 
$\max_{\Omega}\min\{\sum_{i\in\Omega}M_i,\sum_{i\in\Omega^c}N_i\}\leq
\sum_{i=1}^{|\mathcal{V}|}M_i$ and
$\max_{\Omega}\min\{\sum_{i\in\Omega}M_i,\sum_{i\in\Omega^c}N_i\}\leq
\sum_{i=1}^{|\mathcal{V}|}N_i$ and hence
$\max_{\Omega}\min\{\sum_{i\in\Omega}M_i,\sum_{i\in\Omega^c}N_i\}\leq\min\{\sum_{i=1}^{|\mathcal{V}|}M_i,\sum_{i=1}^{|\mathcal{V}|}N_i\}$
yielding
\begin{equation}
\label{eq:TrivMultAntUB}
\min\{\sum_{i\in\Omega}M_i,\sum_{i\in\Omega^c}N_i\}\leq
\min\{\sum_{i=1}^{|\mathcal{V}|}M_i,\sum_{i=1}^{|\mathcal{V}|}N_i\},\,\,\ \forall \Omega.
\end{equation}
For a trivial upper bound to use in an argument analogous to
\eqref{eq:ConnectMIMOcut}, we can use \eqref{eq:TrivMultAntUB} to see that 
\begin{equation}
\label{eq:TrivMultAntCutBnd}
\overline{C}_{\Omega}\ \leq \overline{C}_{\Omega}^{iid} + 
2\sum_{i=1}^{|\mathcal{V}|}M_i.
\end{equation}

\section{Proof of Lemma \ref{lem:connGaussTrunMIMO}}
\label{app:lemconnGaussTrunMIMO}
We first prove the following two lemmas:


\begin{lemma} \label{lem:connQGaussTrunMIMO}
Let $G$ be the channel gains matrix of a $m \times n$ MIMO system. Assume that there is an average power constraint equal to one at each node. Then for any input distribution $P_{\xbf}$,
\beq  |I(\xbf;[G \xbf+\zbf])- I(\xbf;[G \xbf])| \leq 12n \eeq 
where $\zbf=[z_1, \ldots, z_n]$ is a vector of $n$ i.i.d.  $\mathcal{CN}(0,1)$  random variables.
\end{lemma}

\begin{lemma} \label{lem:connGaussQGaussMIMO}
Let $G$ be the channel gains matrix of a $m \times n$ MIMO system. Assume that there is an average power constraint equal to one at each node. Then for any input distribution $P_{\xbf}$,
\beq  |I(\xbf;G \xbf+\zbf)- I(\xbf;[G \xbf+\zbf])| \leq 7n \eeq 
where $\zbf=[z_1, \ldots, z_n]$ is a vector of $n$ i.i.d. $\mathcal{CN}(0,1)$ random variables.
\end{lemma}

Note that Lemma  \ref{lem:connGaussTrunMIMO} is just a corollary of these two lemmas, which are proved next.
\begin{proof} \textbf{(proof of Lemma \ref{lem:connQGaussTrunMIMO})} \\
First note that
{\small
\begin{eqnarray}
\nonumber I(\xbf;[G \xbf]) & \leq & I(\xbf;[G \xbf+\zbf])+I(\xbf;[G \xbf]|[G \xbf+\zbf]) \\
\nonumber &=& I(\xbf;[G \xbf+\zbf])+ H([G \xbf]|[G \xbf+\zbf])\\
\label{eq:lemmaQGauss1} & \stackrel{\text{(Corollary \ref{cor:condEntropyQ})}}{ \leq} & I(\xbf;[G \xbf+\zbf])+ 12n.
\end{eqnarray}
\begin{eqnarray}
\nonumber I(\xbf;[G \xbf+\zbf]) & \leq & I(\xbf;[G \xbf])+I(\xbf;[G \xbf+\zbf]|[G \xbf]) \\
\nonumber &\leq & I(\xbf;[G \xbf])+ H([G \xbf+\zbf]|[G \xbf])\\
\label{eq:lemmaQGauss2} &  \stackrel{\text{(Corollary \ref{cor:condEntropyQ})}}{ \leq} & I(\xbf;[G \xbf])+ 12n.
\end{eqnarray}}
Now from equations (\ref{eq:lemmaQGauss1}) and (\ref{eq:lemmaQGauss2}) we have
\beq  |I(\xbf;[G \xbf+\zbf])- I(\xbf;[G \xbf])| \leq 12n.  \eeq 
\end{proof}

\begin{proof} \textbf{(proof of Lemma \ref{lem:connGaussQGaussMIMO})} \\
 Define the following random variables:
\begin{eqnarray*}
\ybf &=& G \xbf+\zbf \\
\hat{\ybf}&=&[G \xbf+\zbf] \\
\tilde{\ybf} &=& \hat{\ybf} + \ubf
\end{eqnarray*}
where $\ubf=[u_1, \ldots, u_n]$ is a vector of $n$ i.i.d. complex variables with distribution $\text{uniform}[0,1]$ on both real and complex components, independent of $\xbf$ and $\zbf$.

By the data processing inequality we have $I(\xbf;\ybf) \geq  I(\xbf;\hat{\ybf}) \geq  I(\xbf;\tilde{\ybf})$. Now, note that
\begin{align}
\nonumber & I(\xbf;\ybf)-I(\xbf;\tilde{\ybf}) = h(\ybf)-h(\tilde{\ybf})+h(\tilde{\ybf}|\xbf)-h(\ybf|\xbf)\\
\nonumber &= h(\ybf)-h(\tilde{\ybf})+h(\tilde{\ybf}|\xbf)-n \log \lp  \pi e\rp \\
\nonumber &= h(\ybf|\tilde{\ybf})-h(\tilde{\ybf}|\ybf)+h(\tilde{\ybf}|\xbf)-n \log \lp  \pi e \rp \\
\nonumber &= h(\ybf|\tilde{\ybf})-h(\ubf)+h(\tilde{\ybf}|\xbf)-n \log \lp  \pi e \rp \\
\label{eq:connGQG1} &= h(\ybf|\tilde{\ybf})+h(\tilde{\ybf}|\xbf)-n \log \lp  \pi e \rp
\end{align}
where the last step is true since $h(\ubf)=nh(u_1)=2n \log 1=0$.
Now 
\beq |\text{Re}(y)-\text{Re}(\tilde{y})| \leq \max_{x \in \CC} \lp |[\text{Re}(x)]-\text{Re}(x)| \rp + \max |\text{Re}(u)| =\frac{3}{2}  \eeq
and similarly 
\beq |\text{Im}(y)-\text{Im}(\tilde{y})| \leq \max_{x \in \CC} \lp |[\text{Im}(x)]-\text{Im}(x)| \rp + \max |\text{Im}(u)| =\frac{3}{2}  \eeq
Therefore
{\small
\begin{align}
\nonumber & h(\ybf|\tilde{\ybf}) = h(\ybf-\tilde{\ybf}|\tilde{\ybf})\\
\nonumber & \leq  n \log \lp 2 \pi e \sqrt{\max \lp|\text{Re}(y)-\text{Re}(\tilde{y})| \rp \max \lp|\text{Im}(y)-\text{Im}(\tilde{y})| \rp} \rp\\
\label{eq:connGQG2} &= n \log 3 \pi e.
\end{align}}
For the second term, lets look at the $i$-th element of $\tilde{y}$
\begin{eqnarray}  \tilde{y}_i = [\gbf_i \xbf+z_i]+u_i= \gbf_i \xbf+z_i + \delta(\gbf_i \xbf + z_i)+u_i   \end{eqnarray}
where $\tilde{y}_i$ is the $i$-th component of $\tilde{y}$, $\gbf_i$ is the $i$-th row of $G$, and $\delta(x)=x-[x] $. Clearly $|\text{Re}(\delta(x))|, |\text{Im}(\delta(x))| \leq \frac{1}{2}$  for all $x \in \CC$. Therefore given $x$ the variance of $\tilde{y}_i$ is bounded by
\begin{align}
\nonumber & \var{\text{Re}(\tilde{y}_i)|\xbf} = \var{ \text{Re}(z_i) + \text{Re}(\delta(\gbf_i \xbf+z_i))+ \text{Re}(u_i) }  \\
\nonumber & \leq  \var{\text{Re}(z_i)}+\var{\text{Re}( \delta(\gbf_I \xbf+z_i))|\xbf} +\\
\nonumber & \quad 2\cov{\text{Re}(z_i),\text{Re}(\delta(\gbf_i \xbf +z_i))|\xbf} +\var{\text{Re}(u)}  \\
\nonumber & \leq  \var{\text{Re}(z_i)}+|\max \text{Re}(\delta(.)) |^2 + \\
\nonumber & \quad 2\sqrt{\var{\text{Re}(z_i)}\times |\max \text{Re}(\delta(.))|} +\var{\text{Re}(u_i}) \\
\label{eq:appNum1}&= \frac{1}{2}+\frac{1}{4}+1+\frac{1}{12} =\frac{11}{6}.
\end{align}
Similarly 
\beq \var{\text{Im}(\tilde{y}_i)|\xbf}  \leq \frac{11}{6}. \eeq
Therefore
\begin{eqnarray}
\nonumber h(\tilde{\ybf}|\xbf) &\leq&  \sum_{i=1}^n h(\tilde{y}_i|\xbf) \leq  \sum_{i=1}^n \log 2\pi e \sqrt{|K_{\tilde{y}_i|X}|}  \\ \label{eq:connGQG3} &\stackrel{(\ref{eq:appNum1})} {\leq} & n \log  \frac{11}{3} \pi e .
\end{eqnarray}

Now from Equations (\ref{eq:connGQG1}), (\ref{eq:connGQG2}) and (\ref{eq:connGQG3}) we have
\begin{eqnarray*}
I(\xbf;\ybf)-I(\xbf;\tilde{\ybf}) &\leq &h(\ybf|\tilde{\ybf})+h(\tilde{\ybf}|\xbf)-\frac{n}{2} \log \lp 2 \pi e\rp \\
& \leq & n \log 11 \pi e \approx 6.55n <7n.
\end{eqnarray*}
\end{proof}

\section*{Acknowledgements}
We would like to thank Anant
Sahai for his insightful comments on an earlier draft of this work. In
particular, they motivated the simpler proof of the approximation
theorem presented in this manuscript for Gaussian relay networks.  We
would also like to thank several others for stimulating discussions on
the topic of this paper including C. Fragouli, S. Mohajer, A. Ozgur and
R. Yeung.


%

\begin{biography}[{\includegraphics[width=1in,height=1.25in,clip,keepaspectratio]{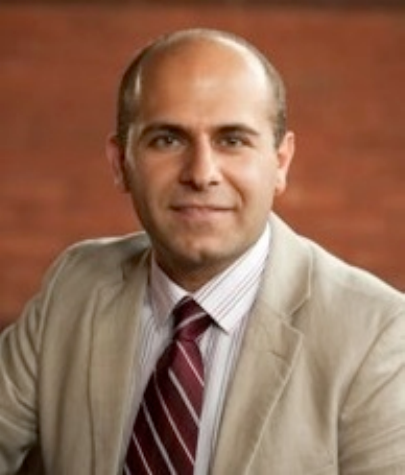}}]{A. Salman Avestimehr} is currently an assistant Professor at the School of Electrical and Computer Engineering at Cornell University, where he has co-founded the Foundations of Information Engineering (FoIE) center. He received his Ph.D. in 2008 and M.S. degree in 2005 in Electrical Engineering and Computer Science, both from the University of California, Berkeley. Prior to that, he obtained his B.S. in Electrical Engineering from Sharif University of Technology in 2003. He was also a postdoctoral scholar at the Center for the Mathematics of Information (CMI) at Caltech in 2008. He has received a number of awards including the NSF CAREER award (2010), the David J. Sakrison Memorial Prize from the U.C. Berkeley EECS Department (2008), and the Vodafone U.S. Foundation Fellows Initiative Research Merit Award (2005).  His research interests include information theory, communications, and networking.
\end{biography}

\begin{biography}[{\includegraphics[width=1in,height=1.25in,clip,keepaspectratio]{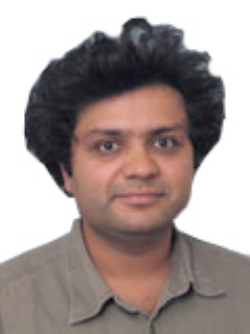}}]{Suhas N. Diggavi} received the B. Tech. degree in electrical
engineering from the Indian Institute of Technology, Delhi, India, and
the Ph.D. degree in electrical engineering from Stanford University,
Stanford, CA, in 1998.

After completing his Ph.D., he was a Principal Member Technical Staff
in the Information Sciences Center, AT\&T Shannon Laboratories, Florham
Park, NJ. After that he was on the faculty at the School of Computer
and Communication Sciences, EPFL, where he directed the Laboratory for
Information and Communication Systems (LICOS).  He is currently a
Professor, in the Department of Electrical Engineering, at the
University of California, Los Angeles.  His research interests include
wireless communications networks, information theory, network data
compression and network algorithms.

He is a recipient of the 2006 IEEE Donald Fink prize paper award, 2005
IEEE Vehicular Technology Conference best paper award and the Okawa
foundation research award.  He is currently an editor for ACM/IEEE
Transactions on Networking and IEEE Transactions on Information
Theory. He has 8 issued patents.
\end{biography}

\begin{biography}[{\includegraphics[width=1in,height=1.25in,clip,keepaspectratio]{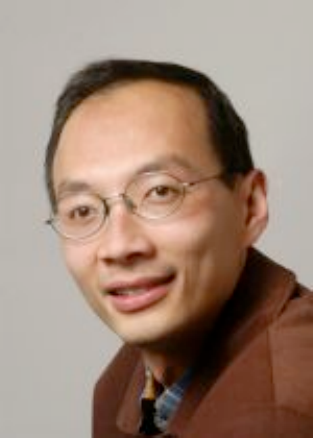}}]{David N. C. Tse} received the B.A.Sc. degree in systems design engineering from University of Waterloo, Canada in 1989, and the M.S. and Ph.D. degrees in electrical engineering from Massachusetts Institute of Technology in 1991 and 1994 respectively. From 1994 to 1995, he was a postdoctoral member of technical staff at A.T. \& T. Bell Laboratories. Since 1995, he has been at the Department of Electrical Engineering and Computer Sciences in the University of California at Berkeley, where he is currently a Professor. He received a 1967 NSERC 4-year graduate fellowship from the government of Canada in 1989, a NSF CAREER award in 1998, the Best Paper Awards at the Infocom 1998 and Infocom 2001 conferences, the Erlang Prize in 2000 from the INFORMS Applied Probability Society, the IEEE Communications and Information Theory Society Joint Paper Award in 2001, the Information Theory Society Paper Award in 2003, and the 2009 Frederick Emmons Terman Award from the American Society for Engineering Education. He has given plenary talks at international conferences such as ICASSP in 2006, MobiCom in 2007, CISS in 2008, and ISIT in 2009. He was the Technical Program co-chair of the International Symposium on Information Theory in 2004, and was an Associate Editor of the IEEE Transactions on Information Theory from 2001 to 2003. He is a coauthor, with Pramod Viswanath, of the text ``Fundamentals of Wireless Communication'', which has been used in over 60 institutions around the world.\end{biography}

\end{document}